\RequirePackage{fix-cm}

\documentclass[twocolumn]{svjour3}
\usepackage{amsmath}
\usepackage{txfonts}
\usepackage{amssymb}
\usepackage{times}
\usepackage{graphicx}
\usepackage{epsfig,tabularx,amssymb,amsmath,subfigure,multirow}
\usepackage[linesnumbered,ruled,noend]{algorithm2e}
\usepackage[noend]{algorithmic}
\usepackage{multirow}
\usepackage{graphicx,floatrow}
\usepackage{listings}
\usepackage{threeparttable}
\usepackage{tikz}
\usepackage[T1]{fontenc}
\usepackage{pgfplots}

\def\ojoin{\setbox0=\hbox{$\bowtie$}%
  \rule[-.02ex]{.25em}{.4pt}\llap{\rule[\ht0]{.25em}{.4pt}}}
\def\leftouterjoin{\mathbin{\ojoin\mkern-5.8mu\bowtie}}

\usetikzlibrary{patterns}
\newcommand{\nop}[1]{}
\newcommand{\tabincell}[2]{\begin{tabular}{@{}#1@{}}#2\end{tabular}}

\smartqed

\newlength{\arrayrulewidthOriginal}

\newcommand{\overbar}[1]{\mkern 1.5mu\overline{\mkern-1.5mu#1\mkern-1.5mu}\mkern 1.5mu}

\newcolumntype{R}[1]{%
>{\raggedleft\arraybackslash\hspace{0pt}}p{#1}}%

\newcolumntype{C}[1]{%
>{\centering\arraybackslash\hspace{0pt}}p{#1}}%

\usetikzlibrary{shapes,shapes.callouts,,decorations.text,shapes.misc,decorations.pathmorphing,shapes.geometric,backgrounds,fit,positioning,arrows}

\tikzstyle{redarrow}=[->, >=triangle 90,dashed, red]
\tikzstyle{smredarrow}=[->, >=triangle 45,dashed, red]
\tikzstyle{redthickarrow}=[->, >=triangle 90,thick, red]
\tikzstyle{bluearrow}=[->, >=triangle 90, blue,dashed]

\begin{document}

\title{Processing SPARQL Queries Over Distributed RDF Graphs
%Grants or other notes
}


\author{Peng Peng   \and Lei Zou      \and  M. Tamer {\"O}zsu \and
        Lei Chen \and
        Dongyan Zhao %etc.
}


\small{
\institute{Peng Peng, Lei Zou, Dongyan Zhao \at
              Institute of Computer Science and Technology, \\Peking University,
              Beijing, China \\
              Tel.: +86-10-82529643\\
              \email{$\{$pku09pp,zoulei,zhaody$\}$@pku.edu.cn}           %  \\
              \and
               M. Tamer {\"O}zsu \at
         David R. Cheriton School of Computer Science \\ University of Waterloo,
         Waterloo, Canada \\
          Tel.: +1-519-888-4043 \\
              \email{Tamer.Ozsu@uwaterloo.ca}
               \and
            Lei Chen \at
         Department of Computer Science and Engineering, \\
         Hong Kong University of Science and Technology, \\
         Hong Kong, China \\
          Tel.: +852-23586980 \\
              \email{leichen@cse.ust.hk}
}
}


\maketitle

\nop{}

\normalsize
\begin{abstract}
We propose techniques for processing SPARQL queries over a large RDF graph in a distributed environment. We adopt a ``partial evaluation and assembly'' framework. Answering a SPARQL query $Q$ is equivalent to finding subgraph matches of the query graph $Q$ over RDF graph $G$. Based on properties of subgraph matching over a distributed graph, we introduce \emph{local partial match} as partial answers in each fragment of RDF graph $G$. For assembly, we propose two methods: centralized and distributed assembly.  We analyze our algorithms from both  theoretically and experimentally. Extensive experiments over both real and benchmark RDF repositories of billions of triples confirm that our method is superior to the state-of-the-art methods in both the system's performance and scalability.
\end{abstract}

%!TEX root =  distributedgStore.tex

\section{Introduction}\label{sec:introduction}
The semantic web data model, called the ``Resource Description Framework'', or RDF, represents data as a collection of triples of the form $\langle$subject, property, object$\rangle$. A triple can be naturally seen as a pair of entities connected by a named relationship or an entity associated with a named attribute value. Hence, an RDF dataset can be represented as a graph where subjects and objects are vertices, and triples are edges with property names as edge labels. With the increasing amount of RDF data published on the Web, system performance and scalability issues have become increasingly pressing. For example, LOD (Linking Open Data) project builds a RDF data cloud by linking more than 3000 datasets, which currently have more than 84 billion triples \footnote{The statistic is reported in {http://stats.lod2.eu/} }. The recent work \cite{SBP14} shows that the number of data sources has doubled within three years (2011-2014). Obviously, the computational and storage requirements coupled with rapidly growing datasets have stressed the limits of single machine processing.

There have been a number of recent efforts in distributed evaluation of SPARQL queries over large RDF datasets \cite{Hartig:2014aa}. We broadly classify these solutions into three categories: \emph{cloud-based}, \emph{partition-based}, and \emph{federated} approaches. These are discussed in detail in Section \ref{sec:relatedwork}; the highlights are as follows.

Cloud-based approaches  (e.g., \cite{DBLP:JenaHBase,DBLP:SHARD,TKDE11:HadoopRDF,DBLP:PredicateJoin,ICDE13:EAGRE,WWW12:H2RDF,SIGMOD14:H2RDF})  maintain a large RDF graph using existing cloud computing platforms, such as Hadoop (http://hadoop.apache.org) or Cassandra (http://cassandra.apache.org), and employ triple pattern-based join processing most commonly using MapReduce.

\nop{
First, the intermediate results of triple patterns are found, then, a left-deep join of these intermediate results is executed to find SPARQL answers. The process is similar to relational database query processing. However, an interesting difference is that answering a query in relational systems usually involves  fewer joins between different datasets \cite{DBLP:conf/rweb/HoseSTW11} than is the case for RDF systems where some of the joins may even be self-joins. This makes it difficult to directly employ relational techniques  in RDF processing systems. These approaches benefit from the high scalability and fault-tolerance offered by cloud platforms, but may suffer a performance penalty due to the difficulties of adapting MapReduce to graph computation. For a comprehensive survey of RDF data management in cloud environments, we refer the reader to  \cite{VLDBJ14:RDFCloudSurvey}.}

Partition-based approaches \cite{VLDB11:GraphPartition,DBLP:WARP,DBLP:Partout,DBLP:journals/pvldb/LeeL13,DBLP:conf/IEEEcloud/LeeLTZZ13,SIGMOD2014:TriAD} divide the RDF graph $G$ into a set of subgraphs (fragments) $\{F_{i}\}$, and decompose the SPARQL query $Q$ into subqueries $\{Q_{i}\}$. These subqueries are then executed over the partitioned data using techniques similar to relational distributed databases.

%\cite{OV11}.

\nop{
During RDF data partitioning, most existing approaches \cite{VLDB11:GraphPartition,ICDE13:EAGRE,DBLP:conf/IEEEcloud/LeeLTZZ13,DBLP:journals/pvldb/LeeL13,DBLP:WARP} partition the RDF graph into multiple vertex-disjoint fragments. Each fragment $F_{i}$ is allocated to a site. For query evaluation, similar to distributed relational query processing \cite{OV11}, the query is transformed into multiple subqueries on fragments. Then, each subquery $Q_{j}$ is executed at relevant fragments, and the subquery results are then aggregated. However, there is an important difference between distributed relational query processing and partition-based distributed SPARQL processing. In distributed relational processing, query decomposition is used to achieve inter-query parallelism \emph{and} to prune execution on some of the fragments since they would not contribute to the results. Ppruning is not possible in partition-based distributed SPARQL processing, as we discuss in Section \ref{sec:relatedwork}. Consequently, the size of intermediate results that are generated at each site tends to be high and this negatively impacts performance. Furthermore, each of the proposed approaches in this class depends on its own data partitioning algorithm, while partition-agnostic approaches are generally more desirable.}

Federated SPARQL processing systems  \cite{DBLP:DARQ,WWW2010:QTree,DBLP:SPLENDID,DBLP:HiBISCuS,DBLP:TopFed} evaluate queries over multiple SPARQL endpoints. These systems typically target LOD and follow a query processing over data integration approach.These systems operate in a very different environment we are targeting, since we focus on exploiting distributed execution for speed-up and scalability.

%\vspace{-0.1in}
%\subsection{Our Solution}

In this paper we propose an alternative strategy that is based on only partitioning the data graph but not decomposing the query. Our approach is based on the ``partial evaluation and assembly'' framework \cite{DBLP:journals/csur/Jones96}. An RDF graph is partitioned using some graph partitioning algorithm such as METIS \cite{DBLP:metis} into vertex-disjoint fragments  (edges that cross fragments are replicated in source and target fragments).  Each site receives the full SPARQL query $Q$ and executes it on the local RDF graph fragment providing data parallel computation. To the best of our knowledge, this is the first work that adopts the partial evaluation and assembly strategy to evaluate SPARQL queries over a distributed RDF data store. The most important advantage of this approach is that the number of involved vertices and edges in the intermediate results are minimized, which is proven theoretically (see Proposition \ref{proposition:minimalISNum} in Section \ref{sec:partialcomputing}).

The basic idea of the partial evaluation strategy is the following: given a function $f(s,d)$, where $s$ is the known  input and $d$ is the yet unavailable input, the part of $f$'s computation that depends only on $s$ generates a partial answer. In our setting, each site $S_i$ treats fragment $F_i$ as the known input in the partial evaluation stage; the unavailable input is the rest of the graph ($\overbar{G} = G \setminus F_{i}$). The partial evaluation technique has been used in compiler optimization \cite{DBLP:journals/csur/Jones96}, and querying XML trees \cite{SIGMOD07:XPath}. Within the context of graph processing, the technique has been used to evaluate reachability queries  \cite{fan:reachability}, and graph simulation \cite{DBLP:conf/www/MaCHW12,DBLP:journals/pvldb/FanWWD14} over graphs. However, SPARQL query semantics is different than these --- SPARQL is based on graph homomorphism \cite{DBLP:journals/tods/PerezAG09} --- and pose additional challenges.  Graph simulation defines a \emph{relation} between vertices in the query graph $Q$ (i.e. $V(Q)$) and that in the data graph $G$ (i.e., $V(G)$), but,  graph homomorphism is a \emph{function} (not a relation) between $V(Q)$ and $V(G)$ \cite{DBLP:journals/pvldb/FanWWD14}. Thus, the solutions proposed for graph simulation \cite{DBLP:journals/pvldb/FanWWD14} and graph pattern matching \cite{DBLP:conf/www/MaCHW12} cannot be applied to the problem studied in this paper.

\begin{figure}
\begin{center}
    \includegraphics[scale=0.17]{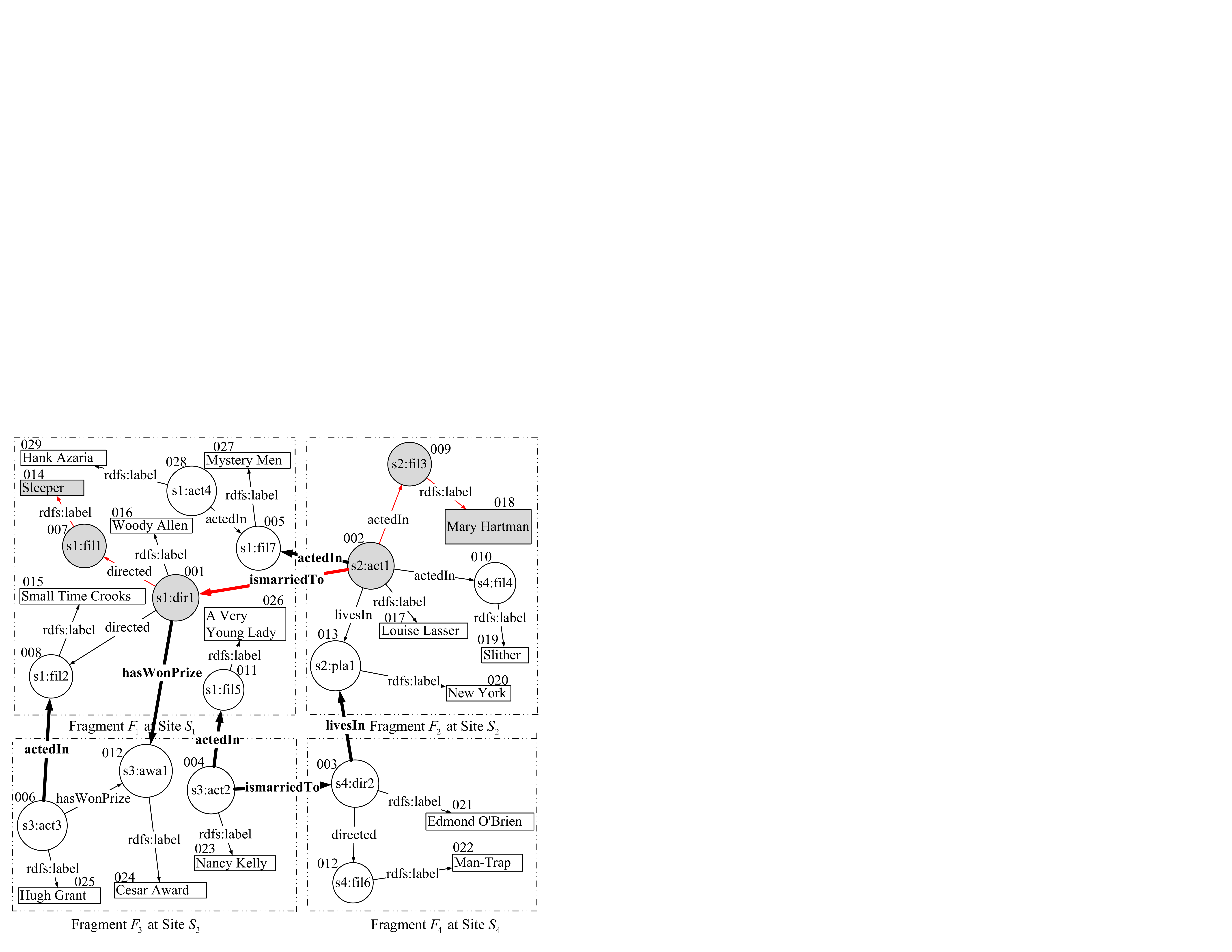}
   \caption{A Distributed RDF Graph}
   \label{fig:lodfgraph}
\end{center}
\end{figure}

Because of interconnections between graph fragments, application of graph homomorphism over graphs requires special care. For example, consider a distributed RDF graph in Figure \ref{fig:lodfgraph}. Each entity in RDF is represented by a URI (uniform resource identifier), the prefix of which always denotes the location of the dataset. For example, ``s1:dir1'' has the prefix ``s1'', meaning that the entity is located at site $s1$. Here, the prefix is just for simplifying presentation, not a general assumption made by the approach. There are \emph{crossing links} between two datasets identified in bold font. For example, ``$\langle$s2:act1 isMarriedTo s1:dir1$\rangle$'' is a crossing link (links between different datasets), which means that act1 (at site $s2$) is married to dir1 (at site $s1$).

Now consider the following SPARQL query $Q$ that consists of five triple patterns (e.g., ?a isMarriedTo ?d) over this distributed RDF graph:

\small{
\begin{lstlisting}[language=SQL]
SELECT  ?a ?d  WHERE
{?a isMarriedTo ?d.  ?a actedIn ?f1.
?f1 rdfs:label ?n1.  ?d directed ?f2.
?f2 rdfs:label ?n2.}
\end{lstlisting}
}

\normalsize

Some SPARQL query matches are contained within a fragment, which we call \emph{inner matches}. These inner matches can be found locally by existing centralized techniques at each site. However, if we consider the four datasets independently and ignore the crossing links, some correct answers will be missed, such as (?a=s2:act1, ?d=s1:dir1). The key issue in the distributed environment is how to find subgraph matches that cross multiple fragments---these are called \emph{crossing matches}. For  query $Q$ in Figure \ref{fig:querygraph}, the subgraph induced by vertices 014, 007, 001, 002, 009 and 018 is a crossing match between fragments $F_1$ and $F_2$ in Figure \ref{fig:lodfgraph} (shown in the shaded vertices and red edges). This is the focus of this paper.

%\vspace{-0.05in}
\begin{figure}%[H]
\begin{center}
%\vspace{-0.15in}
    \includegraphics[scale=0.32]{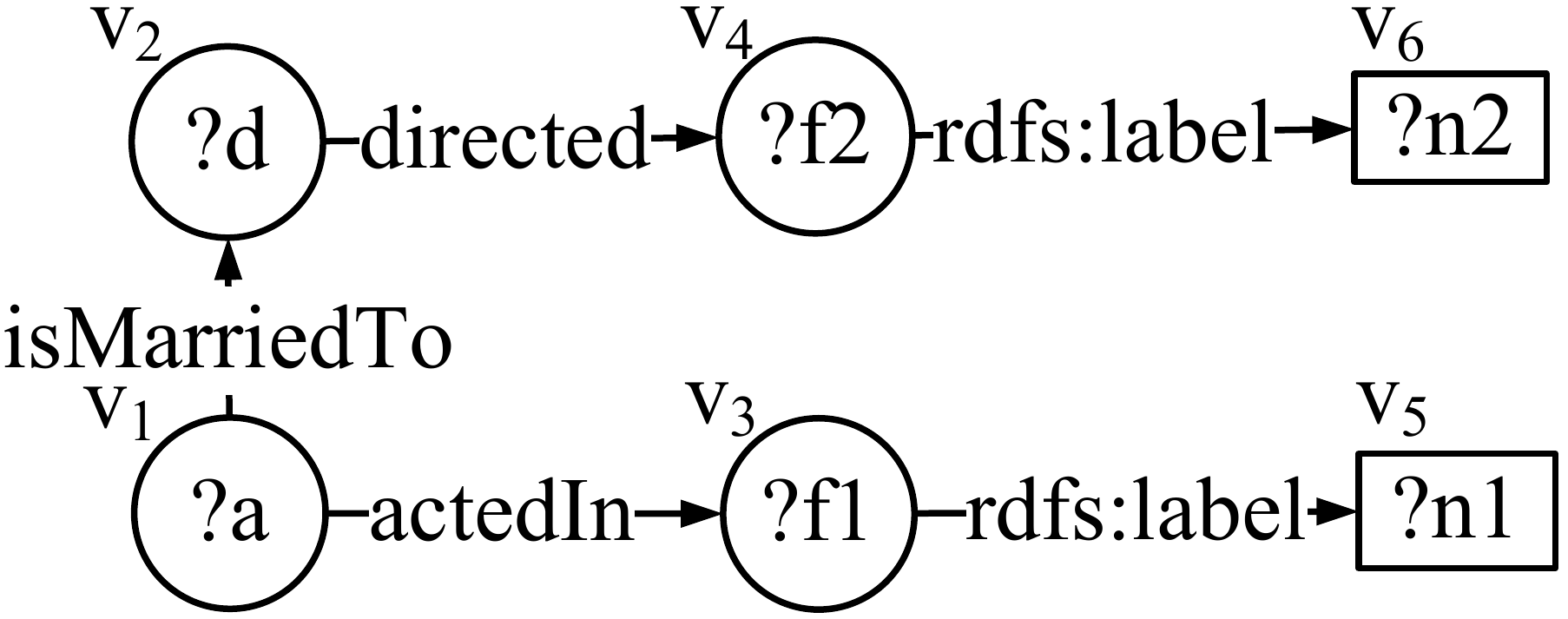}
%       \vspace{-0.1in}
   \caption{SPARQL Query Graph $Q$}
   \label{fig:querygraph}
%   \vspace{-0.2in}
\end{center}
\end{figure}
%\vspace{-0.25in}

There are two important issues to be addressed in this framework. The first is to compute the partial evaluation results at each site given a query graph $Q$ (i.e., the  \emph{local partial match}), which, intuitively, is the overlapping part between a crossing match and a fragment. This is discussed in Section \ref{sec:partialcomputing}. The second one is the assembly of these local partial matches to compute crossing matches. We consider two different strategies: \emph{centralized assembly}, where all local partial matches are sent to a single site (Section \ref{sec:centralized}); and \emph{distributed assembly}, where the local partial matches are assembled at a number of sites in parallel (Section \ref{sec:distributed}).

%\vspace{0.1in}
The main benefits of our solution are twofold:

\begin{itemize}
\item Our solution does not depend on any specific partitioning strategy. In existing partition-based methods, the query processing always depends on a certain RDF graph partitioning strategy, which may be difficult to enforce in certain circumstances. The partition-agnostic framework enables us to adopt any partition-based optimization, although this is orthogonal to our solution in this paper.

\item Our method guarantees to involve fewer vertices or edges in intermediate results than other partition-based solutions, which we prove in Section \ref{sec:partialcomputing} (Proposition \ref{proposition:minimalISNum}). This property often results in smaller number of intermediate results and lowers the cost of our approach, which we demonstrate experimentally in Section \ref{sec:experiment}.

\nop{Intuitively, this conclusion comes from two reasons: (1) Query-decomposition methods cannot distinguish the subquery results contributing to inner matches and that contributing to crossing matches. All subgraph results are regarded as intermediate results that are sent to the control site for assembly. However, in our solution, only local partial matches are needed for assembly, which lead communication cost. Inner matches are found locally at each site in parallel. (2) We prove that any local partial match should be found at any graph partition-method; otherwise, it may lead to true result dismissal (see Theorem \ref{theorem:optimal}). In other words, it is impossible that there exists one edge/vertex in some local partial match but does not exist in any subquery result of other graph partition-based algorithms.}

\end{itemize}

\nop{
This is the first work that proposes the partial evaluation and assembly framework for executing distributed SPARQL queries. Within this context we make the following contributions.

\begin{enumerate}
\item We formally define the partial evaluation result for SPARQL queries and propose an efficient algorithm to find them. We prove theoretically the correctness and the optimality of our algorithm.
\item We propose centralized and distributed methods to assemble the local partial matches and experimentally study their performance advantages and bottlenecks in different settings.
\end{enumerate}
}
The rest of the paper is organized as follows. We discuss related work in the areas of distributed SPARQL query processing and partial query evaluation in Section \ref{sec:relatedwork}. Section \ref{sec:background} provides the fundamental definitions that form the background for this work and introduces the overall execution framework. Computation of local matches at each site is covered in Section \ref{sec:partialcomputing} and the centralized and distributed assembly of partial results to compute the final query result is discussed in Section \ref{sec:assembling}. We also study how to evaluate general SPARQLs in Section \ref{sec:Extensions}. We evaluate our approach, both in terms of its internal characteristics and in terms of its relative performance against other approaches in Section \ref{sec:experiment}. Section \ref{sec:conclusion} concludes the paper and outlines some future research directions.

%!TEX root =  distributedgStore.tex

\section{Related Work}\label{sec:relatedwork}

%There are two threads of related work: distributed SPARQL query processing and partial evaluation.

\textbf{Distributed SPARQL Query Processing.} As noted above, there are three general approaches to distributed SPARQL query processing:  \emph{cloud-based} approaches, \emph{partition-based} approaches and \emph{federated SPARQL query} systems.

(1) \underline{Cloud-based Approaches}

There have been a number of works (e.g., \cite{DBLP:JenaHBase,DBLP:SHARD,TKDE11:HadoopRDF,DBLP:PredicateJoin,ICDE13:EAGRE,VLDB13:Trinity,WWW12:H2RDF,SIGMOD14:H2RDF}) focused on managing large RDF datasets using existing cloud platforms; a very good survey of these is \cite{VLDBJ14:RDFCloudSurvey}. Many of these approaches follow the MapReduce paradigm; in particular they use HDFS \cite{DBLP:SHARD,TKDE11:HadoopRDF,DBLP:PredicateJoin,ICDE13:EAGRE}, and  store RDF triples in flat files in HDFS. When a SPARQL query is issued,  the HDFS files are scanned to find the matches of each triple pattern, which are then joined  using one of the MapReduce join implementations (see \cite{Li:2013uq} for more detailed description of these). The most important difference among these approaches is how the RDF triples are stored in HDFS files; this determines how the triples are accessed and the number of MapReduce jobs. In particular, SHARD \cite{DBLP:SHARD} directly stores the data in a single file and each line of the file represents all triples associated with a distinct subject. HadoopRDF \cite{TKDE11:HadoopRDF} and PredicateJoin \cite{DBLP:PredicateJoin} further partition RDF triples based on the predicate and store each partition within one HDFS file. EAGRE \cite{ICDE13:EAGRE} first groups all subjects with similar properties into an entity class, and then constructs a compressed RDF graph containing only entity classes and the connections between them. It partitions the compressed RDF graph  using the METIS algorithm \cite{DBLP:metis}. Entities are placed into HDFS according to the partition set that they belong to.

Besides the HDFS-based approaches, there are also some works  that use other NoSQL distributed data stores to manage RDF datasets. JenaHBase \cite{DBLP:JenaHBase} and H$_2$RDF \cite{WWW12:H2RDF,SIGMOD14:H2RDF} use some permutations of subject, predicate, object to build indices that are then stored in HBase (http://hbase.apache.org). Trinity.RDF \cite{VLDB13:Trinity} uses the distributed memory-cloud graph system Trinity \cite{SIGMOD13:Trinity} to index and store the RDF graph. It uses hashing on the vertex values to obtain a disjoint partitioning of the RDF graph that is placed on nodes in a cluster.

These approaches benefit from the high scalability and fault-tolerance offered by cloud platforms, but may suffer lower performance due to the difficulties of adapting MapReduce to graph computation.

(2) \underline{Partition-based Approaches}

The partition-based approaches \cite{VLDB11:GraphPartition,DBLP:WARP,DBLP:Partout,DBLP:journals/pvldb/LeeL13,DBLP:conf/IEEEcloud/LeeLTZZ13,SIGMOD2014:TriAD} partition an RDF graph $G$ into several fragments and place each at a different site in a parallel/distributed system. Each site hosts a centralized RDF store of some kind. At run time, a SPARQL query $Q$ is decomposed into several subqueries such that each subquery can be answered locally at one site, and the results are then agregated. Each of these papers proposes its own data partitioning strategy, and different partitioning strategies result in different query processing methods.

In GraphPartition \cite{VLDB11:GraphPartition}, an RDF graph $G$ is partitioned into $n$ fragments, and each fragment is extended by including $N$-hop neighbors of boundary vertices. According to the partitioning strategy, the diameter of the graph corresponding to each decomposed subquery should not be larger than $N$ to enable subquery processing at each local site. WARP \cite{DBLP:WARP} uses some frequent structures in workload to further extend the results of GraphPartition. Partout \cite{DBLP:Partout} extends the concepts of minterm predicates in relational database systems, and uses the results of minterm predicates as the fragmentation units. Lee et. al. \cite{DBLP:journals/pvldb/LeeL13,DBLP:conf/IEEEcloud/LeeLTZZ13} define the partition unit as a vertex and its neighbors, which they call a ``vertex block''. The vertex blocks are distributed based on a set of heuristic rules. A query is partitioned into blocks that can be executed among all sites in parallel and without any communication. TriAD uses METIS \cite{DBLP:metis} to divide the RDF graph into many partitions and the number of result partitions is much more than the number of sites. Each result partition is considered as a unit and distributed among different sites. At each site, TriAD maintains six large, in-memory vectors of triples, which correspond to all SPO permutations of triples. Meanwhile, TriAD constructs a summary graph to maintain the partitioning information.

\nop{
TriAD \cite{SIGMOD2014:TriAD} distributes RDF-3x \cite{DBLP:rdf3x}. TriAD uses METIS \cite{DBLP:metis} to divide the RDF graph into many partitions and the number of result partitions is much more than the number of sites. Each result partition is considered as a unit and distributed among different sites. Meanwhile, TriAD constructs a summary graph to maintain the partitioning information. During the query processing phase, TriAD determines a query plan based on this summary graph.
}

All of the above methods require partitioning and distributing the RDF data according to specific requirements of their approaches. However, in some applications, the RDF repository partitioning strategy is not controlled by the distributed RDF system itself. There may be some administrative requirements that influence the data partitioning. For example, in some applications, the RDF knowledge bases are partitioned according to topics (i.e., different domains), or are partitioned according to different data contributors.  Therefore, partition-tolerant SPARQL processing may be desirable. This is the motivation of our partial-evaluation and assembly approach.

As well, these approaches evaluate the SPARQL query based on query decomposition, which generate more intermediate results. We provide a detailed experimental comparison in Section \ref{sec:experiment}.

(3) \underline{Federated SPARQL Query Systems}

Federated queries run SPARQL queries over multiple SPARQL endpoints. A typical example is linked data, where different RDF repositories are interconnected, providing a \emph{virtually integrated distributed database}. Federated SPARQL query processing is a very different environment than what we target in this paper, but we discuss these systems for completeness. 

A common technique is to precompute metadata for each individual SPARQL endpoints. Based on the metadata, the original SPARQL query is decomposed into several subqueries, where each subquery is sent to its relevant SPARQL endpoints. The results of subqueries are then joined together to answer the original SPARQL query. In DARQ \cite{DBLP:DARQ}, the metadata is called \emph{service description} that describes which triple patterns (i.e., predicate) can be answered. In \cite{WWW2010:QTree}, the metadata is called Q-Tree, which is a variant of RTree. Each leaf node in Q-Tree stores a set of source identifers, including one for each source of a triple approximated by the node. SPLENDID \cite{DBLP:SPLENDID} uses Vocabulary of Interlinked Datasets (VOID) as the metadata. HiBISCuS \cite{DBLP:HiBISCuS} relies on capabilities to compute the metadata. For each source, HiBISCuS defines a set of capabilities which map the properties to their subject and object authorities. TopFed \cite{DBLP:TopFed} is a biological federated SPARQL query engine. Its metadata comprises of an N3 specification file and a Tissue Source Site to Tumour (TSS-to-Tumour) hash table, which is devised based on the data distribution.

In contrast to these, FedX \cite{DBLP:FedX} does not require preprocessing, but sends  ``SPARQL ASK'' to collect the metadata on the fly. Based on the results of ``SPARQL ASK'' queries, it decomposes the query into subqueries and assign subqueries with relevant SPARQL endpoints.

Global query optimization in this context has also been studied. Most federated query engines employ existing optimizers, such as dynamic programming \cite{CiteSeerX:SystemR}, for optimizing the join order of local queries. Furthermore, DARQ \cite{DBLP:DARQ} and FedX \cite{DBLP:FedX} discuss the use of semijoins  to compute a join between intermediate results at the control site and SPARQL endpoints.

\textbf{Partial Evaluation}. Partial evaluation has been used in many applications ranging from compiler optimization to distributed evaluation of functional programming languages \cite{DBLP:journals/csur/Jones96}. Recently, partial evaluation has also been used for evaluating queries on distributed XML trees and graphs \cite{VLDB06:XQuery,SIGMOD07:XPath,TODS12:XPath,fan:reachability}. In \cite{VLDB06:XQuery,SIGMOD07:XPath,TODS12:XPath}, partial evaluation is used to evaluate some XPath queries on distributed XML. These works serialize XPath queries to a vector of subqueries, and find the partial results of all subqueries at each site by using a top-down \cite{SIGMOD07:XPath} or bottom-up \cite{VLDB06:XQuery} traversal over the XML tree. Finally, all partial results are assembled together at the server site to form final results. Note that, since XML is a tree-based data structure, these works serialize XPath queries and traverse XML trees in a topological order. However, the RDF data and SPARQL queries are graphs rather than trees. Serializing the SPARQL queries and traversing the RDF graph in a topological order is not intuitive.

There are some prior works that consider partial evaluation on graphs. For example, Fan et al \cite{fan:reachability} study reachability query processing over distributed graphs using the partial evaluation strategy. Partial evaluation-based graph simulation is well studied by Fan et al. \cite{DBLP:journals/pvldb/FanWWD14} and Shuai et al. \cite{DBLP:conf/www/MaCHW12}. However, SPARQL query semantics is based on graph homomorphism \cite{DBLP:journals/tods/PerezAG09}, not graph simulation. The two concepts are formally different  (i.e., they produce different results) and the two problems have very different complexities. Homomorphism defines a ``function'' while simulation defines a ``relation'' -- relation allows  ``one-to-many'' mappings while function does not. Consequently, the results are different. The computational hardness of the two problems are also different. Graph homomorphism is a classical NP-complete problem \cite{DBLP:journals/rsa/DyerG00}, while graph simulation has a polynomial-time algorithm ($O((|V(G)| + |V(Q)|)(|E(G)| + |E(Q)|))$) \cite{DBLP:journals/pvldb/FanLMTWW10}, where $|V(G)|$ ($|V(Q)|$) and $|E(G)|$ ($|E(Q)|$) denote the number of vertices and edges in RDF data graph $G$ (and query graph $Q$).  Thus, the solutions based on graph simulation cannot be applied to the problem studied in this paper. To the best of our knowledge, there is no prior work in applying partial evaluation to SPARQL query processing.

%!TEX root =  distributedgStore.tex

\section{Background and Framework}
\label{sec:background}

An RDF dataset can be represented as a graph where subjects and objects are vertices and triples are labeled edges.

\begin{definition}\label{def:graph} \textbf{(RDF Graph)}
An RDF graph is denoted as $G=\{V,$ $E,\Sigma \}$, where $V$ is a set of vertices that correspond to all subjects and objects in RDF data; $E \subseteq V \times V$ is a multiset of directed edges that correspond to all triples in RDF data; $\Sigma$ is a set of edge labels. For each edge $e \in E$, its edge label is its corresponding property.
\end{definition}

Similarly, a SPARQL query can also be represented as a query graph $Q$. In this paper, we first focus on basic graph pattern (BGP) queries as they are foundational to SPARQL, and focus on techniques for handling these. We extend this discussion in Section \ref{sec:Extensions} to general SPARQL queries involving FILTER, UNION, and OPTIONAL.

\begin{definition}\label{def:query}\textbf{(SPARQL BGP Query)}
A \emph{SPARQL BGP query} is denoted as $Q=\{V^{Q},$ $E^{Q}, \Sigma^{Q}\}$, where $V^{Q} \subseteq V\cup V_{Var}$ is a set of vertices, where $V$ denotes all vertices in RDF graph $G$ and $V_{Var}$ is a set of variables; $E^{Q} \subseteq V^{Q} \times V^{Q}$ is a multiset of edges in $Q$; Each edge $e$ in $E^Q$ either has an edge label in $\Sigma$ (i.e., property) or the edge label is a variable.
\end{definition}

We assume that $Q$ is a connected graph; otherwise, all connected components of $Q$ are considered separately.  Answering a SPARQL query is equivalent to finding all subgraph \emph{matches} (Definition \ref{def:sparqlmatch}) of $Q$ over RDF graph $G$.

\begin{definition}\label{def:sparqlmatch}\textbf{(SPARQL Match)}
Consider an RDF graph $G$ and a connected query graph $Q$ that has $n$ vertices $\{v_1,...,v_n\}$. A subgraph $M$ with $m$ vertices $\{u_1, ...,u_m\}$ (in $G$) is said to be a \emph{match} of $Q$ if and only if there exists a \emph{function} $f$ from $\{v_1,...,v_n\}$ to $\{u_1,...,u_m\}$ ($n \ge m$), where the following conditions hold:
\begin{enumerate}
\item if $v_i$ is not a variable, $f(v_i)$ and $v_i$ have the same URI or literal value ($1\leq i \leq n$);
\item if $v_i$ is a variable, there is no constraint over $f(v_i)$ except that $f(v_i)\in \{u_1,...,u_m\}$ ;
\item if there exists an edge $\overrightarrow{v_iv_j}$ in $Q$, there also exists an edge $\overrightarrow{f{(v_i)}f{(v_j)}}$ in $G$. Let $L(\overrightarrow{v_iv_j})$ denote a multi-set of labels between $v_i$ and $v_j$ in $Q$; and $L(\overrightarrow{f{(v_i)}f{(v_j)}})$ denote a multi-set of labels between $f(v_i)$ and $f(v_j)$ in $G$. There must exist an \emph{injective function} from edge labels in $L(\overrightarrow{v_iv_j})$ to edge labels in $L(\overrightarrow{f{(v_i)}f{(v_j)}})$. Note that a variable edge label in $L(\overrightarrow{v_iv_j})$ can match any edge label in $L(\overrightarrow{f{(v_i)}f{(v_j)}})$.
\end{enumerate}
\end{definition}

Vector $[ f{(v_1)}, ..., f{(v_n)}]$ is a serialization of a SPARQL match. Note that we allow that $f(v_i)=f(v_j)$ when $1\le i\neq j \le n$. In other words, a match of SPARQL $Q$ defines  a \emph{graph homomorphism}.
%\footnote{A graph homomorphism $f$ from a graph $G=(V,E)$ to a graph $G^\prime=(V^\prime,E^\prime)$, written $f: G \rightarrow G^\prime$, is a mapping $f:V \rightarrow V^\prime$ from the vertex set of G to the vertex set of G' such that $ \{u,v\}\in E$ implies $ \{f(u),f(v)\}\in E^\prime$.} from $Q$ over RDF graph $G$.

In the context of this paper, an RDF graph $G$ is vertex-disjoint partitioned  into a number of \emph{fragments}, each of which resides at one site. The vertex-disjoint partitioning has been used in most distributed RDF systems, such as GraphPartition \cite{VLDB11:GraphPartition}, EAGRE \cite{ICDE13:EAGRE} and TripleGroup \cite{DBLP:journals/pvldb/LeeL13}. Different distributed RDF systems utilize different vertex-disjoint partitioning algorithms, and the partitioning algorithm is orthogonal to our approach. Any vertex-disjoint partitioning method can be used in our method, such as METIS \cite{DBLP:metis} and MLP \cite{DBLP:conf/icde/WangXSW14}. \nop{How to find an optimal partitioning is beyond the scope of this paper.}

The vertex-disjoint partitioning methods guarantee that there are no overlapping vertices between fragments. However, to guarantee data integrity and consistency, we store some replicas of crossing edges. Since the RDF graph $G$ is partitioned by our system,  metadata is readily available regarding crossing edges (both outgoing and incoming edges) and the endpoints of crossing edges. Formally, we define the \emph{distributed RDF graph} as follows.

\nop{
distributed over different sites. We call $G$ an \emph{distributed RDF graph}, and each site maintain a subgraph of $G$, called \emph{fragment}. For simplicity, we assume that each entity is only resident in a single site, i.e., we ignore replication (see condition (1) of Definition \ref{def:internaldistributedgraph}), but handling the more general case is not complicated.}

%\vspace{-0.05in}
\begin{definition} \label{def:internaldistributedgraph} \textbf{(Distributed RDF Graph)}  A distributed RDF graph $G=\{V,E, \Sigma \}$ consists of a set of fragments $\mathcal{F} = \{F_1,F_2,...,F_k\}$ where each $F_i$ is specified by $(V_i \cup V_i^e, E_i \cup E_i^c, \Sigma_i)$ ($i=1,...,k$) such that
\begin{enumerate}
\item $\{V_1,...,V_k\}$ is a partitioning of $V$, i.e., $V_i  \cap V_j  = \emptyset ,1 \le i,j \le k,i \ne j $ and $\bigcup\nolimits_{i = 1,...,k} {V_i  = V}$ ;
\item $E_i \subseteq V_i \times V_i$, $i=1,...,k$;

\nop{\item Given an edge $\overrightarrow{uu^\prime}$ in RDF graph $G$, if $u \in V_i \wedge u^\prime \in V_i$, $\overrightarrow{uu^\prime}$ is called an \emph{inner edge} in fragment $F_i$. If $u \in V_i \wedge u^\prime \in V_j \wedge i \neq j$ , $\overrightarrow{uu^\prime}$ is called a \emph{crossing edge} between fragment $F_i$ and $F_j$, and $u$ and $u^\prime$ are called \emph{boundary vertices} in $F_i$ and $F_j$, respectively;}

\item $E_i^c$ is a set of crossing edges between $F_i$ and other fragments, i.e.,
\begin{multline*}
 E_i ^c = (\bigcup\nolimits_{1 \le j \le k \wedge j \ne i} {\{ \overrightarrow {uu^\prime} |u \in F_i  \wedge u^\prime \in F_j  \wedge \overrightarrow{uu^\prime}  \in E \} }) \\ \bigcup
 (\bigcup\nolimits_{1 \le j \le k \wedge j \ne i} {\{ \overrightarrow {u^\prime u} |u \in F_i  \wedge u^\prime \in F_j  \wedge \overrightarrow{u^\prime u} \in E   \}} )
 \end{multline*}

 \item A vertex $u^\prime \in V_i^e$ if and only if vertex $u^\prime$ resides in other fragment $F_j$ and $u^{\prime}$ is an endpoint of a crossing edge between fragment $F_i$ and $F_j$ ($F_i \neq F_j$),  i.e.,
\begin{multline*}
 V_i^e = (\bigcup\nolimits_{1 \le j \le k \wedge j \ne i} \{ {u^\prime} |\overrightarrow {uu^\prime}$  $\in E_i ^c \wedge u \in F_i \} ) \bigcup \\
 (\bigcup\nolimits_{1 \le j \le k \wedge j \ne i} {\{ {u^\prime} |\overrightarrow {u^\prime u}  \in E_i ^c \wedge u \in F_i \} })
 \end{multline*}

\item Vertices in $V_i^e$ are called \emph{extended} vertices of $F_i$ and all vertices in $V_i$ are called \emph{internal} vertices of $F_i$;
\item $\Sigma_i$ is a set of edge labels in $F_i$.
\end{enumerate}
\end{definition}
%\vspace{-0.05in}

%\vspace{-0.05in}
\begin{example}\label{example:dataandquery}
Figure \ref{fig:lodfgraph} shows a distirbuted RDF graph $G$ consisting of four fragments $F_1$, $F_2$, $F_3$ and $F_4$. The numbers besides the vertices are vertex IDs that are introduced for ease of presentation. In Figure \ref{fig:lodfgraph}, $\overrightarrow{002,001}$ is a crossing edge between $F_1$ and $F_2$. As well, edges $\overrightarrow{004,011}$, $\overrightarrow{001,012}$ and $\overrightarrow{006,008}$ are crossing edges between $F_1$ and $F_3$. Hence, $V_1^e=\{002,006,012,004\}$ and $E_1^c=\{\overrightarrow{002,001},$  $\overrightarrow{004,011},\overrightarrow{001,012},$ $ \overrightarrow{006,008}\}$.
\end{example}

%\vspace{-0.1in}
\begin{definition}{(\textbf{Problem Statement})}
Let $G$ be a distributed RDF graph that consists of a set of fragments $\mathcal{F} = \{F_{1}, \ldots, F_{k}\}$ and let $\mathcal{S} = \{S_{1}, \ldots, S_{k}\}$ be a set of computing nodes such that $F_{i}$ is located at $S_{i}$. Given a SPARQL query graph $Q$, our goal is to find all \emph{SPARQL matches} of $Q$ in $G$.
\nop{Given a distributed RDF $G$ that consists of $k$ fragments $\{F_1,F_2,...,F_k\}$, each located at a different computating node, and a SPARQL query graph $Q$, our goal is to find all \emph{SPARQL matches} of $Q$ in $G$.}
\end{definition}
%\vspace{-0.05in}

Note that for simplicity of exposition, we are assuming that each site hosts one fragment. Inner matches can be computed locally using a centralized RDF triple store, such as RDF-3x \cite{DBLP:rdf3x}, SW-store \cite{DBLP:swstore} or gStore \cite{Zou:2013fk}. In our prototype development and experiments, we modify gStore, a graph-based SPARQL query engine \cite{Zou:2013fk}, to perform partial evaluation. The main issue of answering SPARQL queries over the distributed RDF graph is finding crossing matches efficiently. That is a major focus of this paper.

\begin{example}\label{example:crossingmatch}
Given a SPARQL query graph $Q$ in Figure \ref{fig:querygraph}, the subgraph induced by vertices 014,007,001,002,009 and 018 (shown in the shaded vertices and the red edges in Figure \ref{fig:lodfgraph}) is a crossing match of $Q$.
\end{example}

%This is the start of the old Framework and Intuition Section

We utilize a \emph{partial evaluation and assembly} \cite{DBLP:journals/csur/Jones96} framework to answer SPARQL queries over a distributed RDF graph $G$. Each site $S_i$ treats fragment $F_i$ as the known input $s$ and other fragments  as yet unavailable input $\overbar{G}$ (as defined in Section \ref{sec:introduction}) \cite{fan:reachability}.

In our execution model, each site $S_i$ receives the full query graph $Q$. In the partial evaluation stage, at each site $S_i$, we find all \emph{local partial matches} (Definition  \ref{def:localmaximal}) of $Q$ in $F_i$. We prove that an overlapping part between any crossing match and fragment $F_i$ must be a local partial match in $F_i$ (see Proposition \ref{proposition:overlapping}).

To demonstrate the intuition behind dealing with crossing edges, consider the case in Example \ref{example:crossingmatch}. The crossing match $M$ overlaps with two fragments $F_1$ and $F_2$.
If we can find the overlapping parts between $M$ and $F_1$, and $M$ and $F_2$, we can assemble them to form a crossing match. For example, the subgraph induced by vertices 014, 007, 001 and 002 is an overlapping part between $M$ and $F_1$. Similarly, we can also find the overlapping part between $M$ and $F_2$. We assemble them based on the common edge $\overrightarrow{002,001}$ to form a crossing match, as shown in Figure \ref{fig:assembly}.

\begin{figure}[h]
\begin{center}
    \includegraphics[scale=0.21]{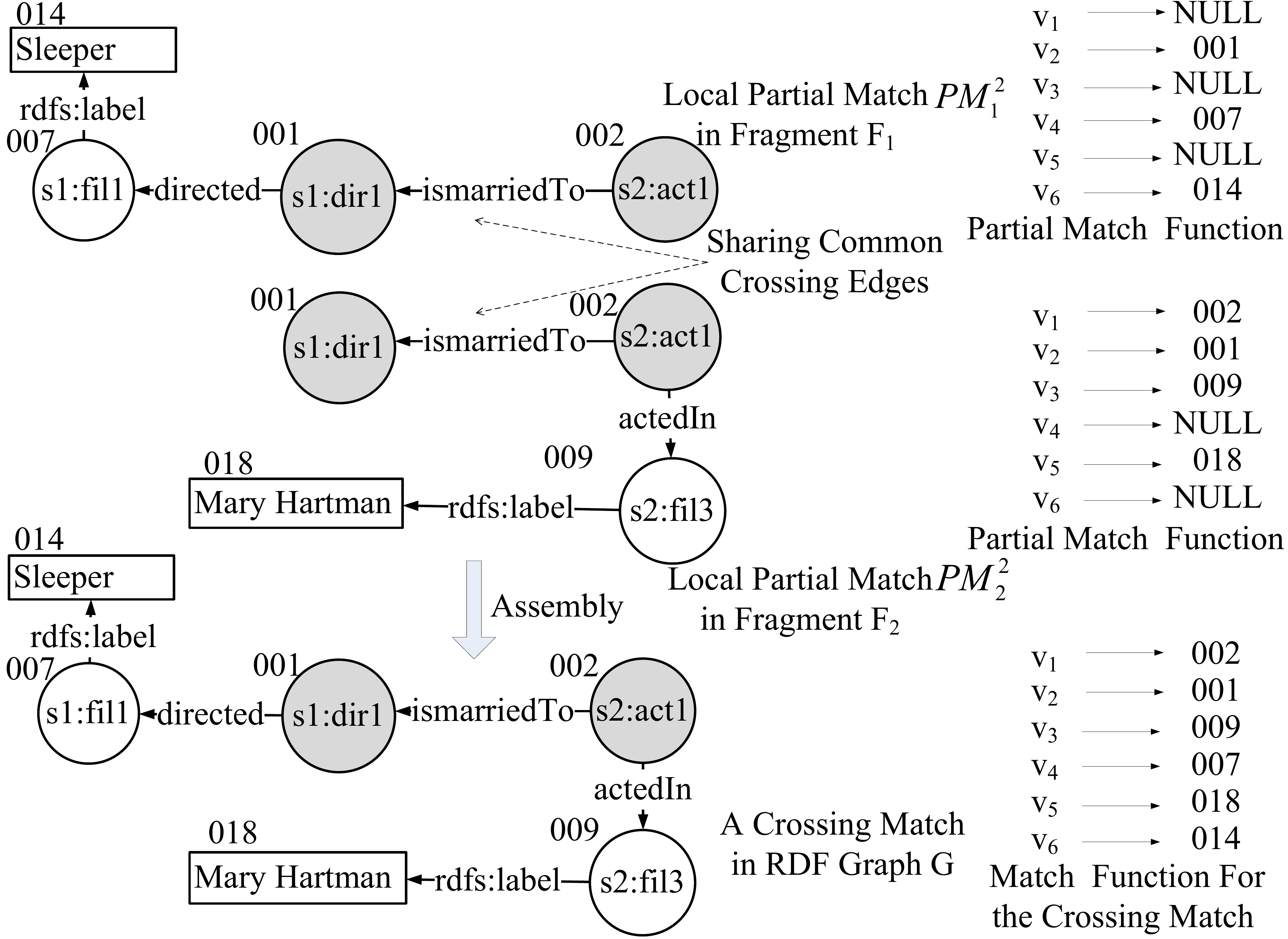}
%       \vspace{-0.1in}
   \caption{Assemble Local Partial Matches}
   \label{fig:assembly}
%   \vspace{-0.2in}
\end{center}
\end{figure}

In the assembly stage, these local partial matches are assembled to form crossing matches. In this paper, we consider two assembly strategies: centralized and distributed (or parallel). In centralized, all local partial matches are sent to a single site for the assembly. In distributed/parallel, local partial matches are combined at a number of sites in parallel (see Section \ref{sec:assembling}).

There are three steps in our method.

\textbf{Step 1} (\emph{Initialization}): A SPARQL query $Q$ is input and sent to each site in $\mathcal{S}$.

\textbf{Step 2} (\emph{Partial Evaluation}): Each site $S_i$ finds  \emph{local partial matches} of $Q$ over fragment $F_i$. This step is executed in  parallel at each site (Section \ref{sec:partialcomputing}).

\textbf{Step 3} (\emph{Assembly}): Finally, we assemble all local partial matches to compute complete crossing matches. The system can use the centralized (Section \ref{sec:centralized}) or the distributed assembly approach (Section \ref{sec:distributed}) to find crossing matches.

%!TEX root =  distributedgStore.tex
\section{Partial Evaluation}
\label{sec:partialcomputing}
We first formally define a local partial match (Section \ref{sec:advancednotation}) and then discuss how to compute it efficiently (Section \ref{sec:computinglocal}).

\vspace{-0.05in}
\subsection{Local Partial Match---Definition}\label{sec:advancednotation}

\nop{We assume that the RDF graph $G$ is partitioned into a set of fragments $\mathcal{F}=\{F_{1}, \ldots, F_{n}\}$ and we let $\mathcal{S}=\{S_{1}, \ldots, S_{n}\}$ denote the set of sites. In this paper we are not concerned with the particular graph partitioning algorithm that is used -- there are a number of good candidates -- and, for ease of exposition, we assume that the number of fragments is equal to the number of sites so that each site holds one fragment.}

Recall that each site $S_i$ receives the full query graph $Q$ (i.e., there is no query decomposition).
In order to answer query $Q$, each site $S_i$ computes the partial answers (called \emph{local partial matches}) based on the known input $F_i$  (recall that, for simplicity of exposition, we assume that each site hosts one fragment as indicated by its subscript). Intuitively, a local partial match $PM_i$ is an overlapping part between a crossing match $M$ and fragment $F_i$ at the partial evaluation stage. Moreover, $M$ may or may not exist depending on the yet unavailable input $\overbar{G}$ . Based only on the known input $F_i$, we cannot judge whether or not $M$ exists. For example, the subgraph induced by vertices 014, 007, 001 and 002 (shown in shared vertices and red edges) in Figure \ref{fig:lodfgraph} is a local partial match between $M$ and $F_1$.

\vspace{-0.05in}
\begin{definition}\label{def:localmaximal} \textbf{(Local Partial Match)} Given a SPARQL query graph $Q$ with $n$ vertices $\{v_1,...,v_n\}$ and a connected subgraph $PM$ with $m$ vertices $\{u_1,...,u_m\}$ ($m \leq n$) in a fragment $F_k$, $PM$ is a \emph{local partial match} in fragment $F_k$ if and only if there exists a function $f:\{v_1,...,v_n\}\to \{u_1,...,u_m\} \cup \{NULL\}$, where the following conditions hold:

\begin{enumerate}
\item If $v_i$ is not a variable, $f(v_i)$ and $v_i$ have the same URI or literal or $f(v_i)=NULL$.
\item If $v_i$ is a variable, $f(v_i) \in \{u_1,...,u_m\}$ or $f(v_i)=NULL$.
\item If there exists an edge $\overrightarrow{v_iv_j}$ in $Q$ ($1 \leq i\neq j \leq n$), then $PM$ should meet one of the following five conditions: (1) there also exists an edge $\overrightarrow{f{(v_i)}f{(v_j)}}$ in $PM$ with property $p$, and $p$ is the same to the property of $\overrightarrow{v_iv_j}$; (2) there also exists an edge $\overrightarrow{f{(v_i)}f{(v_j)}}$ in $PM$ with property $p$, and the property of $\overrightarrow{v_iv_j}$ is a variable; (3) there does not exist an edge $\overrightarrow{f{(v_i)}f{(v_j)}}$, but $f{(v_i)}$ and $f{(v_j)}$ are both in $V_k^e$; (4) $f{(v_i)}=NULL$; (5) $f{(v_j)}=NULL$.
\item  $PM$ contains at least one crossing edge, which guarantees that an empty match does not qualify.
\item If $f(v_i) \in V_k$ (i.e., $f(v_i)$ is an internal vertex in $F_k$) and $\exists \overrightarrow {v_i v_j} \in Q$ (or $\overrightarrow {v_j v_i} \in Q$), there must exist $f(v_j) \neq NULL$ and  $\exists \overrightarrow {f(v_i)f(v_j)} \in PM$ (or $\exists \overrightarrow {f(v_j)f(v_i)} \in PM$). Furthermore, if $\overrightarrow {v_i v_j}$ (or $\overrightarrow {v_j v_i})$ has a property $p$, $\overrightarrow {f(v_i)f(v_j)}$ (or $\overrightarrow {f(v_j)f(v_i)}$) has the same property $p$.
\item If $f(v_i)$ and $f(v_j)$ are both internal vertices in $PM$, then there exists a \emph{weakly connected path} $\pi$ between $v_i$ and $v_j$ in $Q$ and each vertex in $\pi$ maps to an internal vertex of $F_k$ in $PM$.
\end{enumerate}

Vector $[ f{(v_1)}, ..., f{(v_n)}]$ is a serialization of a local partial match.
\end{definition}

\vspace{-0.1in}
\begin{example}\label{exmple:local} Given a SPARQL query $Q$ with six vertices in Figure \ref{fig:querygraph}, the subgraph induced by vertices 001, 002, 007 and 014 (shown in shaded circles and red edges) is a \emph{local partial match} of $Q$ in fragment $F_1$. The function is $\{(v_1,
 002), (v_2, 001),$ $(v_3, NULL), (v_4,$ $007), (v_5, NULL), (v_6, 014)\}$.  The five different local partial matches in $F_1$ are shown in Figure \ref{fig:pmset}.
\end{example}
\vspace{-0.05in}

Definition \ref{def:localmaximal} formally defines a \emph{local partial match}, which is
a subset of a complete SPARQL match. Therefore, some conditions in Definition \ref{def:localmaximal} are analogous to SPARQL match with some subtle differences. In Definition \ref{def:localmaximal}, some vertices of query $Q$ are not matched in a local partial match. They are allowed to match a special value NULL (e.g., $v_3$ and $v_5$ in Example \ref{exmple:local}). As mentioned earlier, a local partial match is the overlapping part of an unknown crossing match and a fragment $F_i$. Therefore, it must have a crossing edge, i.e, Condition 4.

The basic intuition of Condition 5 is that if vertex $v_i$ (in query $Q$) is matched to an internal vertex, all of $v_i$'s neighbours should be matched in this local partial match as well. The following example illustrates the intuition.

\begin{example}
Let us recall the local partial match $PM_1^2$ of Fragment $F_1$ in Figure \ref{fig:pmset}. An internal vertex 001 in fragment $F_1$ is matched to vertex $v_2$ in query $Q$. Assume that $PM$ is an overlapping part between a crossing match $M$ and fragment $F_1$. Obviously, $v_2$'s neighbors, such as $v_1$ and $v_4$, should  also be matched in $M$. Furthermore, the matching vertices should be $001$'s neighbors. Since $001$ is an internal vertex in $F_1$,  $001$'s neighbors are also in fragment $F_1$.
\end{example}

Therefore, if a $PM$ violates Condition 5, it  cannot be a subgraph of a crossing match. In other words, we are not interested in these subgraphs when finding local partial matches, since they do not contribute to any crossing match.

\begin{definition}\label{def:weaklyconnected}
 Two vertices are \emph{weakly connected} in a directed graph if and only if there exists a connected path between the two vertices when all directed edges are replaced with undirected edges. The path is called a \emph{weakly connected path} between the two vertices.
\end{definition}

Condition 6 will be used to prove the correctness of our algorithm in Propositions \ref{proposition:overlapping} and \ref{proposition:optimal}. The following example shows all local partial matches in the running example.

\begin{example}\label{example:pmset}
Given a query $Q$ in Figure \ref{fig:querygraph} and an RDF graph $G$ in Figure \ref{fig:lodfgraph}, Figure \ref{fig:pmset} shows all local partial matches and their serialization vectors in each fragment. A local partial match in fragment $F_i$ is denoted as $PM_i^j$, where the superscript distinguishes different local partial matches in the same fragment. Furthermore, we underline all extended vertices in serialization vectors.
\end{example}

%Condition 6 says that any two vertices in a local partial match $PM$ is at least weakly connected through non-crossing edges in $PM$.

\begin{figure*}[t]
\begin{center}
    \includegraphics[scale=0.15]{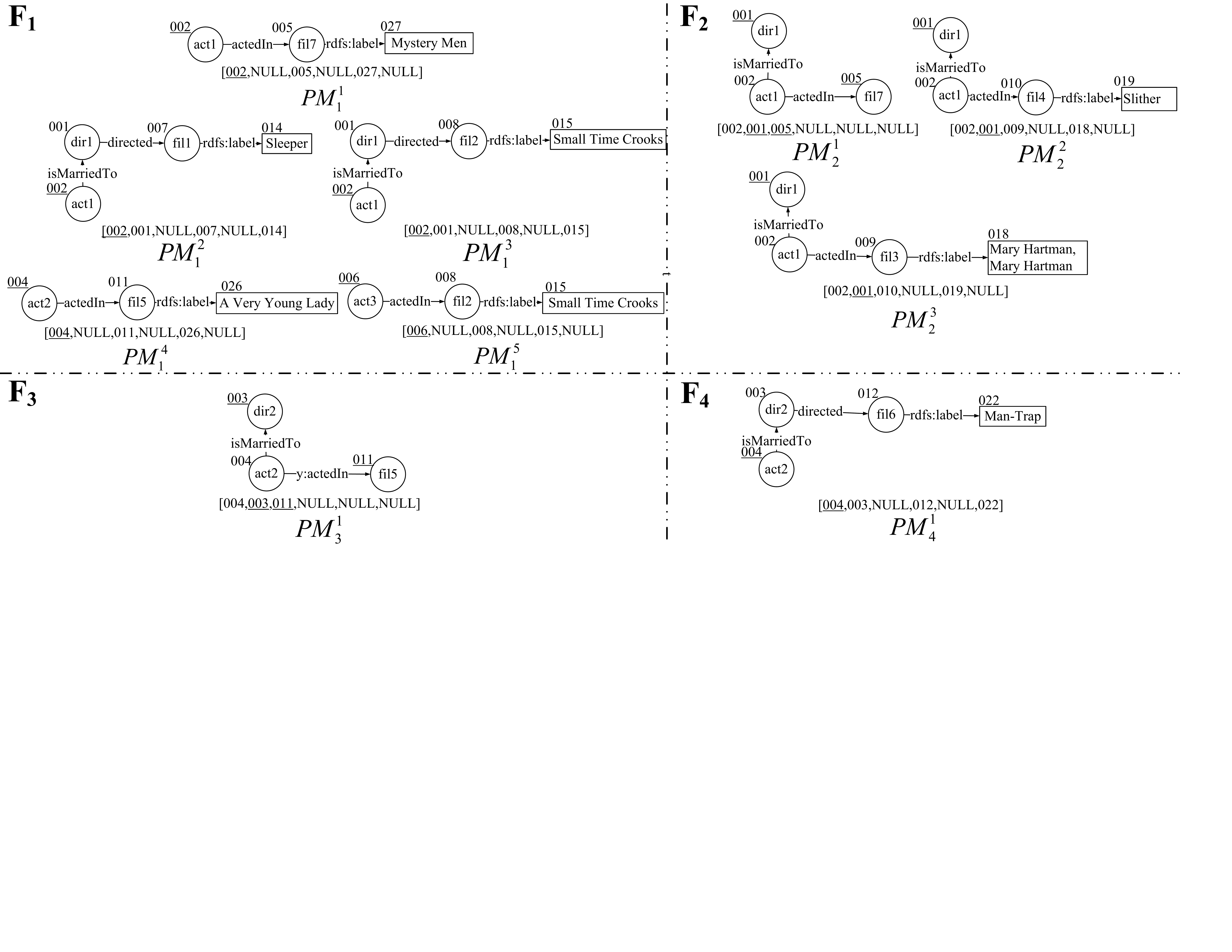}
%       \vspace{-0.0in}
   \caption{ Local Partial Matches of $Q$ in Each Fragment}
   \label{fig:pmset}
%   \vspace{-0.2in}
\end{center}
\end{figure*}

The correctness of our method are stated in the following propositions.
%Readers can skip to Section \ref{sec:computinglocal} if they are not interested in these parts. We highlight the main results as follows.
\begin{enumerate}
\item The overlapping part between any crossing match $M$ and internal vertices of fragment $F_i$ ($i=1,...,k$) must be a local partial match (see Proposition
 \ref{proposition:overlapping}).
\item Missing any local partial match may lead to result dismissal. Thus, the algorithm should find all local partial matches in each fragment (see Proposition \ref{proposition:optimal}).
\item It is impossible to find two local partial matches $M$ and $M^\prime$ in fragment $F$, where $M^\prime$ is a subgraph of $M$, i.e., each local partial match is maximal (see Proposition \ref{proposition:maximal}).
\end{enumerate}

\begin{proposition}\label{proposition:overlapping} Given any crossing match $M$ of SPARQL query $Q$ in an RDF graph $G$, if $M$ overlaps with some fragment $F_i$, let $(M \cap F_i)$ denote the overlapping part between $M$ and fragment $F_i$.  Assume that $(M \cap F_i)$ consists of several weakly connected components, denoted as $(M \cap F_i )=\{PM_1,...,PM_n\}$. Each weakly connected component $PM_a$ ($1\leq a \leq n$) in $(M \cap F_i)$ must be a \emph{local partial match} in fragment $F_i$.

\end{proposition}
\begin{proof} (1) Since $PM_a$ ($1\leq a \leq n$) is a subset of a SPARQL match, it is easy to show that Conditions 1-3 of Definition 6 hold.

(2) We prove that each weakly connected component $PM_a$ ($1\leq a \leq n$) must have at least one crossing edge (i.e., Condition 4) as follows.

Since $M$ is a crossing match of SPARQL query $Q$, $M$ must be weakly connected, i.e., any two vertices in $M$ are weakly connected. Assume that $(M \cap F_i)$ consists of several weakly connected components, denoted as $(M \cap F_i )=\{PM_1,...,PM_n\}$. Let $M = (M \cap F_i ) + \overline {(M \cap F_i )} $, where $ \overline {(M \cap F_i )}$ denotes the complement of $(M \cap F_i )$. It is straightforward show that  $ \overline {(M \cap F_i )}$ must occur in other fragments; otherwise it should be found at $(M \cap F_i )$. $PM_a$ ($1\leq a \leq n$) is weakly disconnected with each other because we remove $ \overline {(M \cap F_i )}$ from $M$. In other words, each $PM_a$ must have at least one \emph{crossing edge} to connect $PM_a$ with $\overline {(M \cap F_i )}$. $\overline {(M \cap F_i )}$ are in other fragments and only crossing edges can connect fragment $F_i$ with other fragments. Otherwise, $PM_a$ is a separated part in the crossing match $M$. Since, $M$ is weakly connected, $PM_a$ has at least one crossing edge, i.e, Condition 4.

(3) For Condition 5, for any internal vertex $u$ in $PM_a$ ($1\leq a \leq n$), $PM_a$ retains all its incident edges. Thus, we can prove that Condition 5 holds.

(4) We define $PM_a$ ($1\leq a \leq n$) as a weakly connected part in $(M \cap F_i )$. Thus, Condition 6 holds.

To summarize, the overlapping part between $M$ and fragment $F_i$ satisfies all conditions in Definition 6. Thus, Proposition \ref{proposition:overlapping} holds. $\Box$

\end{proof}

Let us recall Example \ref{example:pmset}. There are some local partial matches that do not contribute to any crossing match, such as $PM_1^5$ in Figure \ref{fig:pmset}. We call these local partial matches \emph{false positives}. However, the partial evaluation stage only depends on the known input. If we do not know the structures of other fragments, we cannot judge whether or not $PM_1^5$ is a false positive. Formally, we have the following proposition, stating that we have to find all local partial matches in each fragment $F_i$ in the partial evaluation stage.

\begin{proposition}\label{proposition:optimal}  The partial-evaluation-and-assembly algorithm does not miss any crossing matches in the answer set \emph{if and only if} all local partial matches in each fragment are found in the partial evaluation stage.
\end{proposition}
\begin{proof} In two parts:

(1) The ``If'' part: (proven by contradiction).

Assume that all local partial matches are found in each fragment $F_i$ but a cross match $M$ is missed in the answer set. Since $M$ is a crossing match, suppose that $M$ overlaps with $m$ fragments $F_1$,...,$F_m$. According to Proposition \ref{proposition:overlapping}, the overlapping part between $M$ and $F_i$ ($i=1,...,m$) must be a local partial match $PM_i$ in $F_i$. According to the assumption, these local partial matches have been found in the partial evaluation stage. Obviously, we can assemble these partial matches $PM_i$ ($i=1,...,m$) to form the complete cross match $M$.

In other words, $M$ would not be missed if all local partial matches are found. This contradicts the assumption.

(2) The ``Only If'' part: (proven by contradiction).

We assume that a local partial match $PM_i$ in fragment $F_i$ is missed and the answer set can still satisfy no-false-negative requirement. Suppose that $PM_i$ matches a part of $Q$, denoted as $Q^\prime$. Assume that there exists another local partial match $PM_j$ in $F_j$ that matches a complementary graph of $Q^\prime$, denoted as $\overbar{Q} = Q \setminus Q^\prime$. In this case, we can obtain a complete match $M$ by assembling the two local partial matches. If $PM_i$ in $F_i$ is missed, then match $M$ is missed. In other words, it cannot satisfy the no-false-negative requirement. This also contradicts the assumption. $\Box$
\end{proof}

Proposition \ref{proposition:optimal} guarantees that no local partial matches will be missed. This is important to avoid false negatives.
Based on Proposition \ref{proposition:optimal}, we can further prove the following proposition, which guarantees that the intermediate results in our method involve the smallest number of vertices and edges.

\begin{proposition}\label{proposition:minimalISNum} Given the same underlying partitioning over RDF graph $G$, the number of involved vertices and edges in the intermediate results (in our approach) is not larger than that in any other partition-based solution.
\end{proposition}

\begin{proof} In Proposition \ref{proposition:optimal}, we prove that every local partial match should be found for result completeness (i.e., false negatives). The same proposition proves that our method produces complete results. Therefore, if a partition-based solution omits some of the partial matches (i.e., intermediate results) that are in our solution (i.e., has intermediate result smaller than ours) then it cannot produce complete results. Assuming that they all produce complete results, what remains to be proven is that our set of partial matches is a subset of those generated by other partition-based solutions.  We prove that by contradiction.

Let $A$ be a solution generated by an alternative partition-based approach.  Assume that there exists one vertex $u$ in a local partial match $PM$ produced by  our method, but $u$ is not in the intermediate results of the partition-based solution $A$. This would mean that during the assembly phase to produce the final result, any edges adjacent to $u$ will be missed. This would produce incomplete answer, which contradicts the completeness assumption.

Similarly, it can be argued that it is impossible that there exists an \emph{edge} in our local partial matches (i.e., intermediate results) that it is not in the intermediate results of other partition-based approaches.

In other words, all vertices and edges in local partial matches must occur in the intermediate results of other  partition-based approaches. Therefore, Proposition \ref{proposition:minimalISNum} holds. $\Box$
\end{proof}

\nop{says that we cannot miss any local partial match; otherwise, it will lead to false negatives. Let us consider other query decomposition approaches mentioned in Section \ref{sec:introduction}, such as GraphPartition \cite{VLDB11:GraphPartition} and TripleGroup \cite{DBLP:journals/pvldb/LeeL13}. A query is decomposed into several subqueries, each of which is evaluated in all fragments. According to Theorem \ref{theorem:optimal}, it is easy to know that any local partial match (in our approach) $PM$ should be a subquery match in these query-decomposition methods or a join result of some subquery matches. In other words, we guarantee that the number of vertices and edges contained by local partial matches are the smallest, or any other approaches cannot involve fewer vertices or edges to find out all correct matches.}

Finally, we discuss another feature of a local partial match $PM_i$ in fragment $F_i$. Any $PM_i$ cannot be enlarged by introducing more vertices or edges to become a larger local partial match. The following proposition formalizes this.

\begin{proposition}\label{proposition:maximal} Given a query graph $Q$ and an RDF graph $G$, if $PM_i$ is a local partial match under function $f$ in fragment $F_i$, there exists no local partial match $PM^\prime_i$ under function $f^\prime$ in $F_i$, where $f\subset f^\prime$.
\end{proposition}

\begin{proof}(by contradiction) Assume that there exists another local partial match $PM^\prime_i$ of query $Q$ in fragment $F_i$, where $PM_i$ is a subgraph of $PM^\prime_i$. Since $PM_i$ is a subgraph of $PM^\prime_i$, there must exist at least one edge $e=\overrightarrow {uu^\prime}$ where $e\in PM^\prime_i$ and $e \notin PM_i$. Assume that $\overrightarrow{uu^\prime}$ is matching edge $\overrightarrow{vv^{\prime}}$ in query $Q$. Obviously, at least one endpoint of $e$ should be an internal vertex. We assume that $u$ is an internal vertex.  According to Condition 5, edge $\overrightarrow{vv^{\prime}}$ should also be matched in $PM$, since $PM$ is a local partial match.
However, edge $\overrightarrow{uu^\prime}$ (matching $\overrightarrow{vv^{\prime}}$) does not exist in $PM$. This contracts $PM$ being a local partial match. Thus, Proposition \ref{proposition:maximal} holds. $\Box$
\end{proof}

\subsection{Computing Local Partial Matches}\label{sec:computinglocal}

Given a SPARQL query $Q$ and a fragment $F_i$, the goal of partial evaluation is to find all local partial matches (according to Definition \ref{def:localmaximal}) in $F_i$. The matching process consists of determining a function $f$ that
associates vertices of $Q$ with vertices of $F_i$. The matches are expressed as a set of pairs $(v,u)$ ($v \in Q$ and $u \in F_i$). A pair $(v,u)$ represents
the matching of a vertex $v$ of query $Q$ with a vertex $u$ of fragment $F_i$. The set of vertex pairs $(v,u)$ constitutes function $f$ referred to in Definition \ref{def:localmaximal}.

A high-level description of finding local partial matches is outlined in Algorithm \ref{alg:findinglocalmaximal} and Function ComParMatch. According to Conditions 1 and 2 of Definition \ref{def:localmaximal}, each vertex $v$ in query graph $Q$ has a candidate list of vertices in fragment $F_i$. Since function $f$ is as a set of vertex pairs $(v,u)$ ($v \in Q$ and $u \in F_i$), we start with an empty set. In each step, we introduce a candidate vertex pair $(v,u)$ to expand the current function $f$, where vertex $u$ (in fragment $F_i$) is a candidate of vertex $v$ (in query $Q$).

Assume that we introduce a new candidate vertex pair $(v^\prime,u^\prime)$  into the current function $f$ to form another function $f^\prime$. If $f^\prime$ violates any condition except for Conditions 4 and 5 of Definition \ref{def:localmaximal}, the new function $f^\prime$ cannot lead to a local partial match (Lines 6-7 in Function ComParMatch). If $f^\prime$ satisfies all conditions except for Conditions 4 and 5, it means that $f^\prime$ can be further expanded (Lines 8-9 in Function ComParMatch). If $f^\prime$ satisfies all conditions, then $f^\prime$ specifies a local partial match and it is reported (Lines 10-11 in Function ComParMatch).

\begin{algorithm}[h] \label{alg:findinglocalmaximal}

\caption{Computing Local Partial Matches}

\KwIn{A fragment $F_i$ and a query graph $Q$.}
\KwOut{The set of all local maximal partial matches in $F_i$, denoted as $\Omega(F_i)$.}
\For{each vertex $v$ in $Q$}{
\For{each candidate vertex $u$ with regard to $v$}
{
    Initialize a function $f$ with $(v,u)$\\
    Call Function \textbf{ComParMatch}($f$)
}
}
Return $\Omega(F_i)$;
\end{algorithm}

%\vspace{-0.3in}

\begin{function}[h] \label{alg:functioncom}
\small
\caption{ComParMatch($f$)}
\If{all vertices of query $Q$ have been matched in the function $f$}
{
Return;
}
Select an unmatched $v^\prime$ adjacent to a matched vertex $v$ in the function $f$ \\
\For{each candidate vertex $u^\prime$ with regard to $v^\prime$}
{
   $f^\prime$ $\gets$ $f$ $\cup$ $(v^\prime,u^\prime)$ \\
   \If{$f^\prime$ violates any condition (except for condition 4 and 5 of Definition \ref{def:localmaximal})}
   {
      Continue
   }
   \If{$f^\prime$ satisfies all conditions (except for condition 4 and 5 of Definition \ref{def:localmaximal})}
   {
      \textbf{ComParMatch}($f^\prime$)
   }
   \If{$f^\prime$ satisfies all conditions of Definition \ref{def:localmaximal} }
   {
      $f$ specifies a local partial match $PM$ that will be inserted into the answer set $\Omega(F_i)$
   }
 }
\end{function}

At each step, a new candidate vertex pair $(v^\prime,u^\prime)$ is added to an existing function $f$ to form a new function $f^\prime$. The order of selecting the query vertex can be arbitrarily defined. However, QuickSI \cite{VLDB08:QuickSI} proposes several heuristic rules to select an optimized order that can speed up the matching process. These rules are also utilized in our experiments.

To compute local partial matches (Algorithm \ref{alg:findinglocalmaximal}), we revise a graph-based SPARQL query engine, gStore, which is our previous work. Since gStore adopts ``subgraph matching'' technique to answer SPARQL query processing, it is easy to revise its subgraph matching algorithm to find ``local partial matches'' in each fragment. gStore adopts a state transformation technique to find SPARQL matches. Here, a state corresponds to a partial match (i.e. a function from $Q$ to $G$).

Our \emph{state} \emph{transformation} algorithm is as follows. Assume that $u$ matches vertex $v$ in SPARQL query $Q$. We first initialize a state with $u$. Then, we search the RDF data graph for $u$'s neighbor $u^\prime$ corresponding to $v^\prime$ in $Q$, where $v^\prime$ is one of $v$'s neighbors and edge $\overrightarrow{uu^\prime}$ satisfies query edge $\overrightarrow{vv^\prime}$. The search will extend the state step-by-step. The search branch terminates when a state corresponding to a local partial match is found or search cannot continue. In this case, the algorithm backtracks and tries another search branch.

The only change that is required to implement Algorithm \ref{alg:findinglocalmaximal} is in the termination condition (i.e., the final state) so that it stops when a partial match is found rather than looking for a complete match.

%In Appendix \ref{sec:computinglocalnew}, we introduce the algorithm by means of state transformation approach.

\begin{figure}[h]
\begin{center}
    \includegraphics[scale=0.18]{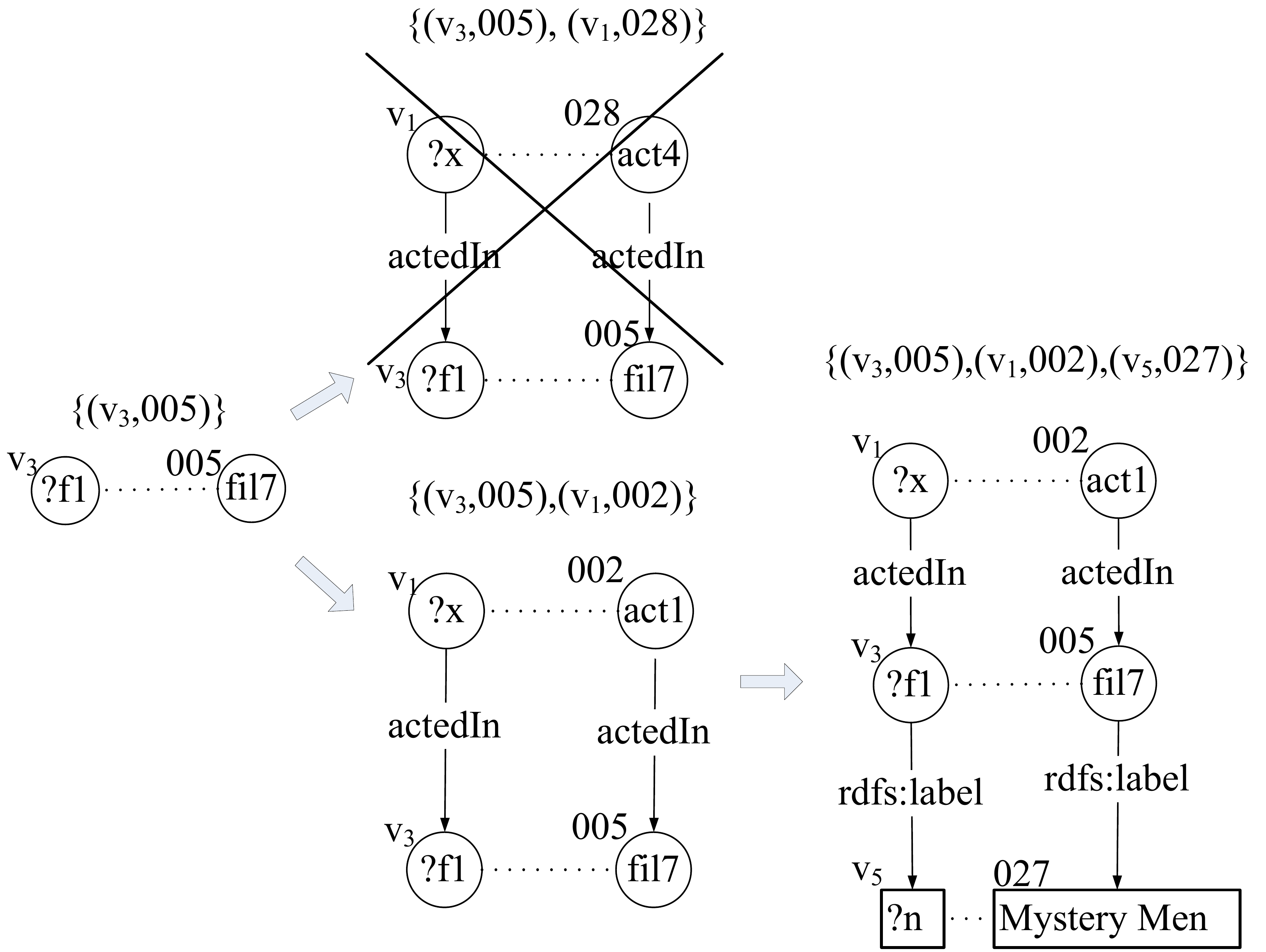}
       \vspace{-0.1in}
   \caption{Finding Local Partial Matches}
   \label{fig:dfslpm}
\end{center}
   \vspace{-0.2in}
\end{figure}

\begin{example}
Figure \ref{fig:dfslpm} shows how to compute $Q$'s local partial matches in fragment $F_1$. Suppose that we initialize a  function $f$ with $(v_3,005)$. In the second step, we expand to $v_1$ and consider $v_1$'s candidates, which are $002$ and $028$. Hence, we introduce two vertex pairs $(v_1,002)$ and $(v_1,028)$ to expand $f$. Similarly, we introduce $(v_5,027)$ into the function $\{(v_3,005),(v_1,002)\}$ in the third step. Then, $\{(v_3,005),(v_1,002),$ $(v_5,027)\}$ satisfies all conditions of Definition \ref{def:localmaximal}, thus, it is a local partial match, and is returned. In another search branch, we check the  function $\{(v_3,005),(v_1,028)\}$, which cannot be expanded, i.e., we cannot introduce a new matching pair;without violating some conditions in Definition \ref{def:localmaximal}. Therefore, this search branch is terminated.
\end{example}

%!TEX root =  distributedgStore.tex
\section{Assembly}\label{sec:assembling}

Each site $S_i$ finds all local partial matches in fragment $F_i$. The next step is to assemble partial matches to compute crossing matches and compute the final results. We propose two assembly strategies: centralized and distributed (or parallel). In centralized, all local partial matches are sent to a single site for assembly. For example, in a client/server system, all local partial matches may be sent to the server. In distributed/parallel, local partial matches are combined at a number of sites in parallel. Here, when $S_i$ sends the local partial matches to the final assembly site for joining, it also tags which vertices in local partial matches are internal vertices or extended vertices of $F_i$. This will be useful for avoiding some computations as discussed in this section.

In Section \ref{sec:basicjoin}, we define a basic join operator for assembly. Then, we propose a centralized assembly algorithm in Section \ref{sec:centralized} using the join operator. In Section \ref{sec:distributed}, we study how to assemble local partial matches in a distributed manner.

\subsection{Join-based Assembly}
\label{sec:basicjoin}

We first define the conditions under which two partial matches are joinable. Obviously, crossing matches can only be formed by assembling partial matches from different fragments. If local partial matches from the same fragment could be assembled, this would result in a larger local partial match in the same fragment, which is contrary to Proposition \ref{proposition:maximal}.

\nop{
We cannot assemble local partial matches from the same fragment to form a crossing match. It is easily proven by the contradiction. Assume that two local partial matches from the same fragment $F_i$ can be assembled to form another larger local partial match in fragment $F_i$. It is contradicted to Theorem \ref{theorem:maximal}, which says that any local partial match is maximal and cannot be enlarged. As a result, the join process only happens between local partial matches from different fragments. The join conditions are defined as follows.
}

\begin{definition}\label{def:joinrelation} \textbf{(Joinable)}
Given a query graph $Q$ and two fragments $F_i$ and $F_j$ ($i \neq j$), let $PM_i$ and $PM_j$ be the corresponding local partial matches over fragments $F_i$ and $F_j$ under functions $f_i$ and $f_j$. $PM_i$ and $PM_j$ are \emph{joinable} if and only if the following conditions hold:

\begin{enumerate}
\item There exist no vertices $u$ and $u^\prime$ in $PM_i$ and $PM_j$, respectively, such that $f^{-1}_i(u)=f^{-1}_j(u^\prime)$.
\item There exists at least one crossing edge $\overrightarrow{uu^\prime}$ such that $u$ is an internal vertex and $u^\prime$ is an extended vertex in $F_i$, while $u$ is an extended vertex and $u^\prime$ is an internal vertex in $F_j$. Furthermore, $f^{-1}_i{(u)}=f^{-1}_j{(u)}$ and $f^{-1}_i{(u^\prime)}=f^{-1}_j{(u^\prime)}$.
\end{enumerate}
\end{definition}

The first condition says that the same query vertex cannot be matched by different vertices in joinable partial matches. The second condition says that two local partial matches share at least one common crossing edge that corresponds to the same query edge.

\begin{example}\label{example:joinable}  Let us recall query $Q$ in Figure \ref{fig:querygraph}.  Figure \ref{fig:assembly} shows two different local partial matches $PM_1^2$ and $PM_2^2$. We also show the functions in Figure \ref{fig:assembly}. There do not exist two different vertices in the two local partial matches that match the same query vertex. Furthermore, they share a common crossing edge $\overrightarrow{002,001}$, where $002$ and $001$ match query vertices $v_2$ and $v_1$ in the two local partial matches, respectively. Hence, they are joinable.
\end{example}

\nop{
\begin{figure}[h]
\begin{center}
\vspace{-0.1in}
    \includegraphics[scale=0.15]{pics/q1_joinable.pdf}
    \vspace{-0.1in}
   \caption{Joinable Local Partial Matches $PM_1^4$ and $PM_3^1$; and Not Joinable Local Partial Matches $PM_1^4$ and $PM_2^1$}
   \vspace{-0.2in}
   \label{fig:q1joinable}
\end{center}
\end{figure}
}

The join result of two joinable local partial matches is defined as follows.

\begin{definition}\textbf{(Join Result)}\label{def:joinresult}
Given a query graph $Q$ and two fragments $F_i$ and $F_j$, $i\neq j$, let $PM_i$ and $PM_j$ be two joinable local partial matches of $Q$ over fragments $F_i$ and $F_j$ under functions $f_i$ and $f_j$, respectively.  The \emph{join result} of $PM_i$ and $PM_j$ is defined under a new function $f$ (denoted as $PM=PM_i \Join_{f} PM_j$), which is defined as follows for any vertex $v$ in $Q$:

\begin{enumerate}
\item if $f_i(v) \neq NULL \wedge f_j(v)=NULL $ \footnote{$f_j(v)=NULL$ means that vertex $v$ in query $Q$ is not matched in local partial match $PM_j$. It is formally defined in Definition \ref{def:localmaximal} condition (2)}, $f(v)$ $\gets$ $ f_i(v)$ \footnote{In this paper, we use ``$\gets$'' to denote the assignment operator.};
\item if $f_i(v) = NULL \wedge f_j(v)\neq NULL $, $f(v)$  $\gets$ $f_j(v)$;
\item if $f_i(v) \neq NULL \wedge f_j(v) \neq NULL $, $f(v)$ $\gets$ $f_i(v)$ (In this case, $f_i(v)=f_j(v)$)
\item if $f_i(v) = NULL \wedge f_j(v) = NULL $, $f(v)$ $\gets$ $NULL$
\end{enumerate}
\end{definition}

Figure \ref{fig:assembly} shows the join result of $PM_1^2 \Join_{f} PM_2^2$.

\nop{It is a crossing match of query $Q$. The function of the join result
is also given in Figure \ref{fig:assembly} . }

\nop{
If the join result has been a complete crossing match, it is returned as a result. Otherwise, the join result is called an intermediate result. For example, joining $PM_1^4$ and $PM_3^1$ (denoted as $PM_1^4 \Join PM_3^1$) results in an intermediate result. Actually, the intermediate result is a local partial match in the merged fragment $F_1 \cup F_3$. For the ease of presentation, we also call an ``intermediate result'' as a ``local partial match'' when the context is clear.
}

\subsection{Centralized Assembly}\label{sec:centralized}
In  centralized assembly, all local partial matches are sent to a final assembly site. We propose an iterative join algorithm (Algorithm 2) to find all crossing matches. In each iteration, a local partial match join with a intermediate result generated in the past iteration. When the join is complete (i.e., a match has been found), the result is returned (Lines 12-13 in Algorithm 2);  otherwise, the result is joined with other local partial matches in the next iteration (Lines 14-15). There are $|V(Q)|$ iterations of Lines 4-16 in the worst case, since at each iteration only a single new matching vertex is introduced (worst case) and $Q$ has $|V(Q)|$ vertices.  If no new intermediate results are generated at some iteration, the algorithm can stop early (Lines 5-6).

\begin{example}\label{example:join} Figure \ref{fig:assembly} shows the join result of $PM_1^2 \Join_{f} PM_2^2$. In this example, we consider a crossing match formed by three local partial matches. Let us consider three local partial matches $PM_1^4$, $PM_4^1$ and $PM_3^1$ in Figure \ref{fig:pmset}. In the first iteration, we obtain the intermediate result $PM_1^4 \Join_{f} PM_3^1$ in Figure \ref{fig:q1join}. Then, in the next iteration, $(PM_1^4 \Join_{f} PM_3^1)$ joins with $PM_4^1$ to obtain a crossing match.
\end{example}

\begin{figure}
\begin{center}
\vspace{-0.1in}
    \includegraphics[scale=0.15]{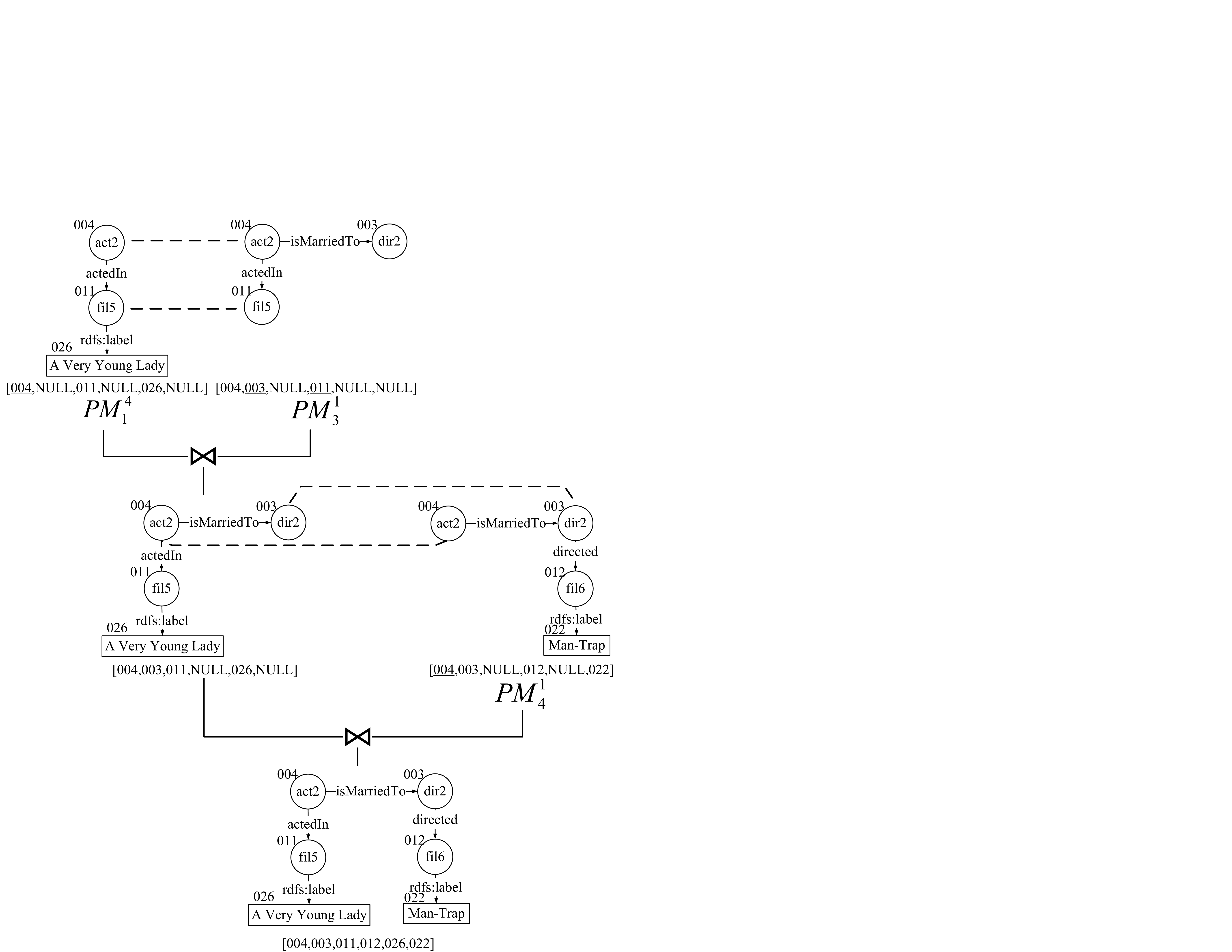}
    \vspace{-0.1in}
   \caption{Joining $PM_1^4$, $PM_3^1$ and $PM_4^1$}
   \vspace{-0.2in}
   \label{fig:q1join}
\end{center}
\end{figure}

\begin{algorithm} \label{alg:naivemerging}
\caption{Centralized Join-based Assembly}

\KwIn{$\Omega(F_i)$, i.e, the set of local partial matches in each fragment $F_i$, $i=1,...,k$}
\KwOut{ All crossing matches set $RS$.}

Each fragment $F_i$ sends the set of local partial matches in each fragment $F_i$ (i.e., $\Omega(F_i)$) to a single site for the assembly\\
Let $\Omega\gets \bigcup\nolimits_{i = 1}^{i = k} {\Omega (F_i )}$\\
Set $MS$ $\gets$ $\Omega$ \\
\While{$MS \ne \emptyset$}
{
    Set $MS^\prime$ $\gets$  $\emptyset$ \\
    \For{each intermediate result $PM$ in $MS$}
    {
      \For{each local partial match $PM^\prime$ in $\Omega$}
      {
      \If{$PM$ and $PM^\prime$ are joinable}
      {
         Set $PM^{\prime\prime}$= $PM \Join PM^\prime$\\
         \If{$PM^{\prime\prime}$ is a complete match of $Q$}
         {
         put $PM^{\prime\prime}$ into $RS$
         }
         \Else
         {
         put $PM^{\prime\prime}$ into $MS^\prime$
         }
       }
      }
    }
    $MS\gets MS^\prime$
}
Return $RS$
\end{algorithm}

\subsubsection{Partitioning-based Join Processing}
\label{sec:optimizedjoin}

The join space in Algorithm 2 is large, since we need to check if every pair of local partial matches $PM_i$ and $PM_j$ are joinable. This subsection proposes an optimized technique to reduce the join space.

The intuition of our method is as follows. We divide all local partial matches into multiple partitions such that two local partial matches in the same set cannot be joinable; we only consider joining local partial matches from different partitions. The following theorem specifies which local partial matches can be put in the same partition.

\begin{theorem}\label{theorem:pmpartition} Given two local partial matches $PM_i$ and $PM_j$ from fragments $F_i$ and $F_j$ with functions $f_i$ and $f_j$, respectively, if there exists a query vertex $v$ where both $f_i(v)$ and $f_j(v)$ are internal vertices of fragments $F_i$ and $F_j$, respectively, $PM_i$ and $PM_j$ are not joinable.
\end{theorem}
\begin{proof}If $f_i(v)\ne f_j(v)$, then a vertex $v$ in query $Q$ matches two different vertices in $PM_i$ and $PM_j$, respectively. Obviously, $PM_i$ and $PM_j$ cannot be joinable.

If $f_i(v)= f_j(v)$, since $f_i(v)$ and $f_j(v)$ are both internal vertices, both $PM_i$ and $PM_j$ are from the same fragment. As mentioned earlier, it is impossible to assemble two local partial matches from the same fragment (see the first paragraph of Section \ref{sec:basicjoin}), thus, $PM_i$ and $PM_j$ cannot be joinable. $\Box$
\end{proof}

\begin{example}
Figure \ref{fig:dividingpmgroup} shows the serialization vectors (defined in Definition \ref{def:localmaximal}) of four local partial matches. For each local partial match, there is an internal vertex that  matches $v_1$ in query graph. The underline indicates the extended vertex in the local partial match. According to Theorem \ref{theorem:pmpartition}, none of them are joinable.
\end{example}

\begin{figure}[h]
\begin{center}
\vspace{-0.1in}
    \includegraphics[scale=0.15]{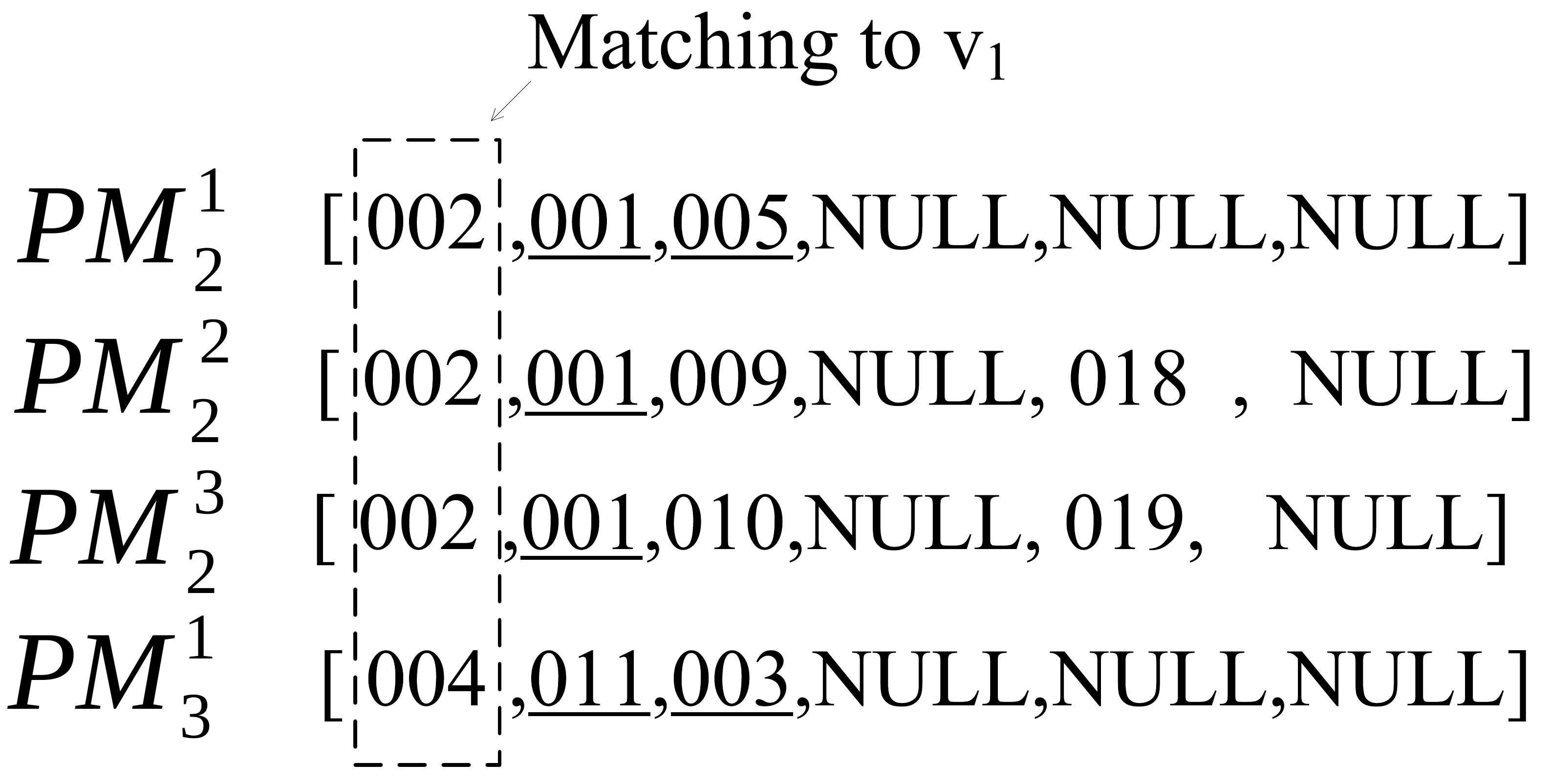}
    \vspace{-0.1in}
   \caption{The Local Partial Match Partition on $v_1$}
   \vspace{-0.2in}
   \label{fig:dividingpmgroup}
\end{center}
\end{figure}

\begin{definition}\label{def:lpmgroup} \textbf{(Local Partial Match Partitioning)}.
Consider a SPARQL query $Q$ with $n$ vertices $\{v_1,...,v_n\}$. Let $\Omega$ denote all local partial matches. $\mathcal{P}=\{P_{v_1},...,P_{v_n}\}$ is a partitioning of $\Omega$ if and only if the following conditions hold.
\begin{enumerate}
\item Each partition $P_{v_i}$ ($i=1,...,n$) consists of a set of local partial matches, each of which has an internal vertex that matches $v_i$.
\item $P_{v_i}  \cap P_{v_j}  = \emptyset$, where $1 \le i \ne j \le n$.
\item $P_{v_1}  \cup ... \cup P_{v_n}  = \Omega$

\end{enumerate}
\end{definition}
\nop{
We refer to each $P_{v_i}$ as a partition.
}

\begin{figure*}%
   \centering
   \subfigure[][]{%
      \includegraphics[scale=0.15]{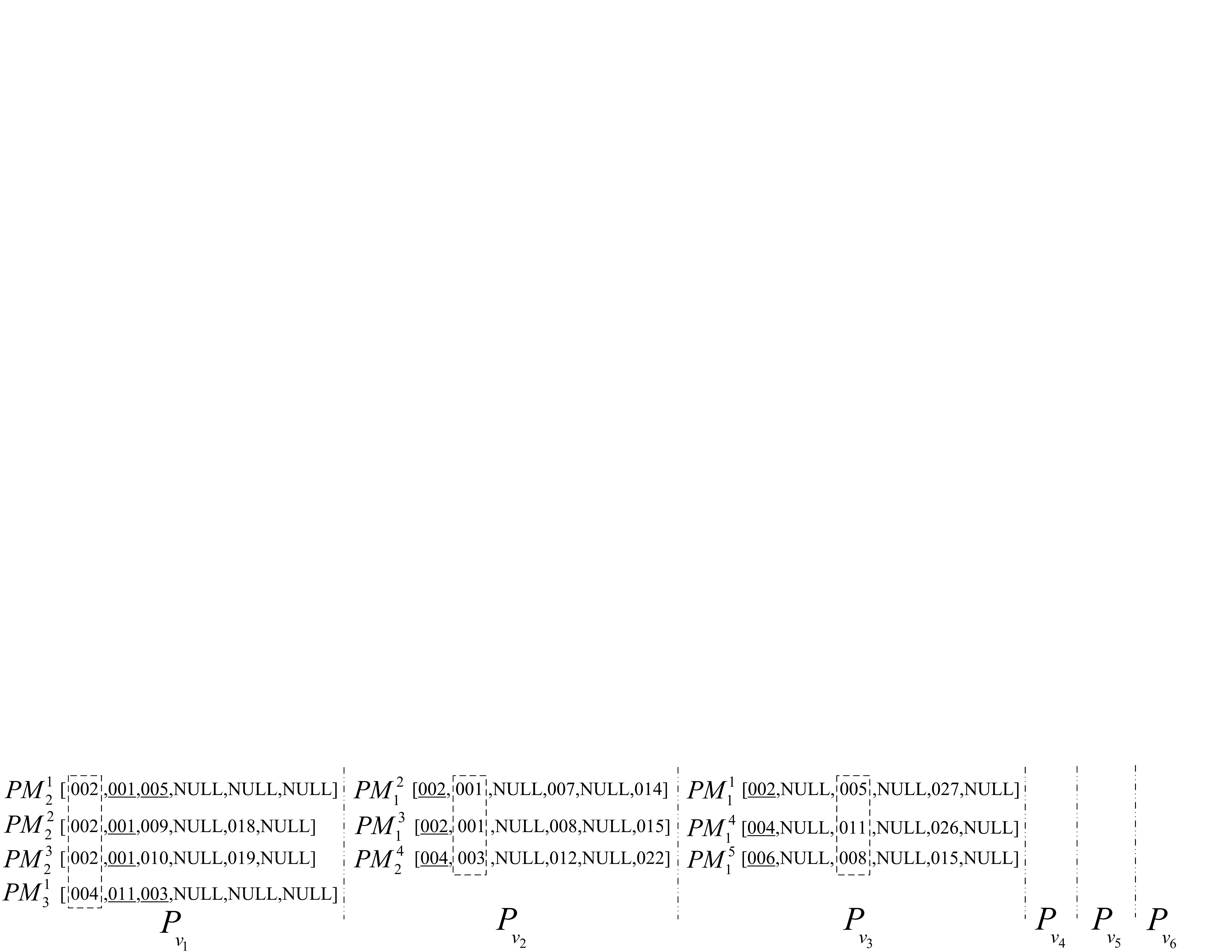}
       \label{fig:groupuset1}%
       }%
       \\
   \subfigure[][]{%
      \includegraphics[scale=0.15]{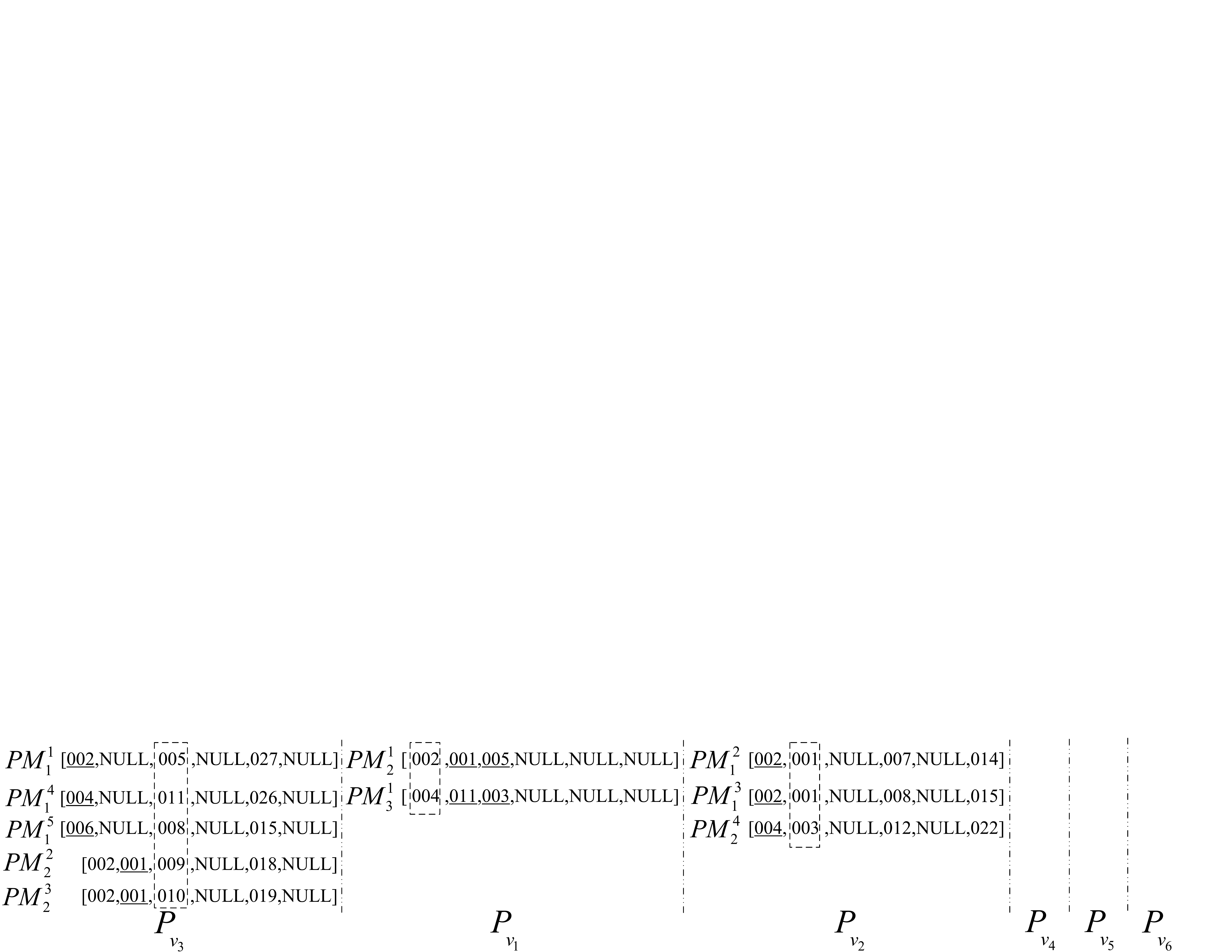}
       \label{fig:groupuset2}%
       }%
 \caption{\small Evaluation of Two Partitionings of Local Partial Matches}%
 \label{fig:allreducedpmgroups}
\end{figure*}

\begin{example} Let us consider all local partial matches of our running example in Figure \ref{fig:pmset}. Figure \ref{fig:allreducedpmgroups} shows two different partitionings.
\end{example}

As mentioned earlier, we only need to consider joining local partial matches from different partitions of $\mathcal{P}$. Given a partitioning $\mathcal{P}=\{P_{v_{1}},...,P_{v_{n}}\}$, Algorithm 3 shows how to perform partitioning-based join of local partial matches. Note that different partitionings and the different join orders in the partitioning will impact the performance of Algorithm 3. In Algorithm 3, we assume that the partitioning $\mathcal{P}=\{P_{v_{1}},...,P_{v_{n}}\}$ is given, and that the join order is from $P_{v_1}$ to $P_{v_n}$, i.e. the order in $\mathcal{P}$. Choosing a good partitioning and the optimal join order will be discussed in Sections \ref{sec:optimalpartition} and  \ref{sec:analysis}.

\begin{algorithm} \label{alg:advancedmerging}
\caption{Partitioning-based Joining Local Partial Matches}

\KwIn{A partitioning $\mathcal{P}=\{P_{v_{1}},...,P_{v_{n}}\}$ of all local partial matches.}
\KwOut{ All crossing matches set $RS$.}
$MS \gets P_{v_{1}}$ \\
\For{$i \gets 2$ to $n$}
{
$MS^\prime \gets \emptyset$\\
\For{each intermediate result $PM$ in $MS$}
{
\For{each local partial match $PM^\prime$ in $P_{v_i}$}
     {
               \If{$PM$ and $PM^\prime$ are joinable}
      {
         Set $PM^{\prime\prime}$ $\gets$ $PM \Join PM^\prime$\\
          \If{$PM^{\prime\prime}$ is a complete match}
         {
           Put $PM^{\prime\prime}$ into the answer set $RS$ \\
         }
         \Else
         {
           Put $PM^{\prime\prime}$ into $MS$
         }

      }
      Put $PM^{\prime}$ into $MS$
     }
 }

 Insert $MS$ into $MS$
 }
Return $RS$
\end{algorithm}

The basic idea of Algorithm 3 is to iterate the join process on each partition of $\mathcal{P}$. First, we set $MS$ $\gets$ $P_{v_{1}}$ (Line 1 in Algorithm 3). Then, we try to join local partial match $PM$ in $MS$ with local partial match $PM^{\prime}$ in $P_{v_{2}}$ (the first loop of Line 3-13). If the join result is a complete match, it is inserted into the answer set $RS$ (Lines 8-9). If the join result is an intermediate result, we insert it into a temporary set $MS^{\prime}$ (Lines 10-11). We also need to insert $PM^{\prime}$ into $MS^{\prime}$, since the local partial match $PM^{\prime}$ (in $P_{v_{2}}$) will join local partial matches in the later partition of $\mathcal{P}$ (Line 12). At the end of the iteration, we insert all intermediate results (in $MS^{\prime}$) into $MS$, which will join local partial matches in the later partition of $\mathcal{P}$ in the next iterations (Line 13). We iterate the above steps for each partition of $\mathcal{P}$ in the partitioning (Lines 3-13).

\subsubsection{Finding the Optimal Partitioning}\label{sec:optimalpartition}

Obviously, given a set $\Omega$ of local partial matches, there may be multiple feasible local partial match partitionings, each of which leads to a different join performances. In this subsection, we discuss how to find the ``optimal'' local partial match partitioning over $\Omega$, which can minimize the joining time of Algorithm 3.

First, there is a need for a measure that would define more precisely the \emph{join cost} for a local partial match partitioning. We define it as follows.

\begin{definition}\textbf{(Join Cost)}.\label{def:joincost} Given a query graph $Q$ with $n$ vertices $v_1$,...,$v_n$ and a partitioning $\mathcal{P}=\{P_{v_1}$,...,$P_{v_n}\}$ over all local partial matches $\Omega$, the join cost is
\begin{equation}\label{equ:joincost}
Cost({\Omega}) = O(\prod\nolimits_{i = 1}^{i = n} {(|P_{v_i} |}  + 1))
\end{equation}
where $|P_{v_i}|$ is the number of local partial matches in $P_{v_i}$ and 1 is introduced to avoid the ``0'' element in the product.
\end{definition}

Definition \ref{def:joincost} assumes that each pair of local partial matches (from different partitions of $\mathcal{P}$) are joinable so that we can quantify the worst-case performance. Naturally, more sophisticated and more realistic cost functions can be used instead, but, finding the most appropriate cost function is a major research issue in itself and outside the scope of this paper.

\begin{example} The cost of the partitioning in Figure \ref{fig:groupuset1} is $5\times 4 \times 4=80$, while that of Figure \ref{fig:groupuset2} is $6\times 3 \times 4=72$. Hence, the partitioning in Figure \ref{fig:groupuset2} has lower join cost.
\end{example}

Based on the definition of join cost, the ``optimal'' local partial match partitioning is one with the minimal join cost. We formally define the \emph{optimal partitioning} as follows.

\begin{definition} \textbf{(Optimal Partitioning)}. Given a partitioning $\mathcal{P}$ over all local partial matches $\Omega$, $\mathcal{P}$ is the optimal partitioning if and only if there exists no another partitioning that has smaller join cost.
\end{definition}

Unfortunately, Theorem \ref{theorem:npc} shows that finding the optimal partitioning is NP-complete.

\begin{theorem} \label{theorem:npc} Finding the optimal partitioning is NP-complete problem.\end{theorem}
\begin{proof} We can reduce a 0-1 integer planning problem to finding the optimal partitioning.
We build a bipartite graph $B$, which contains two vertex groups $B_1$ and $B_2$. Each vertex $a_j$ in $B_1$ corresponds to a local partial match $PM_j$ in $\Omega$, $j=1,...,|\Omega|$. Each vertex $b_i$ in $B_2$ corresponds to a query vertex $v_i$, $i=0,...,n$. We introduce an edge between $a_j$ and $b_i$ if and only if $PM_j$ has a internal vertex that is matching query vertex $b_i$. Let a variable $x_{ji}$ denote the edge label of the edge $\overline{a_jb_i}$. Figure \ref{fig:PMBipartiteGraph} shows an example bipartite graph of all local partial matches in Figure \ref{fig:pmset}.

\begin{figure}[h]
\begin{center}
    \includegraphics[scale=0.4]{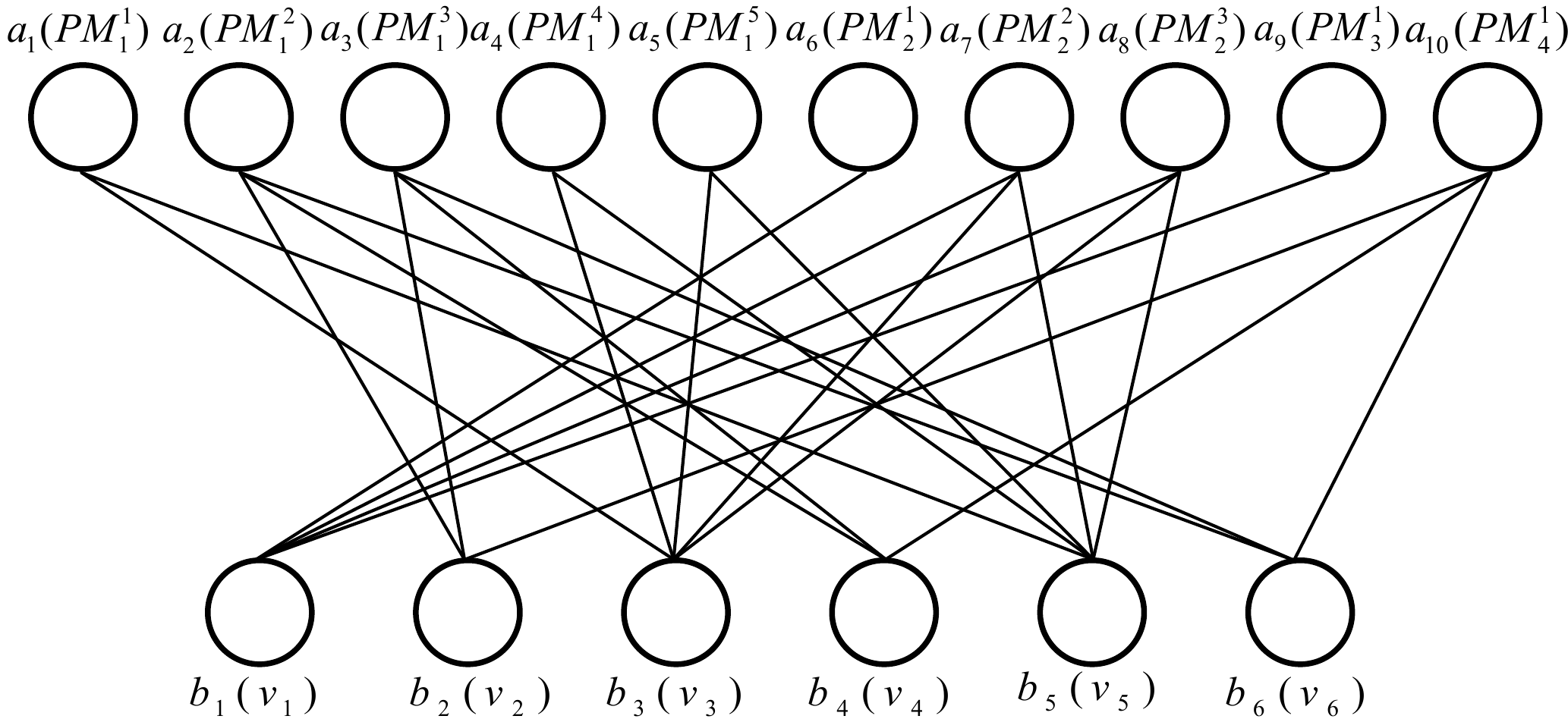}
   \caption{Example Bipartite Graph}
   \vspace{-0.3in}
   \label{fig:PMBipartiteGraph}
\end{center}
\end{figure}

We formulate the 0-1 integer planning problem as follows:
\[
\begin{array}{l}
 \min {\rm{ }}\prod\nolimits_{{\rm{i = 0}}}^{{\rm{i = n}}} {(\sum\nolimits_j {x_{ji} }  + 1)}  \\
 st.\forall j,\sum\nolimits_i {x_{ji} }  = 1 \\
 \end{array}
\]
The above equation means that each local partial match should be assigned to only one query vertex.

The equivalence between the 0-1 integer planning and finding the optimal partitioning is straightforward. The former is a classical NP-complete problem. Thus, the theorem holds.
$\Box$
\end{proof}

Although finding the optimal partitioning is NP-complete (see Theorem \ref{theorem:npc}), in this work, we propose an algorithm with time complexity $(2^n\times |\Omega|)$, where $n$ (i.e., $|V(Q)|$) is small in practice. Theoretically, this algorithm is called fixed-parameter tractable \cite{DBLP:journals/siamcomp/DowneyFVW99} \footnote{An algorithm is called fixed-parameter tractable for a problem of size $l$, with respect to a parameter $n$, if it can be solved in time $O(f(n)g(l))$, where $f(n)$ can be any function but $g(l)$ must be polynomial \cite{DBLP:journals/siamcomp/DowneyFVW99}.}.

Our algorithm is based on the following feature of optimal partitioning (see Theorem \ref{theorem:optimalpartition}). Consider a query graph $Q$ with $n$ vertices $v_1$,...,$v_n$. Let $U_{v_i}$ ($i=1,...,n$) denote all local partial matches (in $\Omega$) that have internal vertices matching $v_i$. Unlike the partitioning defined in Definition \ref{def:lpmgroup}, $U_{v_i}$ and $U_{v_j}$ ($1 \leq i\neq j \leq n$) may have overlaps. For example, $PM_2^3$ (in Figure \ref{fig:fullgrouppm}) contains an internal vertex 002 that matches $v_1$, thus, $PM_2^3$ is in $U_{v_1}$. $PM_2^3$ also has internal vertex 010 that matches $v_3$, thus, $PM_2^3$ is also in $U_{v_3}$. However, the partitioning defined in Definition \ref{def:lpmgroup} does not allow overlapping among partitions of $\mathcal{P}$.

\begin{figure}[h]
\begin{center}
\vspace{-0.1in}
    \includegraphics[scale=0.135]{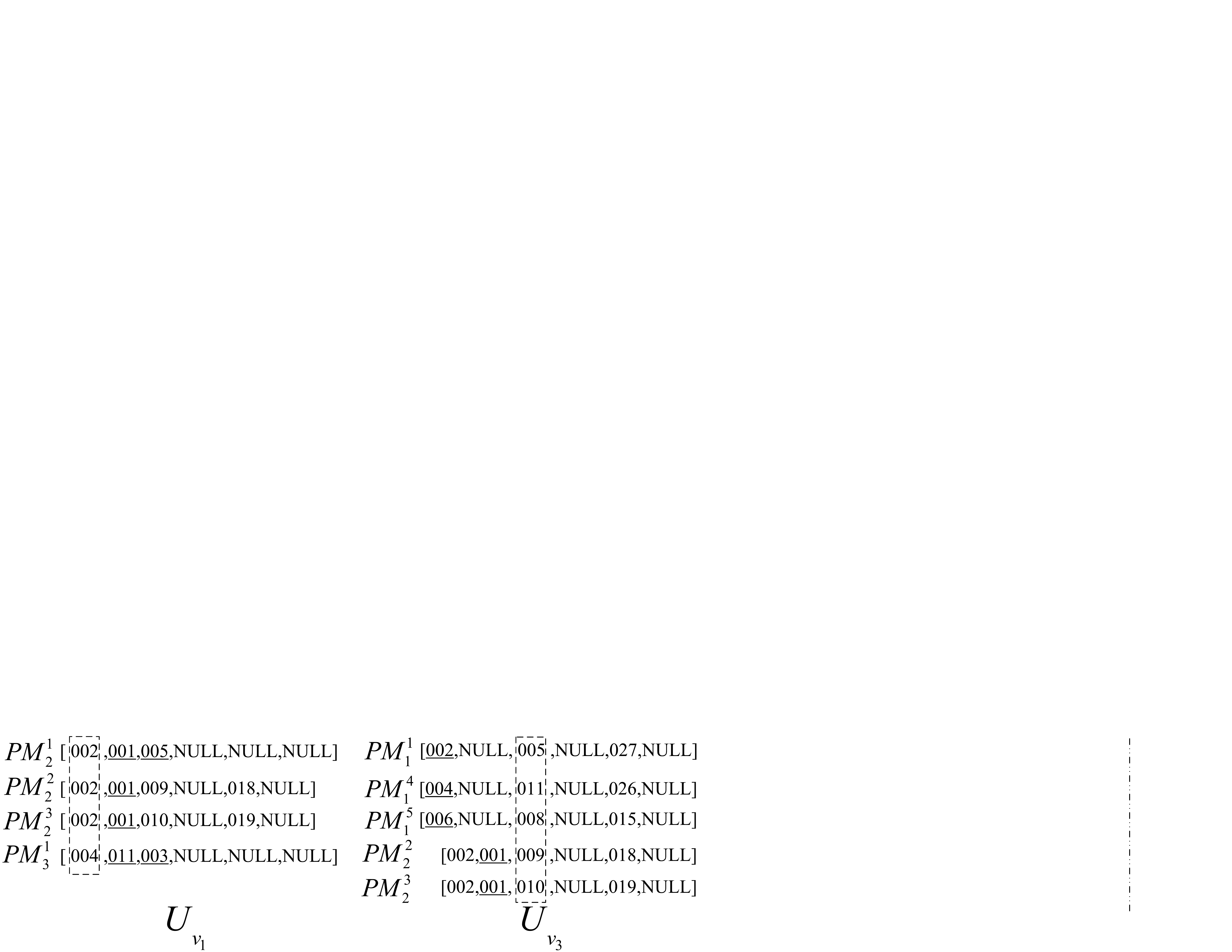}
    \vspace{-0.1in}
   \caption{$U_{v_1}$ and $U_{v_3}$}
   \vspace{-0.2in}
   \label{fig:fullgrouppm}
\end{center}
\end{figure}

%If a partitioning $\mathcal{P}_{opt}=\{P_{v_1},..,P_{v_n}\}$ is optimal, the following theorem holds.

\begin{theorem}\label{theorem:optimalpartition} Given a query graph $Q$ with $n$ vertices $\{v_1$,...,$v_n\}$ and a set of all local partial matches $\Omega$, let $U_{v_i}$ ($i=1,...,n$) be all local partial matches (in $\Omega$) that have internal vertices matching $v_i$. For the optimal partitioning $\mathcal{P}_{opt}=\{P_{v_1},..,P_{v_n}\}$ where $P_{v_n}$ has the largest size (i.e., the number of local partial matches in $P_{v_n}$ is maximum) in $\mathcal{P}_{opt}$, $P_{v_n} = U_{v_n}$.
\end{theorem}

\begin{proof}
(by contradiction) Assume that $P_{v_n} \ne U_{v_n}$ in the optimal partitioning $\mathcal{P}_{opt}=\{P_{v_1},..,P_{v_n}\}$. Then, there exists a local partial match $PM\notin P_{v_n}$ and $PM\in U_{v_n}$. We assume that $PM\in P_{v_j}$, $j \neq n$. The cost of $\mathcal{P}_{opt}=\{P_{v_1},..,P_{v_n}\}$ is:

\begin{equation}\label{equ:1}
Cost({\Omega})_{opt}=(\prod_{1\le i< n \land i\ne j}(|P_{v_i}|+1)) \times (|P_{v_j}|+1) \times (|P_{v_n}|+1)
\end{equation}

Since $PM\in U_{v_n}$, $PM$ has an internal vertex matching $v_n$. Hence, we can also put $PM$ into $P_{v_n}$. Then, we get a new partitioning $\mathcal{P}^\prime=\{P_{v_1},...,P_{v_j}- \{PM\},...,,P_{v_n}\cup \{PM\}\}$. The cost of the new partitioning is:

\begin{equation}\label{equ:2}
Cost(\Omega) =(\prod_{1\le i< n \land i\ne j}(|P_{v_i}|+1)) \times |P_{v_j}| \times (|P_{v_n}|+2)
\end{equation}

Let $C=\prod_{1\le i< n \land i\ne j}(|P_{v_i}|+1)$, which exists in both Equations \ref{equ:1} and \ref{equ:2}. Obviously, $C >0$.

\[
\begin{array}{l}
 Cost(\Omega)_{opt} - Cost(\Omega) \\
  = C \times (|P_{v_n } | + 1) \times (|P_{v_j } | + 1) - C \times (|P_{v_n } | + 2) \times (|P_{v_j } |) \\
  = C \times (|P_{v_n } | + 1 - |P_{v_j } |) \\
 \end{array}
\]

Because ${P}_{v_n}$ is the largest partition in $\mathcal{P}_{opt}$, $|P_{v_n}| + 1 - |P_{v_j}|>0$. Furthermore, $C >0$. Hence, $Cost(\Omega)_{opt}-Cost({\Omega})>0$, meaning that the optimal partitioning has larger cost. Obviously, this cannot happen.

Therefore, in the optimal partitioning $\mathcal{P}_{opt}$, we cannot find a local partial match $PM$, where $|P_{v_n}|$ is the largest, $PM \notin P_{v_n}$ and $PM \in U_{v_n}$.  In other words, $P_{v_n}=U_{v_n}$ in the optimal partitioning. $\Box$
\end{proof}

Let $\Omega$ denote all local partial matches. Assume that the optimal partitioning is $\mathcal{P}_{opt}=\{P_{v_1},P_{v_2},...,P_{v_n}\}$. We re-order the partitions of $\mathcal{P}_{opt}$ in non-descending order of sizes, i.e., $\mathcal{P}_{opt}=\{P_{v_{k_1}},...,P_{v_{k_n}}\}$, $|P_{v_{k_1}}| \ge |P_{v_{k_2}}| \ge ... \ge |P_{v_{k_n}}|$. According to Theorem \ref{theorem:optimalpartition}, we can conclude that  $P_{v_{k_1}}=U_{v_{k_1}}$ in the optimal partitioning $\mathcal{P}_{opt}$.

Let  $\Omega_{\overline{v_{k_1}}}=\Omega - U_{v_{k_1}}$, i.e.,  the set of local partial matches excluding the ones with an internal vertex matching $v_{k_1}$. It is straightforward to know $Cost(\Omega)_{opt}=|P_{v_{k_1}}| \times Cost(\Omega_{\overline{v_{k_1}}})_{opt}=|U_{v_{k_1}}| \times Cost(\Omega_{\overline{v_{k_1}}})_{opt}$. In the optimal partitioning over $\Omega_{\overline{v_{k_1}}}$, we assume that $P_{v_{k_2}}$ has the largest size. Iteratively, according to Theorem \ref{theorem:optimalpartition}, we know that $P_{v_{k_2} } = U^{\prime}_{v_{k_2}}$, where $U^{\prime}_{v_{k_2}}$ denotes the set of local partial matches with an internal vertex matching $v_{k_2}$ in $\Omega_{\overline{v_{k_1}}}$.

According to the above analysis, if a vertex order is given, the partitioning over $\Omega$ is fixed. Assume that the optimal vertex order that leads to minimum join cost is given as $\{v_{k_1},..., v_{k_n}\}$. The partitioning algorithm work as follows.

Let $U_{v_{k_1}}$ denote all local partial matches (in $\Omega$) that have internal vertices matching vertex $v_{k_1}$\footnote{When we find local partial matches in fragment $F_i$ and send them to join, we tag which vertices in local partial matches are internal vertices of $F_i$.}. Obviously, $U_{v_{k_1}}$ is fixed if $\Omega$ and the vertex order is given. We set $P_{v_{k_1}}=U_{v_{k_1}}$.  In the second iteration, we remove all local partial matches in $U_{v_{k_1}}$ from $\Omega_{\overline{v_{k_1}}}$, i.e,  $\Omega_{\overline{v_{k_1}}}=\Omega - U_{v_{k_1}}$. We set $U_{v_{k_2}}^{\prime}$ to be all local partial matches (in $\Omega_{\overline{v_{k_1}}}$) that have internal vertices matching vertex $v_{k_2}$. Then, we set $P_{v_{k_2}}=U_{v_{k_2}}^{\prime}$. Iteratively, we can obtain $P_{v_{k_3}},..., P_{v_{k_n}}$.

\begin{example}
Consider all local partial matches in Figure \ref{fig:LPMPartitioning}. Assume that the optimal vertex order is $\{v_3,v_1,v_2\}$. We will discuss how to find the optimal order later. In the first iteration, we set $P_{v_{3}}=U_{v_{3}}$, which contains five matches. For example, $PM_1^1=[\underline{002}\footnote{We underline all extended vertices in serialization vectors.},NULL,005,NULL,$ $027,NULL]$ is in $U_{v_3}$, since internal vertex $005$ matches $v_3$. In the second iteration, we set $\Omega_{\overline{v_3}}=\Omega-P_{v_{3}}$. Let $U_{v_{1}}^{\prime}$ to be all local partial matches in $\Omega_{\overline{v_3}}$ that have internal vertices matching vertex $v_{1}$. Then, we set $P_{v_{1}}=U_{v_{1}}^{\prime}$. Iteratively, we can obtain the partitioning $\{P_{v_{3}},P_{v_{1}}, P_{v_{2}}\}$, as shown in Figure \ref{fig:LPMPartitioning}.
\end{example}
\vspace{0.1in}
Therefore, the challenging problem is how to find the optimal vertex order $\{v_{k_1},...,v_{k_n}\}$. Let us denote by $\Omega_{\overline{v_{k_1} } }$ all local partial matches (in $\Omega$) that do not contain internal vertices matching $v_{k_1}$, i.e., $\Omega_{\overline{v_{k_1} } }= \Omega - U_{v_{k_1}}$.
It is straightforward to have the following \emph{optimal substructure}\footnote{ A problem is said to have \emph{optimal substructure} if an optimal solution can be constructed efficiently from optimal solutions of its subproblems \cite{DBLP:books/daglib/0023376}. This property is often used in dynamic programming formulations. } in Equation \ref{equ:optimalsubstructure}.

\nop{
Based on Theorem \ref{theorem:optimalpartition}, we propose a dynamic programming algorithm (Algorithm 4 with Function ComCost). The intuition of the algorithm is as follows. Assume that the optimal partitioning is $\mathcal{P}_{opt}=\{P_{v_1},P_{v_2},...,P_{v_n}\}$. We re-order the partitions of $\mathcal{P}_{opt}$ in non-descending order of sizes, i.e., $\mathcal{P}_{opt}=\{P_{v_{k_1}},...,P_{v_{k_n}}\}$, $|P_{v_{k_1}}| \ge |P_{v_{k_2}}| \ge ... \ge |P_{v_{k_n}}|$. According to Theorem \ref{theorem:optimalpartition}, $P_{v_{k_1}}=U_{v_{k_1}}$. Let us denote by $\Omega_{\overline{v_{k_1} } }$ all local partial matches (in $\Omega$) that do not contain internal vertices matching $v_{k_1}$, i.e., $\Omega_{\overline{v_{k_1} } }= \Omega - U_{v_{k_1}}$. }

\begin{equation}\label{equ:optimalsubstructure}
\begin{array}{l}
 Cost(\Omega )_{opt}  = |P_{v_{k_1 } } | \times Cost(\Omega _{\overline {v_{k_1 } } } )_{opt}  \\
  = |U_{v_{k_1 } } | \times Cost(\Omega _{\overline {v_{k_1 } } } )_{opt}  \\
 \end{array}
\end{equation}

Since we do not know which vertex is $v_{k_1}$, we introduce the following optimal structure that is used in our dynamic programming algorithm (Lines 3-7 in Algorithm 4 ).

\begin{equation}\label{equ:optimalsubstructure1}
\begin{array}{l}
 Cost(\Omega )_{opt}  = MIN_{1 \le i \le n} (|P_{v_i } | \times Cost(\Omega _{\overline {v_i } } )_{opt} ) \\
  = MIN_{1 \le i \le n} (|U_{v_i } | \times Cost(\Omega _{\overline {v_i } } )_{opt} ) \\
 \end{array}
\end{equation}

Obviously, it is easy to design a naive dynamic algorithm based on Equation \ref{equ:optimalsubstructure1}.  However, it can be further optimized by recording some intermediate results. Based on Equation \ref{equ:optimalsubstructure1}, we can prove the following equation.

\begin{equation}\label{equ:optimalsubstructuremultilevel}
\begin{array}{l}
 Cost(\Omega )_{opt}  = MIN_{1 \le i \le n;1 \le j \le n;i \ne j} (|P_{v_i } | \times |P_{v_j } | \times Cost(\Omega _{\overline {v_i v_j } } )_{opt} ) \\
  = MIN_{1 \le i \le n;1 \le j \le n;i \ne j} (|U_{v_i } | \times |U_{v_j }^{\prime}| \times Cost(\Omega _{\overline {v_i v_j } } )_{opt} ) \\
 \end{array}
\end{equation}
where $\Omega _{\overline {v_i v_j } }$ denotes all local partial matches that do not contain internal vertices matching $v_i$ or $v_j$, and $U_{v_j}^\prime$ denotes all local partial matches (in $\Omega_{\overline{v_i}}$) that contain internal vertices matching vertex $v_j$.

However, if Equation \ref{equ:optimalsubstructuremultilevel} is used naively in the dynamic programming formulation, it would result in repeated computations. For example, $Cost(\Omega _{\overline {v_1 v_2 } } )_{opt} $ will be computed twice in both $|U_{v_1} | \times |U_{v_2 }^{\prime}| \times Cost(\Omega _{\overline {v_1 v_2 } } )_{opt}$ and $
|U_{v_2 } | \times |U_{v_1 }^{\prime}| \times Cost(\Omega _{\overline {v_1 v_2 } } )_{opt} $. To avoid this, we introduce a map that records  $Cost(\Omega^\prime)$ that is already calculated (Line 16 in Function OptComCost), so that subsequent uses of $Cost(\Omega^\prime)$ can be serviced directly by searching the map (Lines 8-10 in Function ComCost).

\begin{figure*}
   \centering
   \includegraphics[scale=0.32]{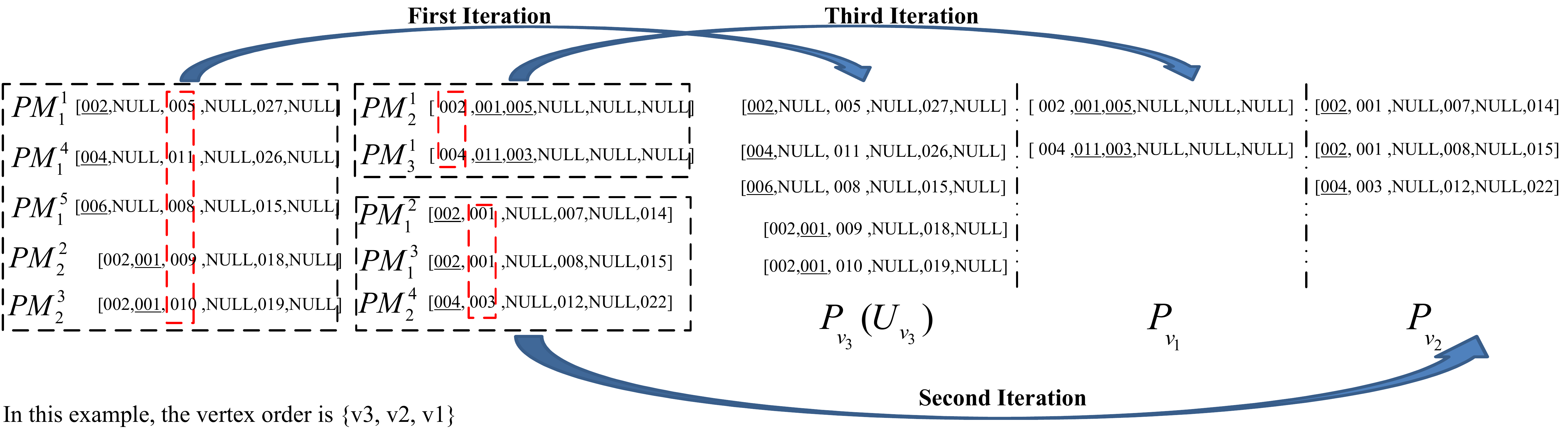}
      \vspace{-0.25in}
 \caption{Example of Partitioning Local Partial Matches}%
 \label{fig:LPMPartitioning}
\end{figure*}

We can prove that there are $
\Huge\sum_{i=1}^{n}\normalsize \left( {\begin{array}{c}
   n  \\
   i \\
\end{array}} \right) =2^n
$  items in the map (worst case), where $n=|V(Q)|$. Thus, the time complexity of the algorithm is $(2^n\times |\Omega|)$. Since $n$ (i.e., $|V(Q)|$) is small in practice, this algorithm is fixed-parameter tractable.

\begin{algorithm}[t] \label{alg:findingoptimalpartition}
\caption{Finding the Optimal Partitioning}
\small
%\begin{algorithmic}
 %\begin{algorithmic}

\KwIn{All local partial matches $\Omega$}
\KwOut{The Optimal Partitioning $\mathcal{P} _{opt}$ and $Cost_{opt}(\Omega)$}
$minID \gets \Phi $\\
$Cost(\Omega)_{opt} \gets \infty $\\
\For{$i=1$ to $n$}
{
$Cost_{opt}(\Omega - U_{v_i})$, $ \mathcal{P}_i^\prime $  $\gets$ $\textbf{ComCost}(\Omega - U_{v_i}, \{v_i\} )$
 \tcc{Call Function ComCost, $U_{v_i}^{\prime}$ denotes all local partial matches (in $\Omega^{\prime}$) that have vertices match $v_i$.}

\If{$Cost(\Omega)_{opt} > |U_{v_i}| \times Cost_{opt}(\Omega - U_{v_i})$}
  {
       $Cost(\Omega)_{opt} \gets |U_{v_i}| \times Cost_{opt}(\Omega - U_{v_i})$ \\
              $minID=i $\\
  }
}
$\mathcal{P}_{opt}$  $\gets$  $\{U_{v_{minID}}\} $ $\bigcup$  $\mathcal{P}^\prime_{minID}$\\
Return $\mathcal{P}_{opt}$
%\end{algorithmic}

\end{algorithm}

\begin{function}[t] \label{alg:functioncomcost}
\small
\caption{ComCost($\Omega^\prime, W$)}
\KwIn{local partial match set $\Omega^{\prime}$ and a set $W$ of vertices that have been used }
\KwOut{The Optimal Partitioning $\mathcal{P}^{\prime} _{opt}$ over $\Omega^\prime$ and $Cost_{opt}(\Omega^{\prime})$}
\If{$\Omega^{\prime}  =\Phi$}
{
Return  $\mathcal{P}^{\prime}_{opt} \gets \Phi$ ;   $Cost(\Omega^\prime)_{opt} \gets 1$ ;
}
\Else{
$minID \gets \Phi $\\
$Cost(\Omega^\prime)_{opt} \gets \infty $\\
\For{$i=1$ to $n$}
{
\If{$v_i   \notin W$ }
{
\If{$MAP$ consists the key ($\Omega^\prime$-$U^{\prime}_{v_i})$ }
{
\tcc{if $Cost(\Omega^\prime$-$U^{\prime}_{v_i}))_{opt}$ has been calculated before}
$Cost_{opt}(\Omega ^{\prime}- U^{\prime}_{v_i})$, $ \mathcal{P}_i^\prime $  $\gets$ $MAP[(\Omega^\prime$-$U^{\prime}_{v_i})]$ \\
 \tcc{finding the corresponding join cost and the optimal partitioning over $(\Omega^\prime$-$U^{\prime}_{v_i} )$ from the map}
}
\Else
{
$Cost_{opt}(\Omega ^{\prime}- U^{\prime}_{v_i})$, $ \mathcal{P}_i^\prime $  $\gets$ $\textbf{ComCost}(\Omega^{\prime} - U^{\prime}_{v_i}, W \cup \{v_i\} )$ \\
 \tcc{Call Function ComCost, $U_{v_i}^{\prime}$ denotes all local partial matches (in $\Omega^{\prime}$) that have vertices match $v_i$.}

}

\If{$Cost(\Omega^\prime)_{opt} > |U^{\prime}_{v_i}| \times Cost_{opt}(\Omega^{\prime} - U^{\prime}_{v_i})$}
  {
       $Cost(\Omega^{\prime})_{opt} \gets |U_{v_i}| \times Cost_{opt}(\Omega^{\prime} - U^{\prime}_{v_i})$ \\
       $minID=i $\\
  }
  }
}
$\mathcal{P}^{\prime}_{opt}$  $\gets$  $\{U_{v_{minID}}\} $ $\bigcup$  $\mathcal{P}^\prime_{minID}$\\
Insert (key=$\Omega^\prime$, value=($Cost(\Omega^\prime)_{opt}$, $\mathcal{P}^{\prime}_{opt}$) ) into the $MAP$. \\
Return  $Cost(\Omega^\prime)_{opt}$ and $\mathcal{P}^{\prime}_{opt}$
}
\end{function}

\nop{
\begin{example}
Figure \ref{fig:LPMPartitioning} shows how to find the optimal local partial match partitioning for all local partial matches in Figure \ref{fig:pmset}. In the first iteration, $U_{v_{3}}$, which contains all five matches having internal vertices matching vertex $v_3$, is the largest among $U_{v_{i}}$ ($1\le i\le 6$). Hence, we set $P_{v_{3}}=U_{v_{3}}$. In the second iteration, we set $\Omega^{\prime}=\Omega-P_{v_{3}}$ and let $U_{v_{i}}^{\prime}$ to be all local partial matches in $\Omega^{\prime}$ that have internal vertices matching vertex $v_{i}$. We find that $U_{v_{2}}^{\prime}$ is the largest. Then, we set $P_{v_{2}}=U_{v_{2}}^{\prime}$. Iteratively, we can obtain the partitioning $\{P_{v_{3}},P_{v_{1}}, P_{v_{2}}\}$, as shown in Figure \ref{fig:LPMPartitioning}.
\end{example}
}

\subsubsection{Join Order}\label{sec:analysis}

When we determine the optimal partitioning of local partial matches, the join order is also determined. If the optimal partitioning is $\mathcal{P}_{opt}=\{P_{v_{k_1}},...,P_{v_{k_n}}\}$ and $|P_{v_{k_1}}| \ge |P_{v_{k_2}}| \ge ... \ge |P_{v_{k_n}}|$, then the join order must be $P_{v_{k_1}} \Join P_{v_{k_2}} \Join ... \Join P_{v_{k_n}}$. The reasons are as follows.

First, changing the join order may not prune any intermediate results. Let us recall the example optimal partitioning $\{P_{v_{3}}, P_{v_{2}}, P_{v_{1}}\}$ shown in Figure \ref{fig:groupuset2}. The join order should be $P_{v_{3}} \Join P_{v_{2}} \Join P_{v_{1}}$, and any changes in the join order would not prune intermediate results. For example, if we first join $P_{v_{2}}$ with $P_{v_{1}}$, we can not prune the local partial matches in $P_{v_{2}}$ that can not join with any local partial matches in $P_{v_{1}}$. This is because there may be some local partial matches $P_{v_{3}}$ that have an internal vertex matching
$v_1$ and can join with local partial matches in $P_{v_{2}}$. In other words, not only the results of $P_{v_{2}} \Join P_{v_{1}}$ but also $P_{v_{2}}$ should join with $P_{v_{3}}$. Similarly, we can observe that any other changes of the join order of the partitioning have no effects.

Second, in some special cases, the join order may have an effect on the performance. Given a partitioning $\mathcal{P}_{opt}=\{P_{v_{k_1}},...,P_{v_{k_n}}\}$ and $|P_{v_{k_1}}| \ge |P_{v_{k_2}}| \ge ... \ge |P_{v_{k_n}}|$, if the set of the first $n^\prime$ vertices, $\{v_{k_1}, v_{k_2}, ..., v_{n^\prime}\}$, is a vertex cut of the query graph, the join order for the remaining $n-n^\prime$ partitions of $\mathcal{P}$ has an effect. For example, let us consider the partitioning $\{P_{v_{1}}, P_{v_{3}}, P_{v_{2}}\}$ in Figure \ref{fig:groupuset1}. If the partitioning is optimal, then both joining $P_{v_{1}}$ with $P_{v_{2}}$ first and joining $P_{v_{1}}$ with $P_{v_{3}}$ first can work. However, it is possible for their cost to be different.\footnote{Note that, in this example, their cost values are the same, but they are possible to be different.} In the worst case, if the query graph is a complete graph, the join order has no effect on the performance.

In conclusion, when the optimal partitioning is determined as $\mathcal{P}_{opt}=\{P_{v_{k_1}},...,P_{v_{k_n}}\}$ and $|P_{v_{k_1}}| \ge |P_{v_{k_2}}| \ge ... \ge |P_{v_{k_n}}|$, then the join order must be $P_{v_{k_1}} \Join P_{v_{k_2}} \Join ... \Join P_{v_{k_n}}$. The join cost can be estimated based on the cost function (Definition \ref{def:joincost}).

\subsection{Distributed Assembly}
\label{sec:distributed}\vspace{-0.05in}

An alternative to centralized assembly is to assemble the local partial matches in a distributed fashion. We adopt Bulk Synchronous Parallel (BSP) model \cite{DBLP:journals/cacm/Valiant90} to design a synchronous algorithm for distributed assembly. A BSP computation proceeds in a series of global \emph{supersteps}, each of which consists of three components: local computation, communication and barrier synchronisation. In the following, we discuss how we apply this strategy to distributed assembly.

\textbf{Local Computation}. Each processor performs some computation based on the data stored in the local memory. The computations on different processors are independent in the sense that different processors perform the computation in parallel.

Consider the $m$-th superstep. For each fragment $F_i$, let  $\Delta_{in}^{m}(F_i)$ denote all received intermediate results in the $m$-th superstep and $\Omega^{m}(F_i)$ denote all local partial matches in fragment $F_i$. In the $m$-th superstep, we join intermediate results in $\Delta_{in}^{m}(F_i)$ with local partial matches in $\Omega^{m}(F_i)$ by Algorithm 5. For each intermediate result $PM$, we check if it can join with some local partial match $PM^{\prime}$ in $\Omega^{m}(F_i) \bigcup \Delta_{in}^{m}(F_i)$. If the join result $PM^{\prime\prime}=$$PM$ $\Join$ $PM^{\prime}$ is a complete crossing match, it is returned. If the join result $PM^{\prime\prime}$ is an intermediate result, we will check if $PM^{\prime\prime}$ can further join with another local partial match in $\Omega^{m}(F_i) \bigcup \Delta_{in}^{m}(F_i)$ in the next iteration. We also insert the intermediate result $PM^{\prime\prime}$ into $\Delta_{out}^{m}(F_i)$ that will be sent to other fragments in the communication step discussed below. Of course, we can also use the partitioning-based solution (in Section \ref{sec:optimizedjoin}) to optimize join processing, but we do not discuss that due to space limitation.

\begin{algorithm}[h] \label{alg:localcomputation}
\caption{Local Computation in Each Fragment $F_i$}

\KwIn{$\Omega^m(F_i)$, the local partial matches in fragment $F_i$}
\KwOut{$RS$, the crossing matches found at this superstep; $\Delta^m_{out}(F_i)$, the intermediate results that will be sent}

Let $\Omega = \Omega^m(F_i) \cup \Delta_{in}^{m}(F_i)$\\
Set $MS$ = $\Delta_{in}^{m}(F_i)$\\

\For{$N=1$ to $|V(Q)|$}
{
    \If{$|MS|$=0}
    {
        Break;
    }
    Set $MS^\prime$ = $\phi$ \\
    \For{each intermediate result $PM$ in $MS$}
    {
      \For{each local partial match $PM^\prime$ in $\Omega^{m}(F_i) \cup \Delta_{in}^{m}(F_i)$}
      {
      \If{$PM$ and $PM^\prime$ are joinable}
      {
         $PM^{\prime\prime}\gets PM \Join PM^\prime$\\
         \If{$PM^{\prime\prime}$ is a SPARQL match}
         {
           Put $PM^{\prime\prime}$ into the answer set $RS$ \\
         }
         \Else
         {
           Put $PM^{\prime\prime}$ into $MS^\prime$
         }
      }
      }
    }

Insert $MS^\prime$ into $\Delta_{out}^m(F_i)$ \\
Clear $MS$ and $MS\gets MS^\prime$
}
$\Omega^{m+1}(F_i)=\Omega^{m}(F_i)\cup \Delta_{out}^m(F_i)$\\
Return $RS$ and $\Delta_{out}^m(F_i)$
\end{algorithm}

\textbf{Communication.} Processors exchange data among themselves. Consider the $m$-th superstep. A straightforward communication strategy is as follows. If an intermediate result $PM$ in $\Delta_{out}^{m}(F_i)$ shares a crossing edge with fragment $F_j$, $PM$ will be sent to site $S_j$ from $S_i$ (assuming fragments $F_i$ and $F_j$ are stored in sites $S_i$ and $S_j$, respectively).

However, the above communication strategy may generate duplicate results. For example, as shown in Figure \ref{fig:pmset}, we can assemble $PM_1^4$ (at site $S_1$) and $PM_3^1$ (at site $S_3$) to form a complete crossing match. According to the straightforward communication strategy, $PM_1^4$ will be sent to $S_1$ from $S_3$ to produce $PM_1^4 \Join PM_3^1$ at $S_3$. Similarly,  $PM_3^1$ is sent from $S_3$ to $S_1$ to assemble at site $S_1$. In other words, we obtain the join result $PM_1^4 \Join PM_3^1$ at both sites $S_1$ and $S_3$. This wastes resources and increases total evaluation time.

To avoid duplicate result computation, we introduce a ``divide-and-conquer'' approach. We define a \emph{total order} ($\prec$) over fragments $\mathcal{F}$ in a non-descending order of $|\Omega(F_i)|$, i.e., the number of local partial matches in fragment $F_i$ found at the partial evaluation stage.

\vspace{-0.05in}
\begin{definition}\label{def:totalorder} Given any two fragments $F_i$ and $F_j$, $F_i \prec F_j $ if and only if $|\Omega(F_i)| \leq |\Omega(F_j)|$ ($1 \leq i,j \leq n$).
\end{definition}
\vspace{-0.05in}

Without loss of generality, we assume that $F_1 \prec F_2 \prec ...\prec F_n$ in the remainder. The basic idea of the divide-and-conquer approach is as follows. Assume that a crossing match $M$ is formed by joining local partial matches that are from different fragments $F_{i_1}$,...,$F_{i_m}$, where $F_{i_1} \prec F_{i_2} \prec ...\prec F_{i_m}$ ($1 \leq i_1,...,i_m \leq n$). The crossing match should only be generated at fragment site $S_{i_m}$ rather than other fragment sites.

For example, at site $S_2$, we generate crossing matches by joining local partial matches from $F_1$ and $F_2$. The crossing matches generated at $S_2$ should not contain any local partial matches from $F_3$ or even larger fragments (such as $F_4$,...,$F_n$). Similarly, at site $S_3$, we should generate crossing matches by joining local partial matches from $F_3$ and fragments smaller than $F_3$. The crossing matches should not contain any local partial match from $F_4$ or even larger fragments (such as $F_5$,...,$F_n$).

The ``divide-and-conquer'' framework can avoid duplicate results, since each crossing match can be only generated at a single site according to the ``divided search space''. To enable the ``divide-and-conquer'' framework, we need to introduce some constraints over data communication. The transmission (of local partial matches) from fragment site $S_i$ to $S_j$ is allowed only if $F_i \prec F_j$.

Let us consider an intermediate result $PM$ in $\Delta^m_{out}(F_i)$. Assume that $PM$ is generated by joining intermediate results from $m$ different fragments  $F_{i_1},...,F_{i_m}$, where $F_{i_1} \prec F_{i_2} \prec... \prec F_{i_m}$. We send $PM$ to another fragment $F_j$ if and only if  two conditions hold: (1) $F_j > F_{i_m}$; and (2) $F_j$ shares common crossing edges with at least one fragment of $F_{i_1}$,...,$F_{i_m}$.

\textbf{Barrier Synchronisation.} All communication in the $m$-th superstep should finish before entering in the $(m+1)$-th superstep.

We now discuss the initial state (i.e., $0$-th superstep) and the system termination condition.

\textbf{Initial State}. In the 0-th superstep, each fragment $F_i$ has only local partial matches in $F_i$, i.e, $\Omega_{F_i}$.  Since it is impossible to assemble local partial matches in the same fragment, the 0-th superstep requires no local computation. It enters the communication stage directly. Each site $S_i$ sends $\Omega_{F_i}$ to other fragments according to the communication strategy that has been discussed before.

\textbf{System Termination Condition.} A key problem in the BSP algorithm is \emph{the number of the supersteps} to terminate the system. In order to facilitate the analysis, we propose using a fragmentation graph topology graph.

\vspace{-0.05in}
\begin{definition}(\textbf{Fragmentation Topology Graph}) Given a fragmentation $\mathcal{F}$ over an RDF graph $G$, the corresponding \emph{fragmentation topology graph} $T$ is defined as follows: Each node in $T$ is a fragment $F_i$, $i=1,...,k$. There is an edge between nodes $F_i$ and $F_j$ in $T$, $1 \leq i \neq j \leq n$, if and only if there is at least one crossing edge between $F_i$ and $F_j$ in RDF graph $G$.
\end{definition}
\vspace{-0.05in}

Let $Dia(T)$ be the diameter of $T$. We need at most $Dia(T)$ supersteps to transfer the local partial matches in one fragment $F_i$ to any other fragment $F_j$. Hence, the number of the supersteps in the BSP-based algorithm is $Dia(T)$.

%!TEX root =  distributedgStore.tex
\section{Handling General SPARQL}\label{sec:Extensions}
So far, we only consider BGP (basic graph pattern) query evaluation. In this section, we discuss how to extend our method to general SPARQL queries involving UNION, OPTIONAL and FILTER statements.

A general SPARQL query and SPARQL query results can be defined recursively based on BGP queries.

\begin{definition}\label{def:query}\textbf{(General SPARQL Query)}
Any BGP is a SPARQL query. If $Q_1$ and $Q_2$ are SPARQL queries, then expressions $(Q_1 \; AND \; Q_2)$, $(Q_1 \;  UNION \; Q_2)$, $(Q_1 \; OPT \; Q_2)$ and $(Q_1\; FILTER\; F)$ are also SPARQL queries.
\end{definition}

\begin{figure}
   \centering
      \includegraphics[scale=0.54]{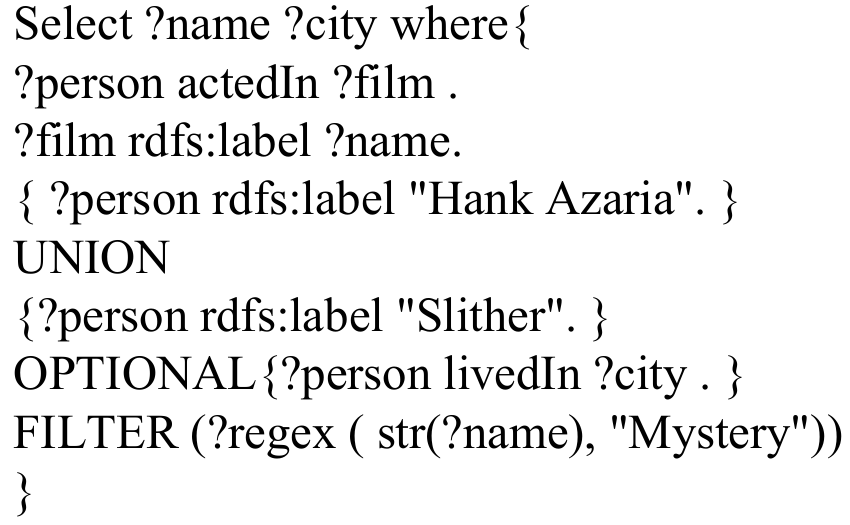}
 \caption{Example General SPARQL Query with UNION, OPTIONAL and FILTER}%
 \label{fig:GeneralSPARQLQuery}
\end{figure}

Figure \ref{fig:GeneralSPARQLQuery} shows an example general SPARQL query with multiple operators, including UNION, OPTIONAL and FILTER.
The set of all matches for $Q$ is denoted as $\llbracket Q \rrbracket$.

\begin{definition}\label{def:queryresult}\textbf{(Match of General SPARQL Query)}
Given an RDF graph $G$, the match set of a SPARQL query $Q$ over $G$, denoted as $[\kern-0.15em[ Q] \kern-0.15em]$, is defined recursively as follows:
\begin{enumerate}
\item If $Q$ is a BGP, $[\kern-0.15em[ Q] \kern-0.15em]$ is the set of matches defined in Definition 3 of Section 3.
\item If $Q=Q_1 \; AND \; Q_2$, then  $
 [\kern-0.15em[ Q ]\kern-0.15em]  =  [\kern-0.15em[ Q_1  ]\kern-0.15em]  \Join [\kern-0.15em[ Q_2  ]\kern-0.15em] $
\item If $Q=Q_1 \; UNION \; Q_2$, then  $
 [\kern-0.15em[ Q ]\kern-0.15em]  =  [\kern-0.15em[ Q_1  ]\kern-0.15em]  \cup [\kern-0.15em[ Q_2  ]\kern-0.15em]$
\item If $Q=Q_1 \; OPT \; Q_2$, then  $
 [\kern-0.15em[ Q ]\kern-0.15em]  =  ([\kern-0.15em[ Q_1  ]\kern-0.15em]  \Join [\kern-0.15em[ Q_2  ]\kern-0.15em] ) \cup ([\kern-0.15em[ Q_1  ]\kern-0.15em] \backslash [\kern-0.15em[ Q_2  ]\kern-0.15em])  $
\item If $Q=Q_1 \; FILTER \; F$, then  $
 [\kern-0.15em[ Q ]\kern-0.15em]  =   \Theta_F ([\kern-0.15em[ Q_1 ] \kern-0.15em]) $
\end{enumerate}

\end{definition}

We can parse each SPARQL query into a parse tree\footnote{We use ANTRL v3's grammar which is an implementation of the SPARQL grammar's specifications. It is available at http://www.antlr3.org/grammar/1200929755392/}, where the root is a \emph{pattern group}. A pattern group specifies a SPARQL statement, and consists of a BGP query with UNION, OPTIONAL and FILTER statements. The UNION and OPTIONAL may recursively contain multiple pattern groups. It is easy to show that each leaf node (in the parser tree) is a BGP query whose evaluation was discussed earlier. We design a recursive algorithm (Algorithm \ref{alg:handlingGeneralSPARQLs}) to find answers to handle UNION, OPTIONAL and FILTER. Specifically, we perform left-outer join between BGP and OPTIONAL query results (Lines 4-5 in Function RecursiveEvaluation). Then, we join the answer set with  UNION query results  (Line 9 in Function RecursiveEvaluation). Finally, we evaluate FILTER operator (Line 13).

\begin{algorithm} \label{alg:handlingGeneralSPARQLs}
\caption{Handling General SPARQLs}

\KwIn{A SPARQL $Q$}
\KwOut{The result set $RS$ of $Q$}

Parse $Q$ into a parser tree $T$\\
$RS=$\textbf{RecursiveEvaluation}($T$) // Call Function
\end{algorithm}

\begin{function}\label{alg:functioncom}
\small
\caption{RecursiveEvaluation($T$)}
Evaluate BGP in $T$ and put all its results into $RS$ \\
\For{each subtree $T^\prime$ in OPTIONAL statement of $T$}
{// \emph{Handling OPTIONAL stamtent.} \\
   $RS^\prime=$\textbf{RecursiveEvaluation}($T^\prime$)\\
   $RS=RS\leftouterjoin RS^\prime$
}
$RS^\prime=\emptyset$\\
\For{each subtree $T^\prime$ of pattern group in UNIONs of $T$}
{// \emph{Handling UNION stamtent.} \\
    $RS^\prime=RS^\prime \ UNION$ \textbf{RecursiveEvaluation}($T^\prime$)\\
}
$RS=RS\Join RS^\prime$\\
\For{each expression $F$ in FILTER operators}
{// \emph{Handling FILTER operator.} \\
    Select $RS$ by using expression $F$\\
}
Return $RS$
\end{function}

\begin{figure}
   \centering
      \includegraphics[scale=0.375]{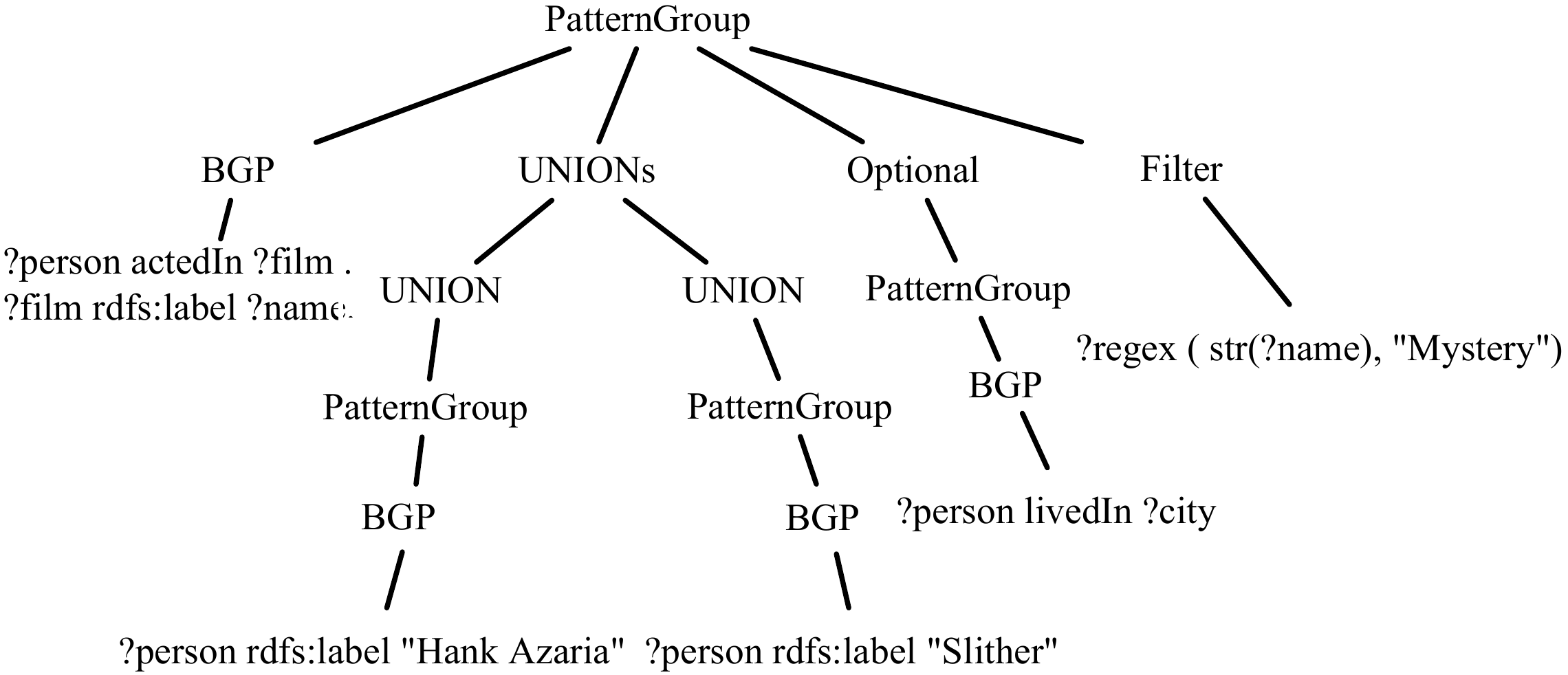}
 \caption{Parser Tree of Example General SPARQL Query}%
 \label{fig:TreePaser}
\end{figure}

Further optimizing general SPARQL evaluation is also possible (e.g., \cite{DBLP:conf/sigmod/Atre15}). However, this issue is independent on our studied problem in this paper.

%!TEX root =  distributedgStore.tex

\section{Experiments}\label{sec:experiment}

We evaluate our method over both real and synthetic RDF datasets, and compare our approach with the state-of-the-art distributed RDF systems, including  a cloud-based approach (EAGRE \cite{ICDE13:EAGRE}), two partition-based approaches (GraphPartition \cite{VLDB11:GraphPartition} and TripleGroup \cite{DBLP:journals/pvldb/LeeL13}), two memory-based systems (TriAD \cite{SIGMOD2014:TriAD} and Trinity.RDF \cite{VLDB13:Trinity}) and two federated SPARQL query systems (FedX \cite{DBLP:FedX} and SPLENDID \cite{DBLP:SPLENDID}). The results of of federated system comparisons are given in Appendix \ref{sec:feder-exp} since, as argued earlier, the environment targeted by these systems is different than ours.

\nop{
federated SPARQL query systems (DARQ \cite{DBLP:conf/esws/QuilitzL08}, FedX \cite{DBLP:conf/esws/SchwarteHHSS11} and Q-Tree \cite{DBLP:conf/edbt/PrasserKK12}) and RDF data partition-based methods (GraphPartition \cite{VLDB11:GraphPartition} and EAGRE \cite{ICDE13:EAGRE}.)
}

\textbf{Setting.} We use two benchmark datasets with different sizes and one real dataset in our experiments, in addition to FedBench used in federated system experiments. Table \ref{table:datasets} summarizes the statistics of these datasets. All sample queries are shown in Appendix \ref{sec:allqueries}.

1) WatDiv \cite{DBLP:WatDiv} is a benchmark that enables diversified stress testing of RDF data management systems. In WatDiv, instances of the same type can have different attribute sets. We generate three datasets varying sizes from 100 million to 1 billion triples. We use 20 queries of the basic testing templates provided by WatDiv \cite{DBLP:WatDiv} to evaluate our method. We randomly partition the WatDiv datasets into several fragments (except in Exp. 6 where we test different partitioning strategies). We assign each vertex $v$ in RDF graph to the $i$-th fragment if $H(v) MOD\ N=i$, where $H(v)$ is a hash function and N is the number of fragments. By default, we use the uniform hash function and $N=10$. Each machine stores a single fragment.

\nop{
In Exp 6, we test the performance of our system under different partitioning strategies on the same WatDiv database. We use the uniform distribution hash function, the exponential distribution hash function and minimum-cut graph partitioning strategies.

We use a hash function to associate each vertex in WatDiv with a value. Each vertex is assigned to a fragment according to its hash value. The number of fragments is the same to the number of machines and each fragment is randomly allocated to one machine. We also use other partitioning strategies to divide RDF graph to evaluate the impact of different partitioning strategies.}

2) LUBM \cite{DBLP:LUBM} is a benchmark that adopts an ontology for the university domain, and can generate synthetic OWL data scalable to an arbitrary size. We assign the university number to 10000. The number of triples is about 1.33 billion. We partition the LUBM datasets according to the university identifiers.
Although LUBM defines 14 queries, some of these are similar; therefore we use the 7 benchmark queries that have been used in some recent studies \cite{DBLP:conf/www/AtreCZH10,Zou:2013fk}. We report the results over all 14 queries in Appendix \ref{sec:allqueries} for completeness. As expected ,the results over 14 benchmark queries are similar to the results over 7 queries.

3) BTC 2012 (http://km.aifb.kit.edu/projects/btc-2012/) is a real dataset that serves as the basis of submissions to the Billion Triples Track of the Semantic Web Challenge. After eliminating all redundant triples, this dataset contains about 1 billion triples. We use METIS to partition the RDF graph, and use the 7 queries in \cite{ICDE13:EAGRE}.

4) FedBench \cite{DBLP:FedBench} is used for testing against federated systems; it is described in Appendix \ref{sec:feder-exp} along with the results.

\begin{table}
\scriptsize
\centering
%\vspace{-0.1in}
\caption{Datasets}
\begin{tabular}{|c|c|r|r|r|}
\hline
\multicolumn{2}{|c|}{Dataset}& \tabincell{c}{Number of\\Triples}& \tabincell{c}{RDF N3 File\\Size(KB)}  & \tabincell{c}{Number of\\Entities}\\
  \hline
  \hline
  \multicolumn{2}{|c|}{WatDiv 100M} & 109,806,750 & 15,386,213 & 5,212,745\\
  \hline
  \multicolumn{2}{|c|}{WatDiv 300M} & 329,539,576  & 46,552,961 &15,636,385 \\
  \hline
  \multicolumn{2}{|c|}{WatDiv 500M} & 549,597,531  & 79,705,831 & 26,060,385\\
  \hline
  \multicolumn{2}{|c|}{WatDiv 700M} & 769,065,496  & 110,343,152 & 36,486,007\\
  \hline
  \multicolumn{2}{|c|}{WatDiv 1B} & 1,098,732,423 & 159,625,433 & 52,120,385\\
  \hline
  \multicolumn{2}{|c|}{LUBM 1000} &133,553,834 & 15,136,798 & 21,715,108\\
  \hline
  \multicolumn{2}{|c|}{LUBM 10000}& 1,334,481,197 & 153,256,699 & 217,006,852\\
  \hline
  \multicolumn{2}{|c|}{BTC}&1,056,184,911 & 238,970,296 & 183,835,054\\
  \hline
  \end{tabular}
%  \vspace{-0.1in}
\label{table:datasets}
\end{table}

We conduct all experiments on a cluster of 10 machines running Linux, each of which has one CPU with four cores of 3.06GHz, 16GB memory and 500GB disk storage. Each site holds one fragment of the dataset. At each site, we install gStore \cite{Zou:2013fk} to find inner matches, since it supports the graph-based SPARQL evaluation paradigm. We revise gStore to find all local partial matches in each fragment as discussed in Section \ref{sec:partialcomputing}. All implementations are in standard C++. We use MPICH-3.0.4 library for communication.

\nop{
We compare our method with the cloud-based competitors (EAGRE, GraphPartition and TripleGroup), distributed memory-based systems (TriAD and Trinity.RDF) and federal distributed RDF query systems (FedX and SPLENDID) in our experiments.
}

\nop{
Because that one author of EAGRE is also co-author of this paper, we compare this paper with the authors' codes of EAGRE. As well, GraphPartition and TripleGroup codes have been released by authors. EAGRE stores all triples as flat files in HDFS and answers SPARQL queries by scanning the files. It further use C to employ MPICH2 for consulting to reduce the unnecessary I/O. GraphPartition and TripleGroup manually decompose the input query into multiple subqueries in offline, and employ the existing centralized RDF systems to find intermediate results. If there exist some crossing matches, GraphPartition and TripleGroup should start MapReduce jobs to join the intermediate results. The codes by GraphPartition and TripleGroup only include MapReduce job codes, which exclude the centralized RDF systems for finding intermediate results. Actually, any centralized RDF systems can be used. In our experiments, we use gStore to find intermediate results. Obviously, MapReduce job related codes are Java in these systems.
}

\nop{
In addition, memory-based systems (like TriAD and Trinity.RDF) are provided by authors, and the two systems are based on C++. Federal query systems (like FedX and SPLENDID) have also been released by authors. We download these systems and use Sesame 2.7\footnote{http://rdf4j.org/} to build SPARQL endpoints.
}

\begin{table*}
\scriptsize
\begin{threeparttable}
  \begin{tabular}{|p{0.85cm}|p{0.35cm}|p{0.15cm}|p{0.08cm}|r|r|r|r|r|r|r|r|r|r|r|}
  \hline
   &  & &  & \multicolumn{3}{c|}{Partial Evaluation} & \multicolumn{3}{c|}{Assembly} &  \multicolumn{3}{c|}{Total} & \tabincell{c}{$\#$ of }  &  \tabincell{c}{$\#$ of}\\
  \cline{5-13}
  &  & & &  & & & \multicolumn{2}{c|}{Time(in ms)}& & \multicolumn{2}{c|}{Time(in ms)}& & LPMFs\tnote{$8$} &CMFs\tnote{$9$}\\
  \cline{8-9}\cline{11-12}
  & & & &  Time(in ms) &  $\#$ of LPMs\tnote{$2$}& $\#$ of IMs\tnote{$3$} & \tabincell{c}{Centralized}	& \tabincell{c}{Distributed} & $\#$ of CMs\tnote{$4$}& PECA\tnote{$5$}	& PEDA\tnote{$6$}& $\#$ of Matches\tnote{$7$} & & \\
 \hline
  \multirow{7}{*}{\tabincell{c}{Star}}& \multirow{7}{*}{\includegraphics[scale=0.07]{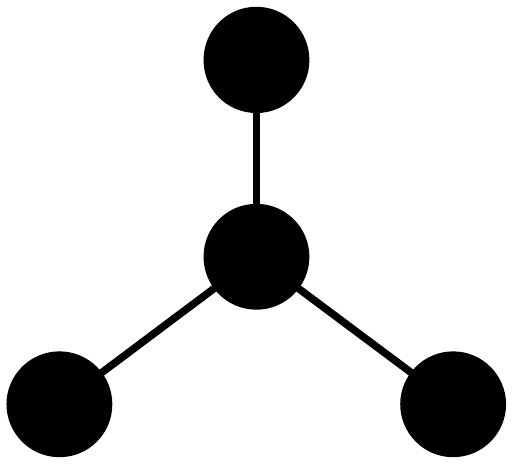}}&  \tabincell{l}{$S_1$}&$\surd$\tnote{$1$} & 43803& 0 & 1 & 0 & 0	& 0 &43803  & 43803 &1  & 0&0\\
  \cline{3-15}
  & & \tabincell{l}{$S_2$} & \tabincell{l}{$\surd$} & 74479& 0 & 13432 & 0 & 0	& 0 &74479  & 74479 & 13432  & 0&0 \\
  \cline{3-15}
  & & \tabincell{l}{$S_3$} & \tabincell{l}{$\surd$} & 8087& 0 & 13335 & 0 & 0	& 0 &8087  & 8087 & 13335  & 0&0 \\
  \cline{3-15}
  & & \tabincell{l}{$S_4$} & \tabincell{l}{$\surd$} & 16520 & 0 & 2 &0  & 0	& 0 &16520  & 16520 & 1 & 0&0 \\
  \cline{3-15}
  & & \tabincell{l}{$S_5$} & \tabincell{l}{$\surd$} & 1861 & 0 & 112 & 0  & 0	& 0  & 1861 & 1861 &940  & 0&0 \\
  \cline{3-15}
  & & \tabincell{l}{$S_6$} & \tabincell{l}{$\surd$} &50865  & 0 & 14 & 0  & 0	& 0  & 50865 & 50865 & 14 & 0&0\\
  \cline{3-15}
  & & \tabincell{l}{$S_7$} & \tabincell{l}{$\surd$} & 56784 &  0 & 1 & 0  & 0	& 0  & 56784 & 56784 &1  & 0&0 \\
  \cline{1-15}
  \multirow{5}{*}{\tabincell{c}{Linear}}& \multirow{5}{*}{\includegraphics[scale=0.07]{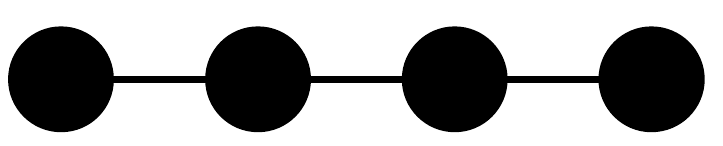}} &\tabincell{l}{$L_1$} & \tabincell{l}{$\surd$} &  15340 & 2 & 0 & 1	&16	& 1 &15341 & 15356 &1 &2 & 2\\
  \cline{3-15}
  & &  \tabincell{l}{$L_2$} & \tabincell{l}{$\surd$} &  1492	& 794	& 88  &  18	& 130 	&  793 &  1510	& 1622  & 881 &10 &10 \\
  \cline{3-15}
  & &  \tabincell{l}{$L_3$} & \tabincell{l}{$\surd$} &  16889	& 0 & 5  &  0	& 0	&  0 &  16889	&16889  & 5  & 0 &0  \\
  \cline{3-15}
  & &  \tabincell{l}{$L_4$} & \tabincell{l}{$\surd$} &  261	&0  &  6005 &   0	& 0	&  0 &  261	& 261 & 6005 & 0 & 0 \\
  \cline{3-15}
  & &  \tabincell{l}{$L_5$} & \tabincell{l}{$\surd$} &  48055	& 1274 & 141  & 572	& 1484	& 1273  &  48627	& 49539 & 1414 & 10 &10 \\
  \cline{1-15}
  \multirow{5}{*}{\tabincell{c}{Snowflake}}& \multirow{5}{*}{\includegraphics[scale=0.07]{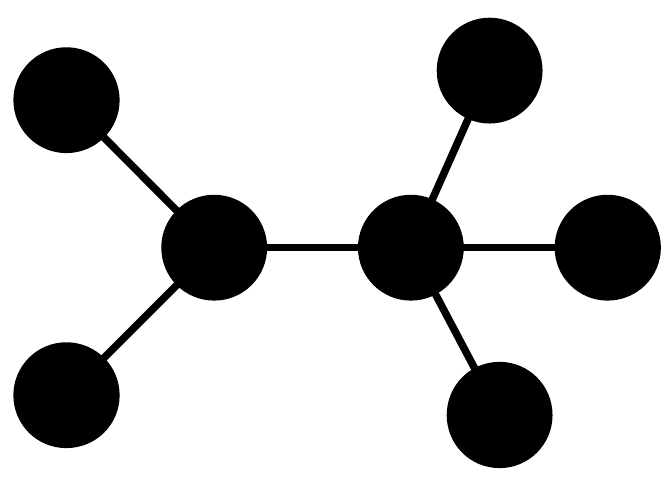}} &\tabincell{l}{$F_1$} & \tabincell{l}{$\surd$} & 64699 & 29	& 1	 & 9	&49	& 14 &  64708	&  64748   &15 & 10 &10 \\
  \cline{3-15}
  & &  \tabincell{l}{$F_2$} & \tabincell{l}{$\surd$} &  203968	& 2184 &  99 &  1598	& 3757	&  1092 & 205566 &	207725 & 1191 & 10 &10 \\
  \cline{3-15}
  & &  \tabincell{l}{$F_3$} & \tabincell{l}{$\surd$} &  2341932	&  4065632 &  58 & 3673409  &	2489325  	&  6200 &  6015341	&4831257 & 6258 & 10 &10 \\
  \cline{3-15}
  & &  \tabincell{l}{$F_4$} & \tabincell{l}{$\surd$} &  251546	& 6909 & 0  & 13693 &	8864	&  1808 & 265239 & 260410 & 1808 & 10 &10 \\
  \cline{3-15}
  & &  \tabincell{l}{$F_5$} & \tabincell{l}{$\surd$} &  25180	& 92 & 3  &  58	& 1028	&  46 & 25238 	&26208 & 49 & 10 &10 \\
  \cline{1-15}
  \multirow{3}{*}{\tabincell{c}{Complex}}& \multirow{3}{*}{\includegraphics[scale=0.07]{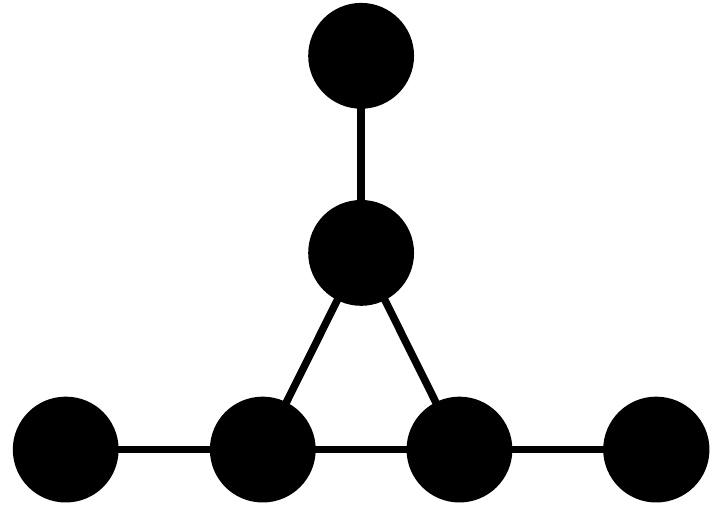}}  &\tabincell{l}{$C_1$} & &206864&   161803	&  4&  9195	&	5265&	356& 	216059&212129	& 360 & 10 &10 \\
  \cline{3-15}
  & & \tabincell{l}{$C_2$} & & 1613525  & 937198 & 0  &   229381& 174167&  155 &  1842906 &1787692 & 155 &  10 &10 \\
  \cline{3-15}
  & & \tabincell{l}{$C_3$} & & 123349& 0 & 80997 & 0 & 0 & 0  & 123349 & 123349 &80997 & 0 & 0 \\
  \hline
  \end{tabular}
  \begin{tablenotes}
  \scriptsize
 % \setlength{\columnsep}{0.8cm}
  %\setlength{\multicolsep}{0cm}
  %\begin{multicols}{2}

  \item[$1$] $\surd$ means that the query involves some selective triple patterns.~~~~~~~~~~~~~~~~~~~~~~~~~~$^6$``PEDA'' is the abbreviation of \emph{P}artial \emph{E}valuation $\&$ \emph{D}istributed \emph{A}ssembly.
  \item[$2$] ``$\#$ of LPMs'' means the number of local partial matches. ~~~~~~~~~~~~~~~~~~~~~~~~~~~~~~~~~ $^7$  ``$\#$ of Matches'' means the number of matches.
  \item[$3$] ``$\#$ of IMs'' means the number of inner matches.  ~~~~~~~~~~~~~~~~~~~~~~~~~~~~~~~~~~~~~~~~~~~~~~~ $^8$  ``$\#$ of LPMFs'' means the number of fragments containing local partial matches.

  \item[$4$] ``$\#$ of CMs'' means the number of crossing matches.  ~~~~~~~~~~~~~~~~~~~~~~~~~~~~~~~~~~~~~~~~~ $^9$  ``$\#$ of CMFs'' means the number of fragments containing crossing matches.

  \item[$5$] ``PECA'' is the abbreviation of \emph{P}artial \emph{E}valuation $\&$ \emph{C}entralized \emph{A}ssembly.
%  \item[$6$] ``PEDA'' is the abbreviation of \emph{P}artial \emph{E}valuation $\&$ \emph{D}istributed \emph{A}ssembly.
 % \item[$7$] ``$\#$ of Matches'' means the number of matches.
 % \item[$8$] ``$\#$ of LPMFs'' means the number of fragments containing local partial matches.
 % \item[$9$] ``$\#$ of CMFs'' means the number of fragments containing crossing matches.

  %\end{multicols}
  \end{tablenotes}
  \end{threeparttable}
  \caption{Evaluation of Each Stage on WatDiv 1B}
  \label{table:queriesperformanceWatDiv}
\end{table*}

 \begin{table*}
\vspace{-0.1in}
\scriptsize
\begin{threeparttable}
  \begin{tabular}{|c|p{0.35cm}|p{0.15cm}|p{0.08cm}|r|r|r|r|r|r|r|r|r|r|r|}
  \hline
   & & & & \multicolumn{3}{c|}{Partial Evaluation} & \multicolumn{3}{c|}{Assembly} &  \multicolumn{3}{c|}{Total}  & \tabincell{c}{$\#$ of }  &  \tabincell{c}{$\#$ of} \\
  \cline{5-13}
  & & & &  & & & \multicolumn{2}{c|}{Time(in ms)}& & \multicolumn{2}{c|}{Time(in ms)}& & LPMFs &CMFs \\
  \cline{8-9}\cline{11-12}
  & & & &  Time(in ms) &  $\#$ of LPMs& $\#$ of IMs & \tabincell{c}{Centralized}	& \tabincell{c}{Distributed} & $\#$ of CMs& \tabincell{c}{PECA}	& \tabincell{c}{PEDA}& $\#$ of Matches & & \\
 \hline
  \multirow{3}{*}{\tabincell{c}{Star}}& \multirow{3}{*}{\includegraphics[scale=0.07]{query_style_star.pdf}}&  \tabincell{l}{$Q_2$}& & 1818 & 0 & 1081187 & 0	& 0	& 0& 1818 & 1818 & 1081187 & 0 & 0  \\
  \cline{3-15}
  & & \tabincell{l}{$Q_4$} & \tabincell{l}{$\surd$} & 82	& 0 & 10 & 0	& 0	& 0 & 82	& 82 & 10 & 0 & 0 \\
  \cline{3-15}
  & & \tabincell{l}{$Q_5$} & \tabincell{l}{$\surd$} & 8	& 0 & 10 & 0	& 0	& 0 & 8	& 8 & 10 & 0 & 0 \\
  \cline{1-15}
  \multirow{2}{*}{\tabincell{p{0.9cm}}{Snowflake}}& \multirow{2}{*}{\includegraphics[scale=0.07]{query_style_snowflake.pdf}} &  \multirow{2}{*}{\tabincell{l}{$Q_6$}}& \multirow{2}{*}{$\surd$} & \multirow{2}{*}{158}	& \multirow{2}{*}{ 6707} & \multirow{2}{*}{110} & \multirow{2}{*}{164}	& \multirow{2}{*}{125}	& \multirow{2}{*}{15} & \multirow{2}{*}{322}	& \multirow{2}{*}{283} & \multirow{2}{*}{125} & \multirow{2}{*}{10} & \multirow{2}{*}{10} \\
  & &   &  &  	&  &   &  	&  	&   &  	&  &  & & \\
  \cline{1-15}
  \multirow{3}{*}{\tabincell{p{0.8cm}}{Complex}}& \multirow{3}{*}{\includegraphics[scale=0.07]{query_style_complex.pdf}}  &\tabincell{l}{$Q_1$} &  & 52548	& 3033 & 2524 & 53	& 60	& 4 & 52601	&52608 & 2528 & 10 & 10 \\
  \cline{3-15}
  & & \tabincell{l}{$Q_3$} &  & 920	&3358 &0 & 36 	&  48 & 0	& 956 	& 968	& 0 & 10&0 \\
  \cline{3-15}
  & & \tabincell{l}{$Q_7$} &  & 3945	&167621 & 42479 & 211670&	35856	& 1709 & 215615	&39801 & 44190 & 10&10 \\
  \hline
  \end{tabular}
  \end{threeparttable}
  \caption{Evaluation of Each Stage on LUBM 1000}
  \label{table:queriesperformanceLUBM}
\end{table*}

\textbf{Exp 1. Evaluating Each Stage's Performance.} In this experiment, we study the performance of our system at each stage (i.e., partial evaluation and assembly process) with regard to different queries in WatDiv 1B and LUBM 1000. We report the running time of each stage (i.e., partial evaluation and assembly) and the number of local partial matches, inner matches, and crossing matches, with regard to different query types in Tables \ref{table:queriesperformanceWatDiv} and \ref{table:queriesperformanceLUBM}. We also compare the centralized and distributed assembly strategies. The time for assembly includes the time for computing the optimal join order. Note that we classify SPARQL queries into four categories according to query graphs' structures: star, linear, snowflake (several stars linked by a path) and complex (a combination of the above with complex structure). \nop{There is no linear query in LUBM.}

\underline{Partial Evaluation}: Tables \ref{table:queriesperformanceWatDiv} and \ref{table:queriesperformanceLUBM} show that if there are some selective triple patterns\footnote{A triple pattern t is a ``selective triple pattern'' if it has no more than 100 matches in RDF graph G} in the query, the partial evaluation is much faster than others. Our partial evaluation algorithm (Algorithm \ref{alg:findinglocalmaximal}) is based on a state transformation, while the selective triple patterns can reduce the search space. Furthermore, the running time also depends on the number of inner matches and local partial matches, as shown in Tables \ref{table:queriesperformanceWatDiv} and \ref{table:queriesperformanceLUBM}. More inner matches and local partial matches lead to higher running time in the partial evaluation stage.

\underline{Assembly}: In this experiment, we compare centralized and distributed assembly approaches. Obviously, there is no assembly process for a star query. Thus, we only study the performances of linear, snowflake and complex queries.  We find that distributed assembly can beat the centralized one when there are lots of local partial matches and crossing matches. The reason is as follows: in centralized assembly, all local partial matches need to be sent to the server where they are assembled. Obviously, if there are lots of local partial matches, the server becomes the bottleneck. However, in distributed assembly, we can take advantage of parallelization to speed up both the network communication and assembly. For example, in $F_3$, there are 4065632 local partial matches. It takes a long time to transfer the local partial matches to the server and assemble them in the server in centralized assembly. So, distributed assembly outperforms the centralized alternative. However, if the number of local partial matches and the number of crossing matches are small, the barrier synchronisation cost dominates the total cost in distributed assembly. In this case, the advantage of distributed assembly is not clear. A quantitative comparison between  distributed  and centralized assembly approaches needs more statistics about the network communication, CPU and other parameters. A sophisticated quantitative study is beyond the scope of this paper and  is left as future work.

In Tables \ref{table:queriesperformanceWatDiv} and \ref{table:queriesperformanceLUBM}, we also show the number of fragments involved in each test query. For most queries, their local partial matches and crossing matches  involve all fragments. Queries containing selective triple patterns ($L_1$ in WatDiv) may only involve a part of the fragmentation.

\textbf{Exp 2: Evaluating Optimizations in Assembly.} In this experiment, we use WatDiv 1B to evaluate two different optimization techniques in the assembly: partitioning-based join strategy (Section \ref{sec:basicjoin}) and the divide-and-conquer approach in the distributed assembly (Section \ref{sec:distributed}). If a query does not have any local partial matches in RDF graph $G$, it does not need the assembly process. Therefore, we only use the benchmark queries that need assembly  ($L_1$, $L_2$, $L_5$, $F_1$, $F_2$, $F_3$, $F_3$, $F_4$, $F_5$, $C_1$ and $C_2$) in our experiments.

\underline{Partitioning-based Join}. First, we compare partitioning-based join (i.e., Algorithm 3) with naive join processing (i.e., Algorithm 2) in Table \ref{table:PartitionJoin}, which shows that the partitioning-based strategy can greatly reduce the join cost. Second, we evaluate the effectiveness of our cost model. Note that the join order depends on the partitioning strategy, which is based on our cost model as discussed in Section \ref{sec:optimalpartition}. In other words, once the partitioning is given, the join order is fixed. So, we use the cost model to find the optimal partitioning and report the running time of the assembly process in Table \ref{table:PartitionJoin}. We find that the assembly with the optimal partitioning is faster than that with random partitioning, which confirms the effectiveness of our cost model. Especially for $C_2$, the assembly with the optimal partitioning is an order of magnitude faster than the assembly with random partitioning.\nop{ Note that, there are so few local partial matches for $L_1$, $L_2$, $L_5$ and $F_1$ that the difference among joining strategies is very small.}

\underline{Divide-and-Conquer in Distributed Assembly}. Table \ref{table:Dividingeffect}
shows that dividing the search space will speed up  distributed assembly. Otherwise some duplicate results can be generated, as discussed in Section \ref{sec:distributed}. Elimination of duplicates and parallelization  speeds up distributed assembly. For example, for $C_1$, dividing search space lowers the time of assembly more than twice as much as no dividing search space.

\textbf{Exp 3: Scalability Test}. In this experiment, we vary the RDF dataset size from 100 million triples (WatDiv 100M) to 1 billion triples (WatDiv 1B) to study the scalability of our methods. Figures \ref{fig:ScalabilityTestPECA} and \ref{fig:ScalabilityTestPEDA} show the performance of different queries using centralized and distributed assembly.

\begin{table}[t]
\scriptsize
\begin{threeparttable}
  \begin{tabular}{|p{0.05in}|r|r|r|}
  \hline
  & \tabincell{c}{Partitioning-based Join Based\\on the Optimal Partitioning }	& \tabincell{c}{Partitioning-based Join Based\\on the Random Partitioning }& \tabincell{c}{Naive Join}\\
  \hline
   $L_1$ & 1 	 &1 &1 \\
     \hline
    $L_2$ &	18 & 23 & 139\\
  \hline
    $L_5$ &572	 & 622 & 3419\\
\hline
   $F_1$ & 1 	 & 1 &1 \\
     \hline
    $F_2$ &	1598 & 2286 & 48096\\
  \hline
    $F_3$ &3673409	 & 4005409 & timeout\tnote{$1$}\\
    \hline
   $F_4$ & 13693 	 & 13972 &timeout\\
     \hline
    $F_5$ &	58  &80  & 8383\\
  \hline
    $C_1$ &	9195  & 10582 &timeout\\
  \hline
  $C_2$ &	229381 & 4083181& timeout\\
  \hline
  \end{tabular}
    \begin{tablenotes}
  \scriptsize
  \item[$1$] timeout is issued if query evaluation does not terminate in 10 hour
  \end{tablenotes}
  \end{threeparttable}
  \caption{Running Time of Partitioning-based Join vs. Naive Join (in ms)}
  \label{table:PartitionJoin}
\end{table}

\begin{table}
\scriptsize
\centering
  \begin{tabular}{|r|r|r|}
  \hline
  &  \multicolumn{2}{c|}{Distributed Assembly Time (in ms)}\\ \hline
  & \tabincell{c}{Dividing}	& \tabincell{c}{No Dividing}\\
  \hline
   $L_1$ & 16	& 19 \\
  \hline
     $L_2$& 130	& 151  \\
  \hline
    $L_5$ & 1484 	&1684 \\
    \hline
    $F_1$ & 49	& 55 \\
  \hline
     $F_2$& 3757	&  5481 \\
  \hline
    $F_3$ & 2489325 	&4439430 \\
  \hline
     $F_4$& 8864	&  19759 \\
  \hline
    $F_5$ & 1028 	&1267 \\
  \hline
     $C_1$& 5265	& 12194  \\
  \hline
    $C_2$ & 174167 	& 225062 \\
  \hline
  \end{tabular}
  \caption{Dividing vs. No Dividing (in ms)}
  \label{table:Dividingeffect}
\end{table}

\begin{figure}%
%\vspace{-0.2in}
   \centering
   \subfigure[Star Queries]{%
		\resizebox{0.48\columnwidth}{!}{
			\begin{tikzpicture}[font=\large]
    \begin{semilogyaxis}[
        xlabel=Datasets,
        ylabel=Query Response Time (in ms),
        xtick = {10,30,50,70,100},
        xticklabels={100M,300M,500M,700M,1000M},
        legend cell align=left,
        legend style={draw=none},
        legend pos= north west
    ]
 \addplot plot[mark size=3.5pt] coordinates {
        (10,     4095)
        (30,     5452)
        (50,    7259)
        (70,     28119)
        (100,    43803)
    };
    \addplot plot[mark size=3.5pt]  coordinates {
        (10,     5910)
        (30,     18452)
        (50,    31316)
        (70,     42811.9)
        (100,    74479)
    };
    \addplot plot[mark=triangle*,mark size=3.5pt]  coordinates {
        (10,    869)
        (30,    2869)
        (50,    3357)
        (70,    5722)
        (100,   8087)
    };
    \addplot plot[mark size=3.5pt]  coordinates {
        (10,    1506)
        (30,    4506)
        (50,    8378)
        (70,    10506)
        (100,   16520)
    };
    \addplot plot[mark size=3.5pt]  coordinates {
        (10,    208)
        (30,    608)
        (50,    1015)
        (70,    1408)
        (100,   1861)
    };
    \addplot plot[mark size=3.5pt]  coordinates {
        (10,    5153)
        (30,    8038)
        (50,    12206)
        (70,    35038)
        (100,   50865)
    };
    \addplot plot[mark size=3.5pt]  coordinates {
        (10,    5047)
        (30,    15470)
        (50,    23272)
        (70,    37000)
        (100,   56784)
    };
    \legend{$S_1$\\$S_{2}$\\$S_3$\\$S_{4}$\\$S_5$\\$S_{6}$\\$S_7$\\}

    \end{semilogyaxis}
\end{tikzpicture}
		}
       \label{fig:edgeScalabilityPECA}%
       }%
   \subfigure[Linear Queries]{%
		\resizebox{0.48\columnwidth}{!}{
			\begin{tikzpicture}[font=\large]
    \begin{semilogyaxis}[
        xlabel=Datasets,
        ylabel=Query Response Time (in ms),
        xtick = {10,30,50,70,100},
        xticklabels={100M,300M,500M,700M,1000M},
        legend cell align=left,
        legend style={draw=none},
        legend pos= north west
    ]
      \addplot plot[mark size=3.5pt] coordinates {
        (10,     2301)
        (30,     6301)
        (50,    9880)
        (70,     12301)
        (100,    15341)
    };
    \addplot plot[mark size=3.5pt]  coordinates {
        (10,     271)
        (30,     571)
        (50,    709)
        (70,     1171)
        (100,    1510)
    };
    \addplot plot[mark=triangle*,mark size=3.5pt]  coordinates {
        (10,    1115)
        (30,     3571)
        (50,    6164)
        (70,     7135)
        (100,   16889)
    };
    \addplot plot[mark size=3.5pt]  coordinates {
        (10,    37)
        (30,    107)
        (50,    187)
        (70,    227)
        (100,   261)
    };
    \addplot plot[mark size=3.5pt]  coordinates {
        (10,    7741)
        (30,    22107)
        (50,    29378)
        (70,    38107)
        (100,   48627)
    };
    \legend{$L_1$\\$L_{2}$\\$L_3$\\$L_{4}$\\$L_5$\\}

    \end{semilogyaxis}
\end{tikzpicture}
		}
       \label{fig:starScalabilityPECA}%
       }%
       \\
   \subfigure[Snowflake Queries]{%
		\resizebox{0.48\columnwidth}{!}{
			\begin{tikzpicture}[font=\large]
    \begin{semilogyaxis}[
        xlabel=Datasets,
        ylabel=Query Response Time (in ms),
        xtick = {10,30,50,70,100},
        xticklabels={100M,300M,500M,700M,1000M},
        legend cell align=left,
        legend style={draw=none},
        legend pos= north west
    ]

    \addplot plot[mark size=3.5pt]  coordinates {
        (10,    5754)
        (30,    17680)
        (50,    34757)
        (70,    37680)
        (100,   64708)
    };

    \addplot plot[ mark size=3.5pt]  coordinates {
        (10,    11809)
        (30,    15832)
        (50,    19016)
        (70,    98320)
        (100,   205566)
    };
    \addplot plot[mark=triangle*,mark size=3.5pt]  coordinates {
        (10,     246277)
        (30,    583200)
        (50,    1636342)
        (70,    2283200)
        (100,   6015341)
    };

    \addplot plot[mark size=3.5pt]  coordinates {
        (10,    26439)
        (30,    78320)
        (50,    144505)
        (70,    208320)
        (100,   265239)
    };
    \addplot plot[mark size=3.5pt]  coordinates {
        (10,    11630)
        (30,    14320)
        (50,    17063)
        (70,    21320)
        (100,   25238)
    };
    \legend{ $F_1$\\$F_{2}$\\$F_3$\\$F_{4}$\\$F_5$\\}

    \end{semilogyaxis}
\end{tikzpicture}
		}
       \label{fig:snowflakeScalabilityPECA}%
       }
\subfigure[Complex Queries]{%
		\resizebox{0.48\columnwidth}{!}{
			\begin{tikzpicture}[font=\large]
    \begin{semilogyaxis}[
        xlabel=Datasets,
        ylabel=Query Response Time (in ms),
        xtick = {10,30,50,70,100},
        xticklabels={100M,300M,500M,700M,1000M},
        legend cell align=left,
        legend style={draw=none},
        legend pos= north west
    ]

    \addplot plot[mark size=3.5pt]  coordinates {
        (10,     14980)
        (30,     44980)
        (50,    88932)
        (70,     104980)
        (100,    216059)
    };
    \addplot plot[mark size=3.5pt]  coordinates {
        (10,     147962)
        (30,     447962)
        (50,    784926)
        (70,     947962)
        (100,   1842906)
    };
    \addplot plot[mark=triangle*,mark size=3.5pt]  coordinates {
        (10,    11631)
        (30,    33631)
        (50,    76749)
        (70,    107631)
        (100,    123349)
    };

    \legend{  $C_1$\\$C_2$\\$C_3$\\}

    \end{semilogyaxis}
\end{tikzpicture}
		}
       \label{fig:otherScalabilityPECA}%
       }
 \caption{\small Scalability Test of PECA}%
 \label{fig:ScalabilityTestPECA}
\end{figure}
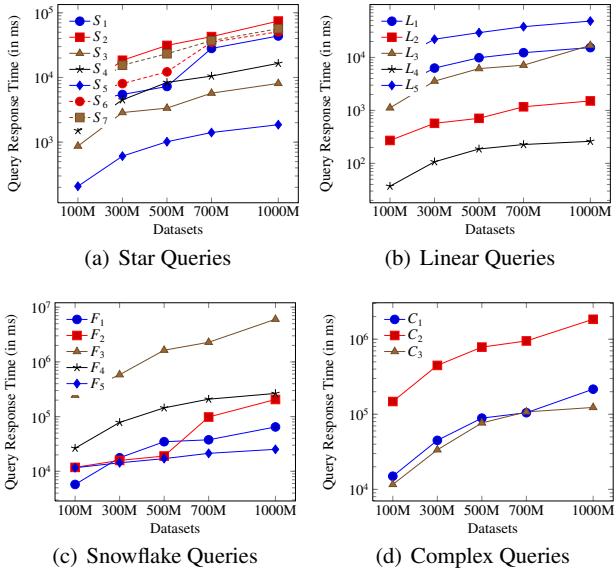

\begin{figure}%
%\vspace{-0.2in}
   \centering
   \subfigure[Star Queries]{%
		\resizebox{0.48\columnwidth}{!}{
			\begin{tikzpicture}[font=\large]
    \begin{semilogyaxis}[
        xlabel=Datasets,
        ylabel=Query Response Time (in ms),
        xtick = {10,30,50,70,100},
        xticklabels={100M,300M,500M,700M,1000M},
        legend cell align=left,
        legend style={draw=none},
        legend pos= north west
    ]
 \addplot plot[mark size=3.5pt] coordinates {
        (10,     4095)
        (30,     5452)
        (50,    7259)
        (70,     28119)
        (100,    43803)
    };
    \addplot plot[mark size=3.5pt]  coordinates {
        (10,     5910)
        (30,     18452)
        (50,    31316)
        (70,     42811.9)
        (100,    74479)
    };
    \addplot plot[mark=triangle*,mark size=3.5pt]  coordinates {
        (10,    869)
        (30,    2869)
        (50,    3357)
        (70,    5722)
        (100,   8087)
    };
    \addplot plot[mark size=3.5pt]  coordinates {
        (10,    1506)
        (30,    4506)
        (50,    8378)
        (70,    10506)
        (100,   16520)
    };
    \addplot plot[mark size=3.5pt]  coordinates {
        (10,    208)
        (30,    608)
        (50,    1015)
        (70,    1408)
        (100,   1861)
    };
    \addplot plot[mark size=3.5pt]  coordinates {
        (10,    5153)
        (30,    8038)
        (50,    12206)
        (70,    35038)
        (100,   50865)
    };
    \addplot plot[mark size=3.5pt]  coordinates {
        (10,    5047)
        (30,    15470)
        (50,    23272)
        (70,    37000)
        (100,   56784)
    };
    \legend{$S_1$\\$S_{2}$\\$S_3$\\$S_{4}$\\$S_5$\\$S_{6}$\\$S_7$\\}

    \end{semilogyaxis}
\end{tikzpicture}
		}
       \label{fig:edgeScalabilityPEDA}%
       }%
   \subfigure[Linear Queries]{%
		\resizebox{0.48\columnwidth}{!}{
			\begin{tikzpicture}[font=\large]
    \begin{semilogyaxis}[
        xlabel=Datasets,
        ylabel=Query Response Time (in ms),
        xtick = {10,30,50,70,100},
        xticklabels={100M,300M,500M,700M,1000M},
        legend cell align=left,
        legend style={draw=none},
        legend pos= north west
    ]
      \addplot plot[mark size=3.5pt] coordinates {
        (10,     2325)
        (30,     6375)
        (50,    10409)
        (70,     12375)
        (100,    15356)
    };
    \addplot plot[mark size=3.5pt]  coordinates {
        (10,     339)
        (30,     539)
        (50,    754)
        (70,     1139)
        (100,    1622)
    };
    \addplot plot[mark=triangle*,mark size=3.5pt]  coordinates {
        (10,    1115)
        (30,    3115)
        (50,    6164)
        (70,    10115)
        (100,   16889)
    };
    \addplot plot[mark size=3.5pt]  coordinates {
        (10,    37)
        (30,    107)
        (50,    187)
        (70,    227)
        (100,   261)
    };
    \addplot plot[mark size=3.5pt]  coordinates {
        (10,    7863)
        (30,    15227)
        (50,    29788)
        (70,    38227)
        (100,   49539)
    };
    \legend{$L_1$\\$L_{2}$\\$L_3$\\$L_{4}$\\$L_5$\\}

    \end{semilogyaxis}
\end{tikzpicture}
		}
       \label{fig:starScalabilityPEDA}%
       }%
       \\
   \subfigure[Snowflake Queries]{%
		\resizebox{0.48\columnwidth}{!}{
			\begin{tikzpicture}[font=\large]
    \begin{semilogyaxis}[
        xlabel=Datasets,
        ylabel=Query Response Time (in ms),
        xtick = {10,30,50,70,100},
        xticklabels={100M,300M,500M,700M,1000M},
        legend cell align=left,
        legend style={draw=none},
        legend pos= north west
    ]

    \addplot plot[mark size=3.5pt]  coordinates {
        (10,    5768)
        (30,    17680)
        (50,    34777)
        (70,    37680)
        (100,   64748)
    };

    \addplot plot[ mark size=3.5pt]  coordinates {
        (10,    11832)
        (30,    15832)
        (50,    19351)
        (70,    98320)
        (100,   207725)
    };
    \addplot plot[mark=triangle*,mark size=3.5pt]  coordinates {
        (10,     163642)
        (30,    583200)
        (50,    1265177)
        (70,    2283200)
        (100,   4831257)
    };

    \addplot plot[mark size=3.5pt]  coordinates {
        (10,    26421)
        (30,    78320)
        (50,    143416)
        (70,    208320)
        (100,   260410)
    };
    \addplot plot[mark size=3.5pt]  coordinates {
        (10,    11654)
        (30,    14654)
        (50,    17326)
        (70,    22654)
        (100,   26208)
    };
    \legend{ $F_1$\\$F_{2}$\\$F_3$\\$F_{4}$\\$F_5$\\}

    \end{semilogyaxis}
\end{tikzpicture}
		}
       \label{fig:snowflakeScalabilityPEDA}%
       }
\subfigure[Complex Queries]{%
		\resizebox{0.48\columnwidth}{!}{
			\begin{tikzpicture}[font=\large]
    \begin{semilogyaxis}[
        xlabel=Datasets,
        ylabel=Query Response Time (in ms),
        xtick = {10,30,50,70,100},
        xticklabels={100M,300M,500M,700M,1000M},
        legend cell align=left,
        legend style={draw=none},
        legend pos= north west
    ]

    \addplot plot[mark size=3.5pt]  coordinates {
        (10,     14667)
        (30,     44667)
        (50,    86524)
        (70,     94667)
        (100,    212129)
    };
    \addplot plot[mark size=3.5pt]  coordinates {
        (10,     147406)
        (30,     447406)
        (50,    782545)
        (70,     982545)
        (100,   1787692)
    };
    \addplot plot[mark=triangle*,mark size=3.5pt]  coordinates {
        (10,    11631)
        (30,    33631)
        (50,    76749)
        (70,    77631)
        (100,    123349)
    };

    \legend{  $C_1$\\$C_2$\\$C_3$\\}

    \end{semilogyaxis}
\end{tikzpicture}
		}
       \label{fig:otherScalabilityPEDA}%
       }
 \caption{\small Scalability Test of PEDA}%
 \label{fig:ScalabilityTestPEDA}
\end{figure}

Query response time is affected by both the increase in data size (which is $1x \rightarrow 10x$ in these experiemnts) and the query type. For star queries, the query response time increases proportional to the data size, as shown in Figures \ref{fig:starScalabilityPECA} and \ref{fig:starScalabilityPEDA}. For other query types, the query response times may grow faster than the data size. Especially for $F_3$, the query response time increases 30 times as the datasize increases 10 times. This is because the complex query graph shape causes more complex operations in query processing, such as joining and assembly. However, even for complex queries, the query performance is scalable with RDF graph size on the benchmark datasets.

Note that, as mentioned in Exp. 1, there is no assembly process for star queries, since matches of a star query cannot cross two fragments. Therefore, the query response times for star queries in centralized and distributed assembly are the same. In contrast, for other query types, some local partial matches and crossing matches result in differences between the performances of centralized and distributed assembly. Here, $L_3$, $L_4$ and $C_3$ are special cases. Although they are not star queries, there are few local partial matches for $L_3$, $L_4$ and $C_3$. Furthermore, the crossing match number is 0 in $L_3$, $L_4$ and $C_3$ (in Table \ref{table:queriesperformanceWatDiv}). Therefore, the assembly times for $L_3$, $L_4$ and $C_3$ are so small that the query response times in both centralized and distributed assembly are the almost same.

\textbf{Exp 4: Intermediate Result Size and Query Performance VS. Query Decomposition Approaches}. Table \ref{table:NumberDifferentPartitionStrategies} compares the number of intermediate results in our method with  two typical query decomposition approaches, i.e., GraphPartition and TripleGroup. We use undirected 1-hop guarantee for GraphPartition and 1-hop bidirection semantic hash partition for TripleGroup. The dataset is still WatDiv 1B.

A star query has no intermediate results, so each method can be answered at each fragment locally. Thus, all methods have the same response time, as shown in Table \ref{table:timeDifferentPartitionStrategies} ($S_1$ to $S_6$).

For other query types, both GraphPartition and TripleGroup need to decompose them into several star subqueries, and find these subquery matches (in each fragment) as intermediate results. Neither GraphPartition nor TripleGroup distinguish the star subquery matches that contribute to crossing matches from those that contribute to inner matches -- all star subquery matches are involved in the assembly process. However, in our method, only local partial matches are involved in the assembly process, leading to lower communication cost and the assembly computation cost. Therefore, the intermediate results that need to be assembled with others  is smaller in our approach.

More intermediate results typically lead to more assembly time. Furthermore, both GraphPartition and TripleGroup employ MapReduce jobs for assembly, which takes much more time than our method. Table \ref{table:timeDifferentPartitionStrategies} shows that our query response time is faster than others.

 \begin{table}
\scriptsize
\begin{threeparttable}
  \begin{tabular}{|c|r|r|r|}
  \hline
%  & \multicolumn{3}{c|}{PECA $\&$ PEDA} & & \\
 %   \cline{2-4}
  & \tabincell{c}{PECA $\&$ PEDA}  &  GraphPartition & TripleGroup \\
 \hline
   \tabincell{l}{$S_1 - S_{7}$}& 0&  0 & 0  \\
 \hline
%  \tabincell{l}{$S_2$} & 0&  0 & 0 \\
%  \cline{1-4}
%  \tabincell{l}{$S_3$} & 0&  0 & 0 \\
%  \cline{1-4}
%  \tabincell{l}{$S_4$} &0 &  0 & 0 \\
%  \cline{1-4}
%   \tabincell{l}{$S_5$} & 0&  0 & 0 \\
%  \cline{1-4}
%   \tabincell{l}{$S_6$} & 0&  0 & 0 \\
%  \cline{1-4}
%   \tabincell{l}{$S_7$} &0&  0 & 0 \\
%  \cline{1-4}
  \tabincell{l}{$L_1$} &  2 &  249571 &  249598 \\
  \cline{1-4}
  \tabincell{l}{$L_2$} & 794 & 73307  &  79630 \\
  \cline{1-4}
  \tabincell{l}{$L_3 - L_4$} & 0 &  0 & 0  \\
  \cline{1-4}
%  \tabincell{l}{$L_4$} & 0  & 0  &  0 \\
%  \cline{1-4}
  \tabincell{l}{$L_5$} & 1274  &  99363 &  99363 \\
  \cline{1-4}
  \tabincell{l}{$F_1$} & 29& 76228  &  15702  \\
  \cline{1-4}
  \tabincell{l}{$F_2$} & 2184 &  501146 &  1119881 \\
  \cline{1-4}
  \tabincell{l}{$F_3$} &  4065632 & 4515731  &  4515752 \\
  \cline{1-4}
  \tabincell{l}{$F_4$} & 6909 &  132193 &  329426 \\
  \cline{1-4}
  \tabincell{l}{$F_5$} & 92 & 2500773  & 9000762  \\
  \cline{1-4}
  \tabincell{l}{$C_1$} & 161803 &  4551562 &  4451693 \\
  \cline{1-4}
  \tabincell{l}{$C_2$} & 937198 &  1457156 & 2368405  \\
  \cline{1-4}
  \tabincell{l}{$C_3$} &  0 & 0  & 0  \\
  \hline
  \end{tabular}
  \end{threeparttable}
  \vspace{-0.05in}
  \caption{ Number of Intermediate Results of Different Approaches on Different Partitioning Strategies}
  \label{table:NumberDifferentPartitionStrategies}
\end{table}

 \begin{table}
\scriptsize
\begin{threeparttable}
  \begin{tabular}{|c|r|r|r|r|}
  \hline
  & \tabincell{c}{PECA } & \tabincell{c}{ PEDA}& \tabincell{c}{GraphPartition}  &  \tabincell{c}{TripleGroup}\\
 \hline
   \tabincell{l}{$S_1$}& 43803 & 43803   & 43803 & 43803   \\
  \cline{1-5}
  \tabincell{l}{$S_2$} & 74479&   74479& 74479 &  74479 \\
  \cline{1-5}
  \tabincell{l}{$S_3$} & 8087&  8087 & 8087 & 8087  \\
  \cline{1-5}
  \tabincell{l}{$S_4$} & 16520& 16520  & 16520 &16520   \\
  \cline{1-5}
   \tabincell{l}{$S_5$} & 1861&  1861 & 1861 & 1861 \\
  \cline{1-5}
   \tabincell{l}{$S_6$} & 50865& 50865  & 50865 & 50865 \\
  \cline{1-5}
   \tabincell{l}{$S_7$} & 56784&  56784 &  56784&  56784   \\
  \cline{1-5}
  \tabincell{l}{$L_1$} & 15341&  15776 &  40840 &  39570 \\
  \cline{1-5}
  \tabincell{l}{$L_2$} & 1510& 1622  & 36150 & 36420  \\
  \cline{1-5}
  \tabincell{l}{$L_3$} &16889 &  16889 &  16889 &  16889  \\
  \cline{1-5}
  \tabincell{l}{$L_4$} &261 &  261 &  261 &  261 \\
  \cline{1-5}
  \tabincell{l}{$L_5$} &48627 & 49539  & 57550 &57480  \\
  \cline{1-5}
  \tabincell{l}{$F_1$} & 64708 & 64748  & 66230 & 66200  \\
  \cline{1-5}
  \tabincell{l}{$F_2$} & 205566&  207725 & 240700 & 248180  \\
  \cline{1-5}
  \tabincell{l}{$F_3$} & 6015341& 4831257  & 6244000 &  6142800 \\
  \cline{1-5}
  \tabincell{l}{$F_4$} & 265239&  260410 &340540  &  340600 \\
  \cline{1-5}
  \tabincell{l}{$F_5$} &25238 & 29208  &52180  &  91110 \\
  \cline{1-5}
  \tabincell{l}{$C_1$} &216059 &  212129 &  216720&223670 \\
  \cline{1-5}
  \tabincell{l}{$C_2$} & 1842906&  1787692 & 1954800 & 2168300   \\
  \cline{1-5}
  \tabincell{l}{$C_3$} & 123349&  123349 &  123349&  123349  \\
  \hline
  \end{tabular}
  \end{threeparttable}
  \vspace{-0.05in}
  \caption{ Query Response Time of Different Approaches (in milliseconds)}
  \label{table:timeDifferentPartitionStrategies}
\end{table}

\nop{
In order to further evaluate the impact of MapReduce jobs in the comparison, we revise the join process in GraphPartition and TripleGroup, by applying MPI instead of MapReduce jobs to join intermediate results. Although the performance of GraphPartition and TripleGroup improve by using MPI, they are still slower than our method. We report the detailed experiment results and analysis in Appendix \ref{sec:exp4cont}.}

Existing partition-based solutions, such as GraphPartition and TripleGroup, use MapReduce jobs to join intermediate results to find SPARQL matches. In order to evaluate the cost of MapReduce-jobs, we perform the following experiments over WatDiv 100M. We revise join processing in both GraphPartition and TripleGroup, by applying joins where intermediate results are sent to a central server using MPI.  We use WatDiv 100M and only consider the benchmark queries that need join processing ($L_1$, $L_2$, $L_5$, $F_1$, $F_2$, $F_3$, $F_3$, $F_4$, $F_5$, $C_1$ and $C_2$) in our experiments. Moreover, all partition-based methods generate intermediate results and merge them at a central sever that shares the same framework with PECA, so we only compare them with PECA. The detailed results are given in Appendix \ref{sec:exp4cont}. Our technique is always faster regardless of the use of MPI or MapReduce-based join. This is because  our method produces  smaller intermediate result sets; MapReduce-based join dominates the query cost. our partial evaluation process is more expensive in evaluating local queries than GraphPartition and TripleGroup in many cases. This is easy to understand -- since the subquery structures in GraphPartition and TripleGroup are fixed, such as stars, it is  cheaper to find these local query results than finding local partial matches. Our system generally outperforms GraphPartition and TripleGroup significantly if they use MapReduce-based join. Even when GraphPartition and TripleGroup use distributed joins, our system is still faster than them in most cases (8 out of 10 queries used in this experiment, see Appendix \ref{sec:exp4cont} for details).

\nop{
As mentioned in Proposition \ref{proposition:optimal}, no local partial match can be missed if we wish to avoid false negatives. If a local partial match is a star, it is also found at GraphPartition and TripleGroup as the intermediate results. If a local partial match is not a star, it is straightforward to know that it can be obtained by joining several star subquery matches.
}

%We provide the detailed analysis in Appendix \ref{sec:ExperimentalAnalysis}.

\nop{
their matches may cross multiple fragments. Hence, GraphPartition and TripleGroup should first decompose the queries into many subqueries and find out the intermediate results of each subquery. There usually exist some edges/vertices in some intermediate results that do not contribute to any final results. Therefore, GraphPartition and TripleGroup always have more intermediate results than our methods. Especially for the partitioning in the exponential distribution, since most vertices are in the same partitioning, there are more inner matches and fewer crossing matches. Then, our methods generate a few local partial matches as intermediate results.

Note that, for $L_1$, $L_2$, $L_5$ and $F_3$, the decomposition results of GraphPartition and TripleGroup are the same, so their running time and intermediate results are also the same.

}

\nop{
In addition, GraphPartition and TripleGroup need to join the intermediate results together. Each join between two intermediate results of subqueries starts a MapReduce job. This takes much more time than our methods.
}

\textbf{Exp 5: Performance on RDF Datasets with One Billion Triples}. This experiment is a comparative evaluation of our method against GraphPartition, TripleGroup and EAGRE on three very large RDF datasets with more than one billion triples, WatDiv 1B, LUBM 10000 and BTC. Figure \ref{fig:online_comparison} shows the performance of different approaches.

Note that, almost half of the queries ($S_1$, $S_2$, $S_3$, $S_4$, $S_5$, $S_6$, $S_7$, $L_3$, $L_4$ and $C_3$ in WatDiv, $Q_2$, $Q_4$ and $Q_5$ in LUBM, $Q_1$, $Q_2$ and $Q_3$ in BTC) have no intermediate results generated in any of the approaches. For these queries, the response times of our approaches and partition-based approaches are the same. However, for other queries, the gap between our approach and others is significant. For example, $L_2$ in WatDiv, for $Q_3$, $Q_6$ and $Q_7$ in LUBM and $Q_3$, $Q_4$, $Q_6$ and $Q_5$ in BTC,  our approach outperforms others one or more orders of magnitudes. We already explained the reasons for GraphPartition and TripleGroup in Exp 4; reasons for EAGRE performance follows.

\nop{
We find out that our method outperform other methods greatly in most cases. For example, in BTC, our method is faster than other methods by an order of magnitude for queries $Q_4$, $Q_5$ and $Q_6$.
}

EAGRE stores all triples as flat files in HDFS and answers SPARQL queries by scanning the files. Because HDFS does not provide fine-grained data access, a query can only be evaluated by a full scan of the files followed by a MapReduce job to join the intermediate results. Although EAGRE proposes some techniques to reduce I/O and data processing, it is still very costly. In contrast, we use graph matching to answer queries, which avoids scanning the whole dataset.

\textbf{Exp 6: Impact of Different Partitioning Strategies}. In this experiment, we test the performance under three different partitioning strategies over WatDiv 100M. The impact of different partitioning strategies is shown in Table \ref{table:newtimeDifferentPartitionStrategies}. We implement three partitioning strategies: uniformly distributed hash partitioning, exponentially distributed hash partitioning, and minimum-cut graph partitioning.

The first partitioning strategy uniformly hashes a vertex $v$ in RDF graph $G$ to a fragment (machine). Thus,  fragments on different machines have approximately the same size. The second  strategy uses an exponentially distributed hash function with a rate parameter pf 0.5. Each vertex $v$ has a probability of $0.5^k$ to be assigned to  fragment (machine) $k$. This partitioning strategy results in skewed fragment sizes. Finally, we use min-cut based partitioning strategy (i.e., METIS algorithm) to partition graph $G$.

Minimum-cut partitioning strategy generally leads to fewer crossing edges than the other two. Thus, it beats the other two approaches in most cases, especially in complex queries (such as $F$ and $C$ category queries). For example, in $C_2$, the minimum-cut is faster than the uniform partitioning by more than four times. For star queries (i.e., $S$ category queries), since there exist no crossing matches, the uniform partitioning  and minimum-cut partitioning have the similar performance. Sometimes, the uniform partitioning is better, but the performance gap is very small.  Due to the skew in fragment sizes, exponentially distributed hashing has worse performance, in most cases, than uniformly distributed hashing.

\nop{Although both the uniform and the exponential distribution partitioning employ the hash partitioning strategy, the latter one results in un-balanced partitioning. Obviously, the whole query response time in the exponential distribution partitioning depends on the largest fragment. Thus, in most cases, the exponential distribution partitioning is not as good as the uniform one in the online query performance.
}
Although our partial evaluation-and-assembly framework is agnostic to the particular partitioning strategy, it is clear that it works better when fragment sizes are balanced, and the  crossing edges are minimized. Many heuristic minimum-cut graph partitioning algorithms (a typical one is METIS [31]) satisfy the requirements.

\begin{table}
\scriptsize
\centering
  \begin{tabular}{|p{0.1cm}|c|r|r|r|}
  \hline
  & &\tabincell{c}{Uniform}&	\tabincell{c}{Exponential}& \tabincell{c}{Min-cut}\\
  \hline
   \multirow{3}{*}{\tabincell{c}{$S_1$}} & PECA&4095	&	7472& \textbf{3210} \\
  \cline{2-5}
   & PEDA&4095	&	7472& \textbf{3210}\\
  \hline
   \multirow{3}{*}{\tabincell{c}{$S_2$}} & PECA&5910	&5830	&\textbf{5053}\\
  \cline{2-5}
   & PEDA&5910	&	5830	&\textbf{5053}\\
  \cline{1-5}
   \multirow{3}{*}{\tabincell{c}{$S_3$}} & PECA&\textbf{869}	& 2003 &1098		\\
  \cline{2-5}
   & PEDA&\textbf{869}	&	2003 &1098\\
  \cline{1-5}
   \multirow{3}{*}{\tabincell{c}{$S_4$}} & PECA&\textbf{1506}&	1532&1525	\\
  \cline{2-5}
 & PEDA&\textbf{1506}&	1532&1525\\
  \cline{1-5}
  \multirow{3}{*}{\tabincell{c}{$S_5$}} & PECA&\textbf{208}&	384&	255	\\
  \cline{2-5}
   & PEDA&\textbf{208}&	384&	255\\
  \cline{1-5}
  \multirow{3}{*}{\tabincell{c}{$S_6$}} & PECA&5153	&5642	&\textbf{4145}	\\
  \cline{2-5}
   & PEDA&5153&	5642	&\textbf{4145}\\
  \cline{1-5}
  \multirow{3}{*}{\tabincell{c}{$S_7$}} & PECA&5047&	5720 &\textbf{4085}	\\
  \cline{2-5}
  & PEDA&5047	&	5720 &\textbf{4085}\\
  \cline{1-5}
  \multirow{3}{*}{\tabincell{c}{$L_1$}} & PECA&\textbf{2301}	&4271	& 3162 \\
  \cline{2-5}
   & PEDA&\textbf{2325}	&	4296&  3168\\
  \hline
   \multirow{3}{*}{\tabincell{c}{$L_2$}} & PECA&271	&	502&\textbf{261}\\
  \cline{2-5}
   & PEDA&339	&	505&  \textbf{297}\\
  \cline{1-5}
   \multirow{3}{*}{\tabincell{c}{$L_3$}} & PECA&\textbf{1115}	& 2122 &	1334	\\
  \cline{2-5}
   & PEDA&\textbf{1115}	&	2122 &	1334\\
  \cline{1-5}
   \multirow{3}{*}{\tabincell{c}{$L_4$}} & PECA&37&54	&	\textbf{27}\\
  \cline{2-5}
 & PEDA&37&	54	&	\textbf{27}\\
  \cline{1-5}
  \multirow{3}{*}{\tabincell{c}{$L_5$}} & PECA&7741&6736	&	\textbf{4984}\\
  \cline{2-5}
   & PEDA&7863&	6946&  \textbf{5163}\\
  \cline{1-5}
  \multirow{3}{*}{\tabincell{c}{$F_1$}} & PECA&5754	&	7889 & \textbf{4386} \\
  \cline{2-5}
   & PEDA&5768	&	7943& \textbf{4415} \\
  \hline
   \multirow{3}{*}{\tabincell{c}{$F_2$}} & PECA&11809	&	16461 &\textbf{10209} \\
  \cline{2-5}
   & PEDA&11832	&	16598&  \textbf{10539}\\
  \cline{1-5}
   \multirow{3}{*}{\tabincell{c}{$F_3$}} & PECA&246277	& 155064 &	\textbf{122539}\\
  \cline{2-5}
   & PEDA&163642	&	115214&   \textbf{103618}\\
  \cline{1-5}
   \multirow{3}{*}{\tabincell{c}{$F_4$}} & PECA&26439& 37608	&	\textbf{21979}\\
  \cline{2-5}
 & PEDA&26421&	36817& \textbf{ 22030	}\\
  \cline{1-5}
  \multirow{3}{*}{\tabincell{c}{$F_5$}} & PECA&11630&	16433&	\textbf{8735}\\
  \cline{2-5}
   & PEDA&11654&	16501& \textbf{ 8262}\\
  \cline{1-5}
  \multirow{3}{*}{\tabincell{c}{$C_1$}} & PECA&14980	&30271		& \textbf{14131 }\\
  \cline{2-5}
   & PEDA&14667	&	29861&  \textbf{13807}\\
  \hline
   \multirow{3}{*}{\tabincell{c}{$C_2$}} & PECA&147962	&	105926 &\textbf{36038}\\
  \cline{2-5}
   & PEDA&147406	&	104084&  \textbf{35220}\\
  \cline{1-5}
   \multirow{3}{*}{\tabincell{c}{$C_3$}} & PECA&\textbf{11631}	&16368  &	13959\\
  \cline{2-5}
   & PEDA&\textbf{11631}	&	16368  &	13959\\
  \cline{1-5}
  \end{tabular}
  \caption{ Query Response Time under Different Partitioning Strategies(in milliseconds)}
  \label{table:newtimeDifferentPartitionStrategies}
\end{table}

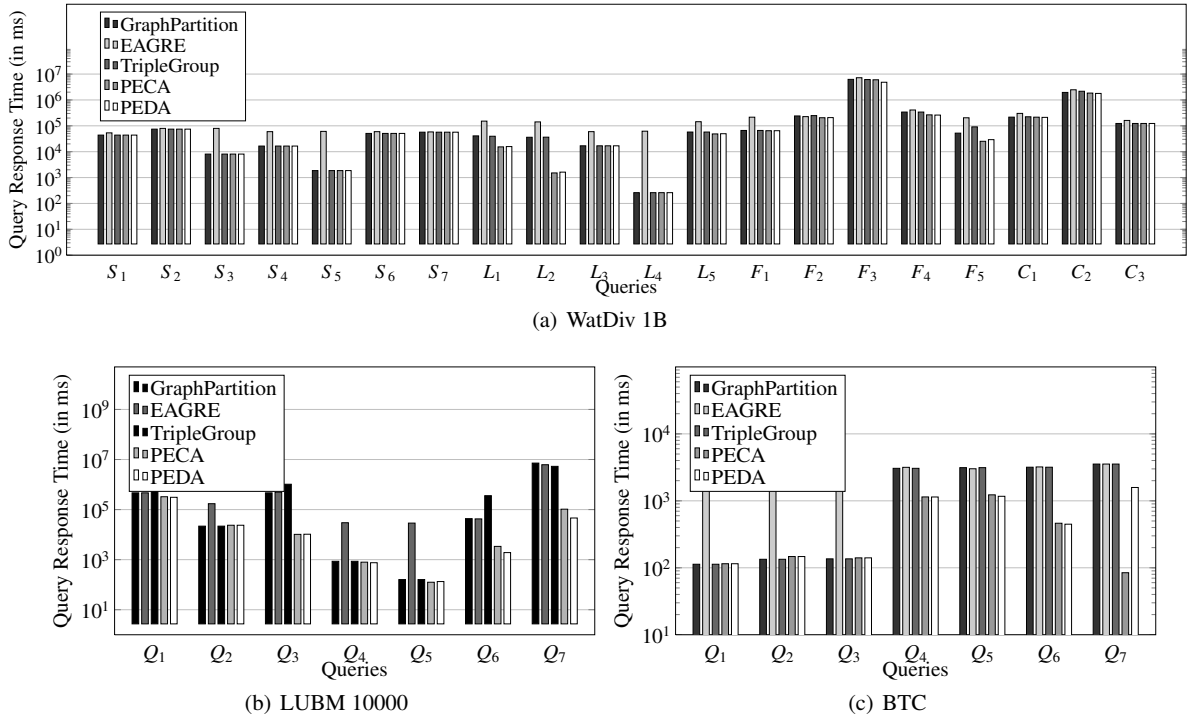
\begin{figure*}%
  \hspace{-0.1in}
       \subfigure[WatDiv 1B]{%
		\resizebox{0.9\columnwidth}{!}{
				\begin{tikzpicture}[font=\Large]
 		 \begin{semilogyaxis}[
                anchor={(0,100)},
               width = 30cm,
               height = 8cm,
    			major x tick style = transparent,
    			ybar,
    			ymin = 1,
    			ymax = 5000000000,
    ytick = {1,10,100,1000,10000,100000,1000000,10000000},
   			ymajorgrids = true,
   			ylabel = {Query Response Time (in ms)},
    			xlabel = {Queries},
    			symbolic x coords = {$S_1$,$S_2$,$S_3$,$S_4$,$S_5$,$S_6$,$S_7$,$L_1$,$L_2$,$L_3$,$L_4$,$L_5$,$F_1$,$F_2$,$F_3$,$F_4$,$F_5$,$C_1$,$C_2$,$C_3$},
    			scaled y ticks = true,
			bar width=4pt,
             enlarge x limits=0.05,
			legend pos= north west,
 legend cell align=left
   		]
   \nop{
   \addplot [fill=white, postaction={pattern=north east lines}] coordinates {($S_1$,859 ) ($S_2$, 80724) ($S_3$, 65449) ($S_4$, 1220) ($S_5$, 3743) ($S_6$, 379) ($S_7$, 111) ($L_1$, 584) ($L_2$, 9689) ($L_3$, 241) ($L_4$, 10757) ($L_5$, 9361) ($F_1$, 342) ($F_2$, 14153) ($F_3$, 20701) ($F_4$, 8872) ($F_5$, 1686) ($C_1$, 43418) ($C_2$, 21979) ($C_3$, 809446)};
}
   \addplot [fill=black!80] coordinates {($S_1$, 43803) ($S_2$, 74479) ($S_3$, 8087) ($S_4$, 16520) ($S_5$, 1861) ($S_6$, 50865) ($S_7$, 56784) ($L_1$, 40840) ($L_2$, 36150) ($L_3$, 16889) ($L_4$, 261) ($L_5$, 57550) ($F_1$, 66230) ($F_2$, 240700) ($F_3$, 6244000) ($F_4$, 340540) ($F_5$, 52180) ($C_1$, 216720) ($C_2$, 1954800) ($C_3$, 123349)};

   \addplot [fill=black!20] coordinates {($S_1$, 53447.5) ($S_2$, 79531) ($S_3$, 79597.5) ($S_4$, 59530.5) ($S_5$, 60996) ($S_6$, 59474) ($S_7$, 58084.5) ($L_1$, 151907) ($L_2$, 141651) ($L_3$, 59539.5) ($L_4$, 61999) ($L_5$, 144243) ($F_1$, 215754.5) ($F_2$, 226682.5) ($F_3$, 7221303.5) ($F_4$, 408847.5) ($F_5$,204495) ($C_1$,  302112.5) ($C_2$, 2472111.5) ($C_3$, 161477)};

    		\addplot [fill=black!60] coordinates {($S_1$, 43803) ($S_2$, 74479) ($S_3$, 8087) ($S_4$, 16520) ($S_5$, 1861) ($S_6$, 50865) ($S_7$, 56784) ($L_1$, 39570) ($L_2$, 36420) ($L_3$, 16889) ($L_4$, 261) ($L_5$, 57480) ($F_1$, 66200) ($F_2$, 248180) ($F_3$, 6140000) ($F_4$, 340600) ($F_5$, 91110) ($C_1$, 223670) ($C_2$, 2168300) ($C_3$, 123349)};

    \addplot [fill=black!40] coordinates {($S_1$, 43803) ($S_2$, 74479) ($S_3$, 8087) ($S_4$, 16520) ($S_5$, 1861) ($S_6$, 50865) ($S_7$, 56784) ($L_1$, 15341) ($L_2$, 1510) ($L_3$, 16889) ($L_4$, 261) ($L_5$, 48627) ($F_1$, 64708) ($F_2$, 205566) ($F_3$, 6015341) ($F_4$, 265239) ($F_5$, 25238) ($C_1$, 216059) ($C_2$, 1842906) ($C_3$, 123349)};

    \addplot [fill=white] coordinates {($S_1$, 43803) ($S_2$, 74479) ($S_3$, 8087) ($S_4$, 16520) ($S_5$, 1861) ($S_6$, 50865) ($S_7$, 56784) ($L_1$, 15776) ($L_2$, 1622) ($L_3$, 16889) ($L_4$, 261) ($L_5$, 49539) ($F_1$, 64748) ($F_2$, 207725) ($F_3$, 4831257) ($F_4$, 260410) ($F_5$, 29208) ($C_1$, 212129) ($C_2$, 1787692) ($C_3$, 123349)};

   		 \legend{GraphPartition,EAGRE,TripleGroup,PECA,PEDA}
  		\end{semilogyaxis}
\end{tikzpicture}
		}
       \label{fig:watdivcomparison}%
       }%
         \\
   \subfigure[LUBM 10000]{%
		\resizebox{0.42\columnwidth}{!}{
				\begin{tikzpicture}[font=\Large]
 		 \begin{semilogyaxis}[
                anchor={(0,100)},
               width = 13cm,
               height = 8cm,
    			major x tick style = transparent,
    			ybar,
    			ymin = 1,
    			ymax = 50000000000,
   			ymajorgrids = true,
   			ylabel = {Query Response Time (in ms)},
    			xlabel = {Queries},
    			symbolic x coords = {$Q_1$,$Q_2$,$Q_3$,$Q_4$,$Q_5$,$Q_6$,$Q_7$},
    			scaled y ticks = true,
			bar width=4.5pt,
             enlarge x limits=0.1,
			legend pos= north west,
 legend cell align=left
   		]
\nop{
   \addplot [fill=white, postaction={pattern=north east lines}] coordinates {($Q_1$, 1084047) ($Q_2$, 81373) ($Q_3$, 72257) ($Q_4$, 7) ($Q_5$, 6) ($Q_6$, 355) ($Q_7$,116325 )};
}

   \addplot [fill=black] coordinates {($Q_1$, 466250) ($Q_2$, 21810) ($Q_3$, 466250) ($Q_4$, 860) ($Q_5$, 160) ($Q_6$, 43460) ($Q_7$, 7194000)};

    		\addplot [fill=black!60] coordinates {($Q_1$, 466250) ($Q_2$, 171810) ($Q_3$, 494900) ($Q_4$, 29860) ($Q_5$, 28860) ($Q_6$, 42460) ($Q_7$, 6194000)};

    \addplot [fill=black] coordinates {($Q_1$, 2257850) ($Q_2$, 21810) ($Q_3$, 1029300) ($Q_4$, 860) ($Q_5$, 160) ($Q_6$, 359040) ($Q_7$, 5333240)};

    \addplot [fill=black!30] coordinates {($Q_1$,326167 ) ($Q_2$, 23685.9) ($Q_3$, 10239) ($Q_4$, 793.71) ($Q_5$, 125.279) ($Q_6$, 3388) ($Q_7$, 103779)};

    \addplot [fill=white] coordinates {($Q_1$,309361.319 ) ($Q_2$, 23685.9) ($Q_3$, 10368) ($Q_4$, 753.998) ($Q_5$, 133.387) ($Q_6$, 1914) ($Q_7$, 46123)};

   		  \legend{GraphPartition,EAGRE,TripleGroup,PECA,PEDA}
  		\end{semilogyaxis}
\end{tikzpicture}
		}
       \label{fig:lubm10000comparison}%
       }%
      % \\\vspace{-0.1in}
\subfigure[BTC]{%
		\resizebox{0.42\columnwidth}{!}{
				\begin{tikzpicture}[font=\Large]
 		 \begin{semilogyaxis}[
               width = 13cm,
               height = 8cm,
    			major x tick style = transparent,
    			ybar,
    			ymin = 10,
    			ymax = 100000,
   			ymajorgrids = true,
   			ylabel = {Query Response Time (in ms)},
    			xlabel = {Queries},
    			symbolic x coords = {$Q_1$,$Q_2$,$Q_3$,$Q_4$,$Q_5$,$Q_6$,$Q_7$},
    			ytick = {10,100,1000,10000},
   yticklabels={$10^1$,$10^2$,$10^3$,$10^4$},
    			scaled y ticks = true,
			bar width=4.5pt,
             enlarge x limits=0.1,
			legend pos= north west,
 legend cell align=left
   		]
    \nop{		
 \addplot [fill=white, postaction={pattern=north east lines}] coordinates {($Q_1$, 52.5) ($Q_2$, 26.5) ($Q_3$, 17.4) ($Q_4$, 52.0) ($Q_5$, 65.0) ($Q_6$, 34.7) ($Q_7$, 93.3 )};
}
    \addplot [fill=black!80] coordinates {($Q_1$, 113) ($Q_2$, 134) ($Q_3$, 136) ($Q_4$, 3058) ($Q_5$, 3131) ($Q_6$, 3178) ($Q_7$, 3541) };

    \addplot [fill=black!20] coordinates {($Q_1$, 2934) ($Q_2$, 2956) ($Q_3$, 2957) ($Q_4$, 3165) ($Q_5$, 3013) ($Q_6$, 3208) ($Q_7$, 3542)    };

    \addplot [fill=black!60] coordinates {($Q_1$, 113) ($Q_2$, 134) ($Q_3$, 136) ($Q_4$, 3058) ($Q_5$, 3131) ($Q_6$, 3178) ($Q_7$, 3541)};

    		\addplot [fill=black!40] coordinates {($Q_1$, 115.011) ($Q_2$, 147.279) ($Q_3$, 140.736) ($Q_4$, 1143.66) ($Q_5$, 1228.86) ($Q_6$, 462.784) ($Q_7$, 84.4228) };

    \addplot [fill=white] coordinates {($Q_1$, 115.002) ($Q_2$, 147.2) ($Q_3$, 140.736) ($Q_4$, 1141.58) ($Q_5$, 1169.01) ($Q_6$, 447.922) ($Q_7$, 1582.75) };

   		  \legend{GraphPartition,EAGRE,TripleGroup,PECA,PEDA}
  		\end{semilogyaxis}
\end{tikzpicture}
		}
       \label{fig:btccomparison}%
       }%
 \caption{\small Online Performance Comparison}%
 \label{fig:online_comparison}
\end{figure*}

\textbf{Exp 7: Comparing with Memory-based Distributed RDF Systems}. We compare our approach (which is disk-based) against TriAD \cite{SIGMOD2014:TriAD} and Trinity.RDF \cite{VLDB13:Trinity} that are memory-based distributed systems. To enable fair comparison, we cache the whole RDF graph together with the corresponding index into memory.  Experiments show that our system is faster than Trinity.RDF and TriAD in these benchmark queries. Results are given in Appendix \ref{sec:memory-exp}.
 %\ref{table:ComparisonMemorySystems} shows that

\nop{
\begin{table}
\scriptsize
\centering
  \begin{tabular}{|c|c|r|r|r|r|}
  \hline
   &\tabincell{c}{Trinity.RDF} & \tabincell{c}{TriAD}&	\tabincell{c}{PECA}& \tabincell{c}{PEDA}\\
  \hline
   $Q_1$ & 281	&	97& 53  & 51\\
  \hline
  $Q_2$ & 132	&	140& 73  & 73\\
  \hline
  $Q_3$ & 110	&	31& 27  & 28\\
  \hline
  $Q_4$ & 5	&	1& 1  & 1\\
  \hline
  $Q_5$ & 4	&	0.2& 0.18  & 0.18\\
  \hline
  $Q_6$ & 9	&	1.8& 1.2  & 1.7\\
  \hline
  $Q_7$ & 630 &	711& 123  & 122\\
  \hline
  \end{tabular}
  \caption{ Comparison with Memory-based Distributed RDF Systems in LUBM 1000}
  \label{table:ComparisonMemorySystems}
\end{table}
}

\textbf{Exp 8: Comparing with Federated SPARQL Systems}.
In this experiment, we compare our methods with some federated SPARQL query systems including (FedX \cite{DBLP:FedX} and SPLENDID \cite{DBLP:SPLENDID}). We evaluate our methods on the standardized benchmark for federated SPARQL query processing, FedBench \cite{DBLP:FedBench}. Results are given in Appendix \ref{sec:feder-exp}.

\textbf{Exp 9: Comparing with Centralized RDF Systems}. In this experiment, we compare our method with RDF-3X in LUBM 10000. Table \ref{table:ComparisonCentralizedSystems} shows the results.

Our method is generally faster than RDF-3X when a query graph is complex, such as $Q_1$, $Q_2$, $Q_3$ and $Q_7$. Since these queries do not contain selective triple patterns and the query graph structure is complex,  the search space for these queries is very large. Our method can take advantage of  parallel processing and reduce  query response time significantly relative to a  centralized system. If the queries ($Q_4$, $Q_5$ and $Q_6$) contain selective triple patterns,  the search space is small. The centralized system (RDF-3X) is faster than our method in these queries, since our approach spends more communication cost between different machines. These queries only spend less than 1-3 seconds in both RDF-3X and our distributed system. However, for some challenging queries (such as $Q_1$, $Q_2$, $Q_3$ and $Q_7$), our method outperforms RDF-3X significantly. For example, RDF-3X spends about $1000$ seconds in $Q_1$, while our approach only spends about $300$ seconds. The performance advantage of our distributed system is more clear in these challenging queries.

\begin{table}
\scriptsize
\centering
  \begin{tabular}{|r|r|r|r|}
  \hline
  & \tabincell{c}{RDF-3X}	& \tabincell{c}{PECA} & \tabincell{c}{PEDA} 	\\
  \hline
  $Q_1$ & 1084047      & 326167 & 309361	\\
  \hline
     $Q_2$& 81373     & 23685 & 23685     \\
  \hline
    $Q_3$ & 72257    & 10239 & 10368    \\
    \hline
    $Q_4$ & 7    & 753 & 753     \\
  \hline
     $Q_5$& 6	& 125 & 125	\\
  \hline
    $Q_6$ & 355	& 3388 & 1914	\\
  \hline
     $Q_7$& 146325	& 143779 & 46123	\\
  \hline
  \end{tabular}
  \caption{Comparison with Centralized System (in milliseconds)}
  \label{table:ComparisonCentralizedSystems}
\end{table}

\nop{
\begin{figure}%
   \centering
		\resizebox{0.75\columnwidth}{!}{
			\input{exps/LUBM_centralized_comparison}
		}
 \caption{\small Comparison with Centralized System}%
 \vspace{-0.2in}
 \label{fig:ComparisonCentralizedSystems}
\end{figure}
}

%!TEX root =  distributedgStore.tex

\section{Conclusion}\label{sec:conclusion}

In this paper, we propose a graph-based approach to distributed SPARQL query processing that adopts the partial evaluation and assembly approach. This is a two-step process. In the first step, we evaluate a query $Q$ on each graph fragment in parallel to find \emph{local partial matches}, which, intuitively, is the overlapping part between a crossing match and a fragment. The second step is to assemble these local partial matches to compute crossing matches. Two different assembly strategies are proposed in this work: \emph{centralized assembly}, where all local partial matches are sent to a single site; and \emph{distributed assembly}, where the local partial matches are assembled at a number of sites in parallel. 

The main benefits of our method are twofold: First, our solution is  partition-agnostic as opposed to existing partition-based methods each of which depends on a particular RDF graph partition strategy, which may be infeasible to  enforce in certain circumstances. Our method is, therefore, much more flexible. Second, compared with other partition-based methods, the number of involved vertices and edges in the intermediate results are minimized in our method, which are proven theoretically and demonstrated experimentally.

There are a number of extensions we are currently working on. An important one is handling SPARQL queries over linked open data (LOD). We can treat the interconnected RDF repositories (in LOD) as a virtually integrated distributed database. Some RDF repositories provide SPARQL endpoints and others may not have query capability. Therefore, data at these sites need to be moved for processing that will affect the algorithm and cost functions. Furthermore, multiple SPARQL query optimization in the context of distributed RDF graphs is also an ongoing work. In real applications, queries in the same time are commonly overlapped. Thus, there is much room for sharing computation when executing these queries. This observation motivates us to revisit the classical problem of multi-query optimization in the context of distributed RDF graphs.

\bibliographystyle{abbrv}
\bibliography{gstore}

\clearpage

\appendix

\section*{Online Supplements}

%--------------------------------------------------------------------------------------------------------

\section{Queries in Experiments}
\label{sec:allqueries}
Table \ref{table:WatDivQueries}, \ref{table:LUBMQueries}, \ref{table:FedBenchCDQueries}, \ref{table:FedBenchLSQueries} and \ref{table:BTCQueries} show all queries used in the paper.

\begin{table}
\scriptsize
\centering
\caption{WatDiv Queries}
\begin{tabular}{|c|m{.8\textwidth}|}
\hline
  $L_1$ & $\#$mapping v1 wsdbm:Website uniform\\
  & SELECT ?v0 ?v2 ?v3 WHERE $\{$ ?v0 wsdbm:subscribes $\%$v1$\%$ . ?v2 sorg:caption ?v3 . ?v0 wsdbm:likes ?v2 . $\}$\\
  \hline
  $L_2$ & $\#$mapping v0 wsdbm:City uniform\\
  & SELECT ?v1 ?v2 WHERE $\{$ $\%$v0$\%$ gn:parentCountry ?v1 . ?v2 wsdbm:likes wsdbm:Product0 . ?v2 sorg:nationality ?v1 . $\}$\\
  \hline
  $L_3$ & $\#$mapping v2 wsdbm:Website uniform\\
  & SELECT ?v0 ?v1 WHERE $\{$ ?v0 wsdbm:likes ?v1 . ?v0 wsdbm:subscribes $\%$v2$\%$ . $\}$\\
  \hline
  $L_4$ & $\#$mapping v1 wsdbm:Topic uniform\\
  & SELECT ?v0 ?v2 WHERE $\{$ ?v0 og:tag $\%$v1$\%$ . ?v0 sorg:caption ?v2 . $\}$\\
  \hline
  $L_5$ & $\#$mapping v2 wsdbm:City uniform\\
  & SELECT ?v0 ?v1 ?v3 WHERE $\{$ ?v0 sorg:jobTitle ?v1 . $\%$v2$\%$ gn:parentCountry ?v3 . ?v0 sorg:nationality ?v3 . $\}$\\
  \hline
  $S_1$ & $\#$mapping v2 wsdbm:Retailer uniform\\
  & SELECT ?v0 ?v1 ?v3 ?v4 ?v5 ?v6 ?v7 ?v8 ?v9 WHERE $\{$ ?v0 gr:includes ?v1 . $\%$v2$\%$ gr:offers ?v0 . ?v0 gr:price ?v3 . ?v0 gr:serial- Number ?v4 . ?v0 gr:validFrom ?v5 . ?v0 gr:validThrough ?v6 . v0 sorg:eligibleQuantity ?v7 . ?v0 sorg:eligibleRegion ?v8 . ?v0 sorg:priceValidUntil ?v9 . $\}$\\
  \hline
  $S_2$ & $\#$mapping v2 wsdbm:Country uniform\\
  & SELECT ?v0 ?v1 ?v3 WHERE $\{$ ?v0 dc:Location ?v1 . ?v0 sorg:nationality $\%$v2$\%$ . ?v0 wsdbm:gender ?v3 . ?v0 rdf:type wsdbm:Role2 . $\}$\\
  \hline
  $S_3$ & $\#$mapping v1 wsdbm:ProductCategory uniform\\
  & SELECT ?v0 ?v2 ?v3 ?v4 WHERE $\{$ ?v0 rdf:type $\%$v1$\%$ . ?v0 sorg:caption ?v2 . ?v0 wsdbm:hasGenre ?v3 . ?v0 sorg:publisher ?v4 . $\}$\\
  \hline
  $S_4$ & $\#$mapping v1 wsdbm:AgeGroup uniform\\
  & SELECT ?v0 ?v2 ?v3 WHERE $\{$ ?v0 foaf:age $\%$v1$\%$ . ?v0 foaf:familyName ?v2 . ?v3 mo:artist ?v0 . ?v0 sorg:nationality wsdbm:Country1 . $\}$\\
  \hline
  $S_5$ & $\#$mapping v1 wsdbm:ProductCategory uniform\\
  & SELECT ?v0 ?v2 ?v3 WHERE $\{$ ?v0 rdf:type $\%$v1$\%$ . ?v0 sorg:description ?v2 . ?v0 sorg:keywords ?v3 . ?v0 sorg:language wsdbm:Language0 . $\}$\\
  \hline
  $S_6$ & $\#$mapping v3 wsdbm:SubGenre uniform\\
  & SELECT ?v0 ?v1 ?v2 WHERE $\{$ ?v0 mo:conductor ?v1 . ?v0 rdf:type ?v2 . ?v0 wsdbm:hasGenre $\%$v3$\%$ . $\}$\\
  \hline
  $S_7$ & $\#$mapping v3 wsdbm:User uniform\\
  & SELECT ?v0 ?v1 ?v2 WHERE $\{$ ?v0 rdf:type ?v1 . ?v0 sorg:text ?v2 . $\%$v3$\%$ wsdbm:likes ?v0 . $\}$\\
  \hline
  $F_1$ & $\#$mapping v1 wsdbm:Topic uniform\\
  & SELECT ?v0 ?v2 ?v3 ?v4 ?v5 WHERE $\{$ ?v0 og:tag $\%v1\%$ . ?v0
rdf:type ?v2 . ?v3 sorg:trailer ?v4 . ?v3 sorg:keywords ?v5 . ?v3 wsdbm:hasGenre ?v0 . ?v3 rdf:type wsdbm:ProductCategory2 . $\}$\\
  \hline
  $F_2$ & $\#$mapping v8 wsdbm:SubGenre uniform\\
  & SELECT ?v0 ?v1 ?v2 ?v4 ?v5 ?v6 ?v7 WHERE $\{$ ?v0 foaf:homepage ?v1 . ?v0 og:title ?v2 . ?v0 rdf:type ?v3 . ?v0 sorg:caption ?v4 . ?v0 sorg:description ?v5 . ?v1 sorg:url ?v6 . ?v1 wsdbm:hits ?v7 . ?v0 wsdbm:hasGenre $\%$v8$\%$ . $\}$\\
  \hline
  $F_3$ & $\#$mapping v3 wsdbm:SubGenre uniform\\
  & SELECT ?v0 ?v1 ?v2 ?v4 ?v5 ?v6 WHERE $\{$ ?v0 sorg:contentRating ?v1 .
?v0 sorg:contentSize ?v2 .
?v0 wsdbm:hasGenre $\%$v3$\%$ .
?v4 wsdbm:makesPurchase ?v5 .
?v5 wsdbm:purchaseDate ?v6 .
?v5 wsdbm:purchaseFor ?v0 .
$\}$\\
  \hline
  $F_4$ & $\#$mapping v3 wsdbm:Topic uniform\\
& SELECT ?v0 ?v1 ?v2 ?v4 ?v5 ?v6 ?v7 ?v8 WHERE $\{$ ?v0 foaf:home- page ?v1 .
?v2 gr:includes ?v0 .
?v0 og:tag $\%$v3$\%$ .
?v0 sorg:description ?v4 .
?v0 sorg:contentSize ?v8 .
?v1 sorg:url ?v5 .
?v1 wsdbm:hits ?v6 .
?v1 sorg:language wsdbm:Language0 .
?v7 wsdbm:likes ?v0 .
$\}$\\
  \hline
  $F_5$ & $\#$mapping v2 wsdbm:Retailer uniform\\
& SELECT ?v0 ?v1 ?v3 ?v4 ?v5 ?v6 WHERE $\{$ ?v0 gr:includes ?v1 .
$\%$v2$\%$ gr:offers ?v0 .
?v0 gr:price ?v3 .
?v0 gr:validThrough ?v4 .
?v1 og:title ?v5 .
?v1 rdf:type ?v6 .
$\}$\\
  \hline
  $C_1$ & SELECT ?v0 ?v4 ?v6 ?v7 WHERE $\{$ ?v0 sorg:caption ?v1 .
?v0 sorg:text ?v2 .
?v0 sorg:contentRating ?v3 .
?v0 rev:hasReview ?v4 .
?v4 rev:title ?v5 .
?v4 rev:reviewer ?v6 .
?v7 sorg:actor ?v6 .
?v7 sorg:language ?v8 .
$\}$\\
  \hline
  $C_2$ & SELECT ?v0 ?v3 ?v4 ?v8 WHERE $\{$ ?v0 sorg:legalName ?v1 .
?v0 gr:offers ?v2 .
?v2 sorg:eligibleRegion wsbm:Country5 .
?v2 gr:includes ?v3 .
?v4 sorg:jobTitle ?v5 .
?v4 foaf:homepage ?v6 .
?v4 wsdbm:makesPurchase ?v7 .
?v7 wsdbm:purchaseFor ?v3 .
?v3 rev:has Review ?v8 .
?v8 rev:totalVotes ?v9 .
$\}$\\
  \hline
  $C_3$ & SELECT ?v0 WHERE $\{$ ?v0 wsdbm:likes ?v1 .
?v0 wsdbm:friendOf ?v2 .
?v0 dc:Location ?v3 .
?v0 foaf:age ?v4 .
?v0 wsdbm:gender ?v5 .
?v0 foaf:givenName ?v6 .
$\}$\\
  \hline
    \end{tabular}
  \vspace{-0.1in}
\label{table:WatDivQueries}
\end{table}

\begin{table}
\scriptsize
\centering
\caption{LUBM Queries}
\begin{tabular}{|c|m{.8\textwidth}|}
\hline
{$Q_1$} & SELECT ?x,?y,?z WHERE $\{$ ?x rdf:type ub:GraduateStudent.?y rdf:type ub:University.?z rdf:type ub:Department.?x ub:memberOf ?z.?z ub:subOrganizationOf ?y.?x ub:undergraduateDegreeFrom ?y$\}$\\
  \hline
{$Q_2$} & SELECT ?x WHERE $\{$ ?x rdf:type ub:Course.?x ub:name ?y. $\}$\\
  \hline
{$Q_3$} & SELECT ?x,?y,?z WHERE $\{$ ?x rdf:type ub:UndergraduateStudent. ?y rdf:type ub:University. ?z rdf:type ub:Department. ?x ub:memberOf ?z. ?z ub:subOrganizationOf ?y. ?x ub:undergraduateDegreeFrom ?y$\}$\\
  \hline
{$Q_4$} & SELECT ?x WHERE $\{$ ?x ub:worksFor http://www.Department0. University0.edu. ?x rdf:type ub:FullProfessor. ?x ub:name ?y1. ?x ub: emailAddress ?y2. ?x ub:telephone ?y3. $\}$\\
  \hline
{$Q_5$} & SELECT ?x WHERE $\{$ ?x	ub:subOrganizationOf	http://www. Department0.University0.edu.
?x	rdf:type	ub:ResearchGroup. $\}$\\
  \hline
{$Q_6$} & SELECT ?x,?y WHERE $\{$  ?y	rdf:type	ub:Department.
   ?y	ub:subOrganizationOf	http://www.University0.edu.
?x	ub:worksFor	?y.
?x	rdf:type	ub:FullProfessor. $\}$\\
  \hline
{$Q_7$} & SELECT ?x,?y,?z WHERE $\{$  ?x rdf:type ub:UndergraduateStudent. ?y rdf:type ub:FullProfessor. ?z rdf:type ub:Course. ?x ub:advisor ?y. ?x ub:takesCourse ?z. ?y ub:teacherOf ?z. $\}$\\
\hline
    \end{tabular}
\label{table:LUBMQueries}
\end{table}

\begin{table}
\scriptsize
\centering
\caption{FedBench Queries in Cross Domain}
\begin{tabular}{|c|m{.8\textwidth}|}
\hline
{$CD_1$} & SELECT ?predicate ?object WHERE $\{$
   $\{$ <http://dbpedia.org/resource/Barack$\_$Obama> ?predicate ?object $\}$
   UNION
   $\{$ ?subject <http://www.w3.org/2002/07/owl$\#$sameAs> <http://dbpedia.org/resource/Barack$\_$Obama> .
     ?subject ?predicate ?object $\}$
$\}$\\
\hline
{$CD_2$} & SELECT ?party ?page  WHERE $\{$
   <http://dbpedia.org/resource/Barack$\_$Obama> <http://dbpedia.org/ontology/party> ?party .
   ?x <http:// data.nytimes.com/elements/topicPage> ?page .
   ?x <http://www.w3.org/2002/07/owl$\#$ sameAs> <http://dbpedia.org/resource/Barack$\_$Obama> .
$\}$\\
\hline
{$CD_3$} & SELECT ?president ?party ?page WHERE $\{$
   ?president <http://www.w3.org/1999/02/22-rdf-syntax-ns$\#$type> <http://dbpedia.org/ontology/President> .
   ?president <http://dbpedia.org/ontology/nationality> <http://dbpedia.org/resource/United$\_$States> .
   ?president <http://dbpedia.org/ontology/party> ?party .
   ?x <http://data.nytimes.com/elements/topicPage> ?page .
   ?x <http://www. w3.org/2002/07/owl$\#$sameAs> ?president .
$\}$\\
\hline
{$CD_4$} & SELECT ?actor ?news WHERE $\{$
   ?film <http://purl.org/dc/terms/title> ``Tarzan'' .
   ?film <http://data.linkedmdb.org/resource/movie/actor> ?actor .
   ?actor <http://www.w3.org/2002/07/owl$\#$sameAs> ?x.
   ?y <http://www.w3.org/2002/07/owl$\#$sameAs> ?x .
   ?y <http://data.nytimes.com/elements/topicPage> ?news
$\}$\\
\hline
{$CD_5$} & SELECT ?film ?director ?genre WHERE $\{$
   ?film <http://dbpedia.org/ontology/director>  ?director .
   ?director <http://dbpedia.org/ontology/nationality> <http://dbpedia.org/resource/Italy> .
   ?x <http://www.w3.org/2002/07/owl$\#$sameAs> ?film .
   ?x <http://data.linkedmdb.org/resource/movie/genre> ?genre .
$\}$\\
\hline
{$CD_6$} & SELECT ?name ?location ?news WHERE $\{$
   ?artist <http://xmlns.com/foaf/0.1/name> ?name .
   ?artist <http://xmlns.com/foaf/0.1/based$\_$near> ?location .
   ?location <http://www.geonames.org/ontology$\#$parentFeature> ?germany .
   ?germany <http://www.geonames.org/ontology$\#$name> ``Federal Republic of Germany''
$\}$\\
\hline
{$CD_7$} & SELECT ?location ?news WHERE $\{$
   ?location <http://www. geonames.org/ontology$\#$parentFeature> ?parent .
   ?parent <http://www. geonames.org/ontology$\#$name> ``California'' .
   ?y <http://www.w3.org/ 2002/07/owl$\#$sameAs> ?location .
   ?y <http://data.nytimes.com/elements/ topicPage> ?news
$\}$\\
\hline
{$ECD_1$} & SELECT ?name WHERE $\{$
<http://data.semanticweb.org/conference/ www/2008>	<http://xmlns.com/foaf/0.1/based$\_$near>	?location .
?location	<http://www.geonames.org/ontology$\#$name>	?name .$\}$\\
\hline
{$ECD_2$} & SELECT ?actor WHERE $\{$
?actor	<http://www.w3.org/2002/07/owl$\#$sameAs>	?x .
?actor	<http://www.w3.org/1999/02/22-rdf-syntax-ns$\#$type>	<http://data.linkedmdb.org/resource/movie/actor> . $\}$\\
\hline
{$ECD_{3}$} & SELECT ?m ?c WHERE $\{$ ?a <http://dbpedia.org/ontology/ residence> ?c .
?c <http://www.w3.org/1999/02/22-rdf-syntax-ns$\#$type> <http://dbpedia.org/ontology/Place> .
?a1	<http://www.w3.org/2002/07/owl$\#$sameAs>	?a .
?m <http://data.linkedmdb.org/resource/movie/actor> ?a1 .$\}$\\
\hline
{$ECD_{4}$} & SELECT ?x ?y WHERE $\{$
?gplace	<http://www.geonames.org/ontology$\#$alternateName>	"Philadelphia"@en .
?gplace	<http://www.w3.org/2002/07/owl$\#$sameAs> ?place .
?x <http://dbpedia.org/ontology/birthPlace> ?place .
?x <http://dbpedia.org/ontology/spouse> ?y .
?z <http://dbpedia.org/ontology/starring> ?x .
?z <http://dbpedia.org/ontology/starring> ?y .
?y1	<http://www.w3.org/2002/07/owl$\#$sameAs>	?y .
?y1 <http://www.w3.org/1999/02/22-rdf-syntax-ns$\#$type> <http://data.linkedmdb.org/resource/movie/actor> .$\}$\\
\hline
{$ECD_{5}$} & Select ?n1 ?n2 ?gn where $\{$
?a1	<http://xmlns.com/foaf/0.1/name> ?n1 .
?p2 <http://www.w3.org/2002/07/owl$\#$sameAs> ?a2 .
?p2 <http://data.linkedmdb.org/resource/movie/actor$\_$name> ?n2 .
?a1 <http://dbpedia.org/ontology/award> ?award .
?a2 <http://dbpedia.org/ontology/award> ?award .
?a1 <http://dbpedia.org/ontology/birthPlace> ?city .
?a2	<http://dbpedia.org/ontology/birthPlace> ?city .
?gc	<http://www.w3.org/2002/07/owl$\#$sameAs> ?city .
?gc	<http://www.geonames.org/ontology$\#$name> ?gn .$\}$\\
\hline
\end{tabular}
\label{table:FedBenchCDQueries}
\end{table}

\begin{table}
\scriptsize
\centering
\caption{FedBench Queries in Life Science Domain}
\begin{tabular}{|c|m{.84\textwidth}|}
\hline
{$LS_1$} & SELECT $\$$drug $\$$melt WHERE $\{$
    $\{$ $\$$drug <http://www4.wiwiss.fu-berlin.de/drugbank/resource/drugbank/meltingPoint> $\$$melt. $\}$
    UNION
    $\{$ $\$$drug <http://dbpedia.org/ontology/Drug/meltingPoint> $\$$melt . $\}$
$\}$\\
\hline
{$LS_2$} & SELECT ?predicate ?object WHERE $\{$
    $\{$ <http://www4.wiwiss.fu-berlin.de/drugbank/resource/drugs/DB00201> ?predicate ?object . $\}$
    UNION
    $\{$ <http://www4.wiwiss.fu-berlin.de/drugbank/resource/ drugs/DB00201> <http://www.w3.org/2002/07/owl$\#$sameAs> ?caff .
      ?caff ?predicate ?object . $\}$
$\}$\\
\hline
{$LS_3$} & SELECT ?Drug ?IntDrug ?IntEffect WHERE $\{$ ?Drug <http://www.w3.org/1999/02/22-rdf-syntax-ns$\#$type> <http://dbpedia.org/ontology/Drug> .
    ?y <http://www.w3.org/2002 /07/owl$\#$sameAs> ?Drug .
    ?Int <http://www4.wiwiss.fu-berlin.de/drugbank/resource/drugbank/ interactionDrug1> ?y .
    ?Int <http://www4.wiwiss.fu-berlin.de/drugbank/ resource/drugbank/interactionDrug2> ?IntDrug .
    ?Int <http://www4.wi wiss.fu-berlin.de/drugbank/resource/drugbank/text> ?IntEffect .
$\}$\\
\hline
{$LS_4$} & SELECT ?drugDesc ?cpd ?equation WHERE $\{$ ?drug <http://www4.wiwiss.fu-berlin.de/drugbank/resource/drugbank/drugCategory> <http://www4. wiwiss.fu-berlin.de/drugbank/resource/drugcategory/ cathartics> .
   ?drug <http://www4.wiwiss.fu-berlin.de/drugbank/resource/drugbank/keggCompoundId> ?cpd .
   ?drug <http://www4.wiwiss.fu-berlin.de/drugbank/ resource/drugbank/description> ?drugDesc .
   ?enzyme <http://bio2rdf.org/ns/kegg$\#$xSubstrate> ?cpd .
   ?enzyme <http://www.w3.org/ 1999/02/22-rdf-syntax-ns$\#$type> <http://bio2rdf.org/ns/kegg$\#$Enzyme> .
   ?reaction <http://bio2rdf.org/ns/kegg$\#$xEnzyme> ?enzyme .
   ?reaction <http://bio2rdf.org/ns/kegg$\#$equation> ?equation .
$\}$\\
\hline
{$LS_5$} & SELECT ?drug ?keggUrl ?chebiImage WHERE $\{$
  ?drug <http://www.w3.org/1999/02/22-rdf-syntax-ns$\#$type> <http://www4.wiwiss.fu-berlin.de/drugbank/resource/drugbank/ drugs> .
  ?drug <http://www4.wiwiss.fu-berlin.de/drugbank/ resource/drugbank/keggCompoundId> ?keggDrug .
  ?keggDrug <http://bio2rdf.org/ns/bio2rdf$\#$url> ?keggUrl .
  ?drug <http://www4.wiwiss.fu-berlin.de/drugbank/resource/drugbank/ genericName> ?drugBankName .
  ?chebiDrug <http://purl.org/dc/ elements/1.1/title> ?drugBankName .
  ?chebiDrug <http://bio2rdf.org/ns/ bio2rdf$\#$image> ?chebiImage .
$\}$\\
\hline
{$LS_6$} & SELECT ?drug ?title WHERE $\{$
	 ?drug <http://www4.wiwiss.fu-berlin.de/drugbank/resource/drugbank/drugCategory> <http://www4.wi wiss.fu-berlin.de/drugbank/resource/drugcategory/ micronutrient> .
	 ?drug <http://www4.wiwiss.fu-berlin.de/drug- bank/resource/drugbank/casRegistryNumber> ?id .
	 ?keggDrug <http://www.w3.org/1999/02/22-rdf-syntax-ns$\#$type> <http://bio2rdf.org/ns/kegg$\#$Drug> .
	 ?keggDrug <http://bio2rdf.org/ns/bio2rdf$\#$xRef> ?id .
	 ?keggDrug <http://purl.org/dc /elements/1.1/ title> ?title .
$\}$\\
\hline
{$LS_7$} & SELECT $\$$drug $\$$transform $\$$mass WHERE $\{$
 	$\{$ $\$$drug <http://www4.wiwiss.fu-berlin.de/drugbank/resource/drugbank/ affectedOrganism>  'Humans and other mammals'.
 	  $\$$drug <http://www4.wiwiss.fu-berlin.de/drug bank/resource/drugbank/casRegistryNumber> $\$$cas .
 	  $\$$kegg Drug <http://bio2rdf.org/ns/bio2rdf$\#$xRef> $\$$cas .
 	  $\$$keggDrug <http://bio 2rdf.org/ns/bio2rdf $\#$mass> $\$$mass
 	      FILTER ( $\$$mass > '5' )
 	$\}$
 	  OPTIONAL $\{$ $\$$drug <http://www4.wiwiss.fu-berlin.de/drug bank/resource/drugbank/biotransformation> $\$$trans form . $\}$
$\}$\\
\hline
{$ELS_1$} & SELECT ?num WHERE $\{$
<http://www4.wiwiss.fu-berlin.de/drugbank/resource/drugs/DB00918>	<http://www.w3.org/2002/07/owl$\#$sameAs>	?x .
?x	<http://dbpedia.org/ontology/casNumber>	?num .$\}$\\
\hline
{$ELS_2$} & SELECT ?y WHERE $\{$
?x	<http://bio2rdf.org/ns/bio2rdf$\#$xRef>	?y .
?x	<http://www.w3.org/1999/02/22-rdf-syntax-ns$\#$type> <http://bio2rdf.org/ns/kegg$\#$Compound> .$\}$\\
\hline
{$ELS_{3}$} & SELECT ?n ?t WHERE $\{$
?Drug <http://www.w3.org/1999/02/22-rdf-syntax-ns$\#$type> <http://dbpedia.org/ontology/Drug> .
?y <http://www.w3.org/2002/ 07/owl$\#$sameAs> ?Drug .
?y	<http://www4.wiwiss.fu-berlin.de/ drugbank/resource/drugbank/brandName> ?n .
?y	<http://www4.wiwiss.fu-berlin.de/drugbank/ resource/drugbank/drugType> ?t .$\}$\\
\hline
{$ELS_{4}$} & SELECT ?parent WHERE $\{$
?y 	<http://bio2rdf.org/ns/bio2rdf$\#$formula>	"C8H10NO6P" .
?z	http://bio2rdf.org/ns/chebi$\#$has$\_$functional$\_$parent>	?parent .
?x	<http://www4.wiwiss.fu-berlin.de/drugbank/resource/drug- bank/keggCompoundId>	?y .
?x	<http://www4.wiwiss.fu-berlin.de/drugbank/resource/drugbank/chebiId>	?z .
?y	<http://bio2rdf.org/ns/bio2rdf$\#$xRef>	?z .$\}$\\
\hline
{$ELS_{5}$} & SELECT ?y1 ?label WHERE $\{$
?x1	<http://bio2rdf.org/ns/chebi$\#$has$\_$role>	?x2 .
?x1	<http://bio2rdf.org/ns/chebi$\#$has$\_$role>	?x3 .
?x2	<http://bio2rdf.org/ns/chebi$\#$is$\_$a>	?x4 .
?x3	<http://bio2rdf.org/ns/chebi$\#$is$\_$a>	?x4 .
?y1	<http://www4.wiwiss.fu-berlin.de/drugbank/resource/drugbank/chebiId>	?x1 .
?y1	<http://www.w3.org/1999/02/22-rdf-syntax-ns$\#$type>	<http://www4.wiwiss.fu-berlin.de/drugbank/resource/drugbank/drugs> .
?x4	<http://www.w3.org/2000/01/rdf-schema$\#$label>	?label . $\}$\\
\hline
\end{tabular}
\label{table:FedBenchLSQueries}
\end{table}

\begin{table}
\scriptsize
\centering
\begin{threeparttable}
\begin{tabular}{|c|m{.8\textwidth}|}
\hline
$Q_1$ & SELECT ?x WHERE $\{$ ?x	rdf:type ub:GraduateStudent. ?x ub:takesCourse http://www.Department0.University0.edu/GraduateCourse0$\}$\\
\hline
{$Q_2$} & SELECT ?x,?y where $\{$ ?x rdf:type ub:GraduateStudent. ?y rdf:type ub:University. ?z rdf:type ub:Department. ?x ub:memberOf ?z. ?z ub:subOrganizationOf ?y. ?x ub:undergraduateDegreeFrom ?y$\}$\\
\hline
{$Q_3$} & SELECT ?x WHERE $\{$ ?x rdf:type ub:Publication. ?x ub:publicationAuthor  http://www.Department0.University0.edu/AssistantProfessor0$\}$\\
\hline
{$Q_4$} & SELECT ?x,?y1,?y3 WHERE $\{$ ?x ub:worksFor http://www.Department0.University0.edu. ?x ub:name ?y1. ?x ub:emailAddress ?y2. ?x ub:telephone ?y3.$\}$\\
\hline
\multirow{2}{*}{$Q_5$}\tnote{$*$} & SELECT ?x WHERE $\{$ ?x ub:memberOf http://www.Department0.University0.edu.$\}$\\
\cline{2-2}
&(removing the triple pattern ``?x rdf:type ub:Person'' with type reasoning)\\
\hline
\multirow{2}{*}{$Q_6$}\tnote{$*$} & SELECT ?x WHERE $\{$ ?x rdf:type ub:UndergraduateStudent.$\}$\\
\cline{2-2}
&(modifying the triple pattern ``?x rdf:type ub:Student'' with type reasoning to ``?x rdf:type ub:UndergraduateStudent'')\\
\hline
\multirow{2}{*}{$Q_7$}\tnote{$*$} & SELECT ?x,?y WHERE $\{$ ?x rdf:type ub:UndergraduateStudent. ?x ub:takesCourse ?y. ?y rdf:type
ub:Course. http://www.Department0.University0.edu//AssociateProfessor0 ub:teacherOf ?y.$\}$\\
\cline{2-2}
&(modifying the triple pattern ``?x rdf:type ub:Student'' with type reasoning to ``?x rdf:type ub:UndergraduateStudent'')\\
\hline
\multirow{2}{*}{$Q_8$}\tnote{$*$} & SELECT ?x,?y WHERE $\{$ ?x rdf:type ub:UndergraduateStudent. ?x ub:memberOf ?y. ?x ub:emailAddress ?z. ?y rdf:type ub:Department. ?y ub:subOrganizationOf http://www.University0.edu$\}$\\
\cline{2-2}
&(modifying the triple pattern ``?x rdf:type ub:Student'' with type reasoning to ``?x rdf:type ub:UndergraduateStudent'')\\
\hline
\multirow{2}{*}{$Q_9$}\tnote{$*$} & SELECT ?x,?y WHERE $\{$ ?x rdf:type ub:UndergraduateStudent. ?y rdf:type ub:FullProfessor. ?z rdf:type ub:Course. ?x ub:advisor ?y. ?x ub:takesCourse ?z. ?y ub:teacherOf ?z. $\}$\\
\cline{2-2}
&(modifying the triple pattern ``?x rdf:type ub:Student'' with type reasoning to ``?x rdf:type ub:UndergraduateStudent'' and the triple pattern ``?x rdf:type ub:Faculty'' with type reasoning to ``?x rdf:type ub:FullProfessor'')\\
\hline
\multirow{2}{*}{$Q_{10}$}\tnote{$*$} & SELECT ?x WHERE $\{$ ?x rdf:type ub:GraduateStudent. ?x ub:takesCourse http://www.Department0.University0.edu/GraduateCourse0.$\}$\\
\cline{2-2}
&(modifying the triple pattern ``?x rdf:type ub:Student'' with type reasoning to ``?x rdf:type ub:UndergraduateStudent'')\\
\hline
{$Q_{11}$} & SELECT ?x,?y WHERE $\{$ ?x rdf:type ub:ResearchGroup. ?x ub:subOrganizationOf ?y. ?y ub:subOrganizationOf http://www.University0.edu.$\}$\\
\hline
\multirow{2}{*}{$Q_{12}$}\tnote{$*$} & SELECT ?x,?y WHERE $\{$ ?x ub:headOf ?y. ?y rdf:type ub:Department. ?y ub:subOrganizationOf http://www.University0.edu.$\}$\\
\cline{2-2}
&(removing the triple pattern ``?x rdf:type ub:Chair'' with type reasoning)\\
\hline
\multirow{2}{*}{$Q_{13}$}\tnote{$*$} & SELECT ?x WHERE $\{$ ?x ub:undergraduateDegreeFrom http://www.University0.edu.$\}$\\
\cline{2-2}
&(removing the triple pattern ``?x rdf:type ub:Person'' with type reasoning and modifying the triple pattern ``http://www.University0.edu ub:hasAlumnus ?x'' with property reasoning to ``?x ub:undergraduateDegreeFrom http://www.University0.edu'')\\
\hline
{$Q_{14}$} & SELECT ?x WHERE $\{$ ?x rdf:type ub:UndergraduateStudent.$\}$\\
\hline
\end{tabular}
\begin{tablenotes}
  \scriptsize
  \item[$*$] means that the query removes or modifys some triple patterns with reasoning.
  \end{tablenotes}
  \end{threeparttable}
\caption{Original SPARQL Queries in LUBM}
\label{table:LUBM14Queries}
\end{table}

\begin{table}
\scriptsize
\centering
\caption{BTC Queries}
\begin{tabular}{|c|m{.8\textwidth}|}
\hline
{$Q_1$} & SELECT ?lat,?long WHERE $\{$  ?a [] ``Eiffel Tower"@en. ?a $<$http://www .w3.org/2003/01/geo/wgs84$\_$pos$\#$lat$>$ ?lat. ?a $<$http://www.w3.org/2003 /01/geo/wgs84$\_$pos$\#$long$>$ ?long. ?a $<$http://dbpedia.org/ontology/location$>$  $<$http://dbpedia.org/resource/France$>$.$\}$\\
\hline
{$Q_2$} & SELECT ?b,?p,?bn WHERE $\{$  ?a [] ``Tim Berners-Lee"@en. ?a $<$http:// dbpedia.org/property/dateOfBirth$>$ ?b. ?a $<$http://dbpedia.org/ property/placeOfBirth$>$ ?p. ?a $<$http://dbpedia.org/property/name$>$ ?bn.$\}$\\
\hline
{$Q_3$} & SELECT ?t,?lat,?long WHERE $\{$  ?a $<$http://dbpedia.org/property/wikiLinks$>$ $<$http://dbpedia.org/resource/List$\_$of$\_$World$\_$Heritage$\_$Sites $\_$in$\_$Europe$>$.
?a $<$http://dbpedia.org/property/title$>$ ?t.
?a $<$http://www .w3.org/2003/01/geo/wgs84$\_$pos$\#$lat$>$ ?lat.
?a $<$http://www.w3.org/2003/ 01/geo/wgs84$\_$pos$\#$long$>$ ?long.
?a $<$http://dbpedia.org/property/wikiLinks$>$  $<$http://dbpedia.org/resource/Middle$\_$Ages$>$.$\}$\\
\hline
{$Q_4$} & SELECT ?l,?long,?lat WHERE $\{$ ?p $<$http://dbpedia.org/property/ name$>$	``Krebs, Emil"@en.
?p $<$http://dbpedia.org/property/deathPlace$>$ ?l.
?c [] ?l.
?c $<$http://www.geonames.org/ontology$\#$featureClass$>$  $<$http:// www.geonames.org/ontology$\#$P$>$.
?c $<$http://www.geonames.org/onto- logy$\#$inCountry$>$  $<$http://www.geonames.org/countries/$\#$DE$>$.
?c $<$http:// www.w3.org/2003/01/geo/wgs84$\_$pos$\#$lat$>$ ?lat.
?c $<$http://www.w3.org/ 2003/01/geo/wgs84$\_$pos$\#$long$>$ ?long$.\}$\\
\hline
{$Q_5$} & SELECT distinct ?l,?long,?lat WHERE $\{$  ?a [] ``Barack Obama''@en.
?a $<$http://dbpedia.org/property/placeOfBirth$>$ ?l.
?l $<$http:// www.w3.org/2003/01/geo/wgs84$\_$pos$\#$lat$>$ ?lat.
?l $<$http://www.w3.org/ 2003/01/geo/wgs84$\_$pos$\#$long$>$ ?long.$\}$\\
\hline
{$Q_6$} & SELECT distinct ?d WHERE $\{$  ?a $<$http://dbpedia.org/property/ senators$>$ ?c.
?a $<$http://dbpedia.org/property/name$>$ ?d.
?c $<$http://dbpedia.org/property/profession$>$ $<$http://dbpedia.org/resource/Veterinarian$>$.
?a $<$http://www.w3.org/2002/07/owl$\#$sameAs$>$ ?b.
?b $<$http://www.geonames.org/ontology$\#$inCountry$>$  $<$http://www.geonames.org/countries/$\#$US$>$.$\}$\\
\hline
{$Q_7$} & SELECT distinct ?a,?b,?lat,?long WHERE $\{$  ?a $<$http://dbpedia.org/property/spouse$>$ ?b.
?a $<$http://dbpedia.org/property/wikiLinks$>$  $<$http://dbpedia.org/property/actor$>$.
?b $<$http://dbpedia.org/property/wikiLinks$>$  $<$http://dbpedia.org/property/actor$>$.
?a $<$http://dbpedia.org/property/placeofbirth$>$ ?c.
?b $<$http://dbpedia.org/property/placeofbirth$>$ ?c.
?c $<$http://www.w3.org/2002/07/owl$\#$sameAs$>$ ?c2.
?c2 $<$http://www.w3.org/2003/01/geo/wgs84$\_$pos$\#$lat$>$ ?lat.
?c2 $<$ http://www.w3.org /2003/01/geo/wgs84 $\_$pos$\#$long$>$ ?long.$\}$\\
\hline
\end{tabular}
\label{table:BTCQueries}
\end{table}

\begin{table*}
\vspace{-0.1in}
\scriptsize
\centering
\begin{threeparttable}
  \begin{tabular}{|c|c|p{0.4cm}|p{0.08cm}|r|r|r|r|r|r|r|r|r|r|r|}
  \hline
   &  &Query &  & \multicolumn{3}{c|}{Partial Evaluation} & \multicolumn{3}{c|}{Assembly} &  \multicolumn{3}{c|}{Total} & \tabincell{c}{$\#$ of }  &  \tabincell{c}{$\#$ of}\\
  \cline{5-13}
  &  & & &  & & & \multicolumn{2}{c|}{Time(in ms)}& & \multicolumn{2}{c|}{Time(in ms)}& & LPMFs &CMFs\\
  \cline{8-9}\cline{11-12}
  & & & &  Time(in ms) &  $\#$ of LPMs\tnote{$2$}& $\#$ of IMs\tnote{$3$} & \tabincell{c}{Centralized}	& \tabincell{c}{Distributed} & $\#$ of CMs\tnote{$4$}& PECA\tnote{$5$}	& PEDA\tnote{$6$}& $\#$ of Matches\tnote{$7$}& & \\
 \hline
  \multirow{4}{*}{Edge}& \multirow{4}{*}{\includegraphics[scale=0.1]{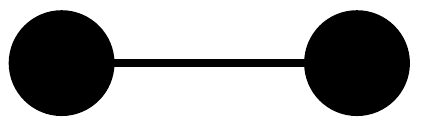}}  &  Q5& $\surd$\tnote{$1$} & 191	& 0 & 678 & 0	& 0	& 0& 191	& 191  & 678 & 0 & 0 \\
  \cline{3-15}
  & & Q6 &  & 13648	& 0 & 7924765 & 0	& 0	& 0& 13648	& 13648 & 7924765&  0 & 0 \\
  \cline{3-15}
  & & Q13 & $\surd$ & 231	& 0 & 3295 & 0	& 0	& 0& 231	& 231 & 3295&  0 & 0\\
  \cline{3-15}
  & & Q14 &  & 13646	& 0 & 7924765 & 0	& 0	& 0& 13646	& 13646 & 7924765&  0 & 0 \\
  \cline{1-15}
  \multirow{5}{*}{Star}& \multirow{5}{*}{\includegraphics[scale=0.1]{query_style_star.pdf}}&  Q1& $\surd$ & 191	& 0 & 1081187 & 0	& 0	& 0& 191	& 191 & 1081187&  0 & 0 \\
  \cline{3-15}
  & & Q3 & $\surd$ & 120	& 0 & 4 & 0	& 0	& 0 & 120	& 120 & 4&  0 & 0\\
  \cline{3-15}
  & & Q4 & $\surd$ & 85	& 0 & 10 & 0	& 0	& 0 & 85	& 85 & 10&  0 & 0 \\
  \cline{3-15}
  & & Q10 & $\surd$ & 190	& 0 & 4 & 0	& 0	& 0 & 190	& 190 & 4&  0 & 0\\
  \cline{3-15}
  & & Q12 & $\surd$ & 127	& 0 & 540 & 0	& 0	& 0 & 127	& 127 & 540&  0 & 0 \\
  \cline{1-15}
  Linear& \multirow{1}{*}{\includegraphics[scale=0.1]{query_style_path.pdf}} & Q11 & $\surd$ & 42	& 0 & 4 & 0	& 0	& 0 & 42	& 42 & 4&  0 & 0 \\
  \cline{1-15}
  \multirow{2}{*}{Snowflake}& \multirow{2}{*}{\includegraphics[scale=0.1]{query_style_snowflake.pdf}} &  Q7& $\surd$ & 1660	& 60172 & 58 & 1338	& 1415	& 1 & 2998	& 3075  & 59& 10 & 10 \\
  \cline{3-15}
  & & Q8 & $\surd$ & 1132	& 20050 & 5901 & 654	& 728	& 15 & 1786	& 1860  & 5916&  10 & 10 \\
  \cline{1-15}
  \multirow{2}{*}{Complex}& \multirow{2}{*}{\includegraphics[scale=0.1]{query_style_complex.pdf}}  &Q2 &  & 52712	& 18035 & 2506 & 639	& 439	& 22 & 53351	& 53151 & 2528& 10 & 10 \\
  \cline{3-15}
  & & Q9 &  & 3328	&12049 & 44086 & 763	& 362	& 104 & 4091	& 3690 & 44190& 10 & 10 \\
  \hline
  \end{tabular}
  \end{threeparttable}
  \caption{Evaluation of Each Stage for 14 Benchmark Queries over LUBM 1000}
  \label{table:LUBM100014QueriesEvaluation}
\end{table*}

%--------------------------------------------------------------------------------------------------------

\vspace{-0.1in}
\section{Results of 14 Benchmark Queries over LUBM 1000}\vspace{-0.1in}
\label{sec:allqueries}
Table \ref{table:LUBM100014QueriesEvaluation} shows the experimental results of each stage for original 14 LUBM benchmark queries listed in Table \ref{table:LUBM14Queries}. Note that, since the current version of gStore does not currently support type reasoning, we revised the original 14 to remove type reasoning. The resulting queries return larger result sets since there is no filtering as a result of type reasoning.

Generally speaking, many of these queries are simple pattern queries or they are quite similar to each other, and the 7 queries chosen by \cite{DBLP:conf/www/AtreCZH10,Zou:2013fk} are representative queries out of this list. Hence, the results of the 14 benchmark queries are also similar to the results of our 7 representative queries in Table \ref{table:queriesperformanceLUBM}.

%\vspace{-0.2in}
%--------------------------------------------------------------------------------------------------------

\section{Exp 4 -- Expanded Results on MapReduce Effect}
\label{sec:exp4cont}
Existing partition-based solutions, such as GraphPartition and TripleGroup, use MapReduce jobs to join intermediate results to find SPARQL matches. In order to evaluate the cost of MapReduce-jobs, we perform the following experiments over WatDiv 100M. We revise join processing in both GraphPartition and TripleGroup, by applying joins where intermediate results are sent to a central server using MPI.  We use WatDiv 100M and only consider the benchmark queries that need join processing ($L_1$, $L_2$, $L_5$, $F_1$, $F_2$, $F_3$, $F_3$, $F_4$, $F_5$, $C_1$ and $C_2$) in our experiments. Moreover, all partition-based methods generate intermediate results and merge them at a central sever that shares the same framework with PECA, so we only compare them with PECA.

Tables \ref{table:LocalJoinCostGraphPartition} and \ref{table:LocalJoinCostTG} show the performance of the three approaches. Our technique is always faster regardless of the use of MPI or MapReduce-based join. This is because  our method produces  smaller intermediate result sets; MapReduce-based join dominates the query cost.

Tables \ref{table:LocalJoinCostGraphPartition} and \ref{table:LocalJoinCostTG} demonstrate that our partial evaluation process is more expensive in evaluating local queries than GraphPartition and TripleGroup in many cases. This is easy to understand -- since the subquery structures in GraphPartition and TripleGroup are fixed, such as stars, it is  cheaper to find these local query results than finding local partial matches.

Our system generally outperforms GraphPartition and TripleGroup significantly if they use MapReduce-based join. Even when GraphPartition and TripleGroup use distributed joins, our system is still faster than them in most cases (8 out of 10 queries in this experiment).

\begin{table*}
\scriptsize
\centering
  \begin{tabular}{|r|r|r|r|r|r|r|r|r|}
  \hline
  & \multicolumn{3}{c|}{PECA}	& \multicolumn{3}{c|}{GraphPartition} & \multicolumn{2}{c|}{MPI-revised GraphPartition} \\
  \hline
  & \tabincell{c}{Partial Evaluation}	& \tabincell{c}{Assembly} & \tabincell{c}{Total Time} & \tabincell{c}{Finding Partial Matches} &\tabincell{c}{MapReduce-based\\Join} & \tabincell{c}{MapReduce-based\\Total Time}	& \tabincell{c}{MPI-based\\Join} & \tabincell{c}{MPI-based\\Total Time}\\
  \hline
  $L_1$ & 2350	& 1 & 2351	& 1423  & 19570	& 20993 & 183	& \textbf{1606}\\
  \hline
     $L_2$& 557	& 1 & \textbf{558}	& 386 & 16420	& 16806   & 204	& 590\\
  \hline
    $L_5$ & 524	& 2 & \textbf{526}	& 479 & 27480	& 27959  & 76	& 555\\
    \hline
    $F_1$ &  3906	& 1 & \textbf{3907}	& 4011 & 36200	& 40211  & 35	& 4046\\
  \hline
     $F_2$& 2659	& 31 & \textbf{2690}	& 2466 & 58180	& 60646  & 1277	& 3743\\
  \hline
    $F_3$ & 16077	& 1945 & \textbf{18022}	& 14136 & 61400	& 75536  & 4191	& 18327\\
  \hline
     $F_4$& 21446	& 47 & 21493	& 15535 & 34060	& 49595  & 165	& \textbf{15700}\\
  \hline
    $F_5$ & 9043	& 43 & \textbf{9086}	& 9910 & 51110	& 61020  & 1900	& 11810\\
  \hline
     $C_1$& 12969	& 52 & \textbf{13021}	& 9799 & 223670	& 233469  & 18522	& 28321\\
  \hline
    $C_2$ & 37850	& 1454 & \textbf{39304}	& 44998 & 2168300	& 2213298  & 19494	& 64492\\
  \hline
  \end{tabular}
  \caption{Query Response Time over Partitioning Strategy of GraphPartition (in milliseconds)}
  \label{table:LocalJoinCostGraphPartition}
\end{table*}

\begin{table*}
\scriptsize
\centering
  \begin{tabular}{|r|r|r|r|r|r|r|r|r|}
  \hline
  & \multicolumn{3}{c|}{PECA}	& \multicolumn{3}{c|}{TripleGroup}& \multicolumn{2}{c|}{MPI-revised TripleGroup} \\
  \hline
  & \tabincell{c}{Partial Evaluation}	& \tabincell{c}{Assembly} & \tabincell{c}{Total Time} & \tabincell{c}{Finding Partial Matches} & \tabincell{c}{MapReduce-based\\Join} & \tabincell{c}{MapReduce-based\\Total Time}	& \tabincell{c}{MPI-based\\Join} & \tabincell{c}{MPI-based\\Total Time}\\
  \hline
  $L_1$ & 2250      & 1 & 2251	& 1122  & 20840	& 21962 & 452	& \textbf{1574}\\
  \hline
     $L_2$& 249     & 1 & \textbf{250}     & 204  & 16150	& 16354 & 50	& 254\\
  \hline
    $L_5$ & 737    & 2 & 739    & 304 & 27550	& 27854 & 70	& \textbf{374}\\
    \hline
    $F_1$ & 5753    & 1 & \textbf{5753}     & 4413 & 36230	& 40643 & 1538	& 5951\\
  \hline
     $F_2$& 4771	& 21 & \textbf{4792}	& 3909 & 40700	& 44609 & 911	& 4820\\
  \hline
    $F_3$ & 10425	& 3174 & \textbf{12599}	& 10517 & 62440	& 72957 & 5346	& 15863\\
  \hline
     $F_4$& 16373	& 66 & \textbf{16439}	& 15403 & 54054	& 69457  & 1212	& 16615\\
  \hline
    $F_5$ & 11611	& 22 & \textbf{11633}	& 13039 & 22180	& 35219 & 4923	& 17962\\
  \hline
     $C_1$& 12794	& 2265 & \textbf{15059}	& 6057 & 216720	& 222777 & 12194	& 18251\\
  \hline
    $C_2$ & 44272	& 8870 & \textbf{53142}	&48204  & 1954800	& 2003004& 15062	& 63266\\
  \hline
  \end{tabular}
  \caption{Query Response Time over Partitioning Strategy of TripleGroup (in milliseconds)}
  \label{table:LocalJoinCostTG}
\end{table*}

%--------------------------------------------------------------------------------------------------------
\section{Comparisons with Memory Systems}
\label{sec:memory-exp}
The comparison results between our system with two typical distributed memory systems are given in Table \ref{table:ComparisonMemorySystems}.
\begin{table}
\scriptsize
\centering
  \begin{tabular}{|c|r|r|r|r|r|}
  \hline
   &\tabincell{c}{Trinity.RDF} & \tabincell{c}{TriAD}&	\tabincell{c}{PECA}& \tabincell{c}{PEDA}\\
  \hline
   $Q_1$ & 281	&	97& 53  & 51\\
  \hline
  $Q_2$ & 132	&	140& 73  & 73\\
  \hline
  $Q_3$ & 110	&	31& 27  & 28\\
  \hline
  $Q_4$ & 5	&	1& 1  & 1\\
  \hline
  $Q_5$ & 4	&	0.2& 0.18  & 0.18\\
  \hline
  $Q_6$ & 9	&	1.8& 1.2  & 1.7\\
  \hline
  $Q_7$ & 630 &	711& 123  & 122\\
  \hline
  \end{tabular}
  \caption{ Comparison with Memory-based Distributed RDF Systems in LUBM 1000}
  \label{table:ComparisonMemorySystems}
\end{table}

%--------------------------------------------------------------------------------------------------------
\section{Comparisons with Federated Systems}
\label{sec:feder-exp}

Comparisons with federated systems is done using the FedBench benchmark \cite{DBLP:FedBench}. The partitioning of this benchmark is pre-defined and fixed. FedBench includes 6 real cross domain RDF datasets and 4 real life science domain RDF datasets, and each dataset maps to one fragment. In this benchmark, 7 federated queries are defined for cross domain RDF datasets, and 7 federated queries are defined for life science RDF datasets. The characteristics of FedEx dataset are given in Table \ref{table:fedex_data}.

\begin{table}
\scriptsize
\centering
\vspace{-0.1in}
\caption{Datasets}
\begin{tabular}{|c|c|r|r|r|}
\hline
\multicolumn{2}{|c|}{Dataset}& \tabincell{c}{Number of\\Triples}& \tabincell{c}{RDF N3 File\\Size(KB)}  & \tabincell{c}{Number of\\Entities}\\
  \hline
  \hline
   \multirow{6}{*}{\tabincell{c}{FedBench\\(Cross Domain)}} & DBPedia subset &42,855,253 & 6,267,080 & 8,027,158\\
  \cline{2-5}
  & NY Times & 337,563 & 103,788 & 21,667\\
  \cline{2-5}
  & LinkedMDB& 6,147,997 & 1,745,790 & 665,441\\
  \cline{2-5}
  & Jamendo& 1,049,647 & 147,280 & 290,292\\
  \cline{2-5}
  & GeoNames&  107,950,085 & 12,112,090 & 7,479,715\\
  \cline{2-5}
  & SW Dog Food& 103,595 & 16,858 &10,459 \\
  \hline
  \multirow{4}{*}{\tabincell{c}{FedBench\\(Life Science)}} & DBPedia subset &42,855,253 & 6,267,080 & 8,027,158\\
  \cline{2-5}
  & KEGG  &  1,090,830 & 120,115 & 34,261\\
  \cline{2-5}
  & Drugbank & 766,920 & 146,906 & 19,694\\
  \cline{2-5}
  & ChEBI &  7,325,744 & 847,936 & 50,478\\
  \hline
  \end{tabular}
  \vspace{-0.1in}
\label{table:fedex_data}
\end{table}

FedBench also provides 7 benchmark queries for Cross Domain and 7 queries for Life Science. However, these queries are highly homogeneous. They are all snowflake queries (several stars linked by a path) and they all contain selective triple patterns \footnote{As defined before, a triple pattern is ``selective'' if it has no more than 100 matches in RDF graph $G$.}. To test the performance against different query structures, we introduce  five other queries for each domain. Two of them ($ECD_1$ and $ECD_2$ for Cross Domain and $ELS_1$ and $ELS_2$ for Life Science) are stars; one ($ECD_{3}$ for Cross Domain and $ELS_{3}$ for Life Science) is snowflake, and two ($ECD_{4}$ and $ECD_{5}$ for Cross Domain and $ELS_{4}$ and $ELS_{5}$ for Life Science) are complex queries. Furthermore, some queries ($ECD_1$, $ECD_{4}$, $ELS_1$ and $ELS_{4}$) contain selective triple patterns, while the others ($ECD_2$, $ECD_{3}$, $ECD_{5}$, $ELS_2$, $ELS_{3}$ and $ELS_{5}$) do not. We report the experimental results in the Figure \ref{fig:fedbenchcomparison}.

We compare our approach against two systems: FedX \cite{DBLP:FedX} and SPLENDID \cite{DBLP:SPLENDID}. Our approach has a superior performance in most cases, especially when the queries do not contain selective triple patterns. FedX and SPLENDID decompose a SPARQL query into a set of subqueries, and join the intermediate results of all subqueries together to find the final results. When the intermediate results of two subqueries join together, FedX employs the bound join and SPLENDID uses hash join. This means that these systems first use the intermediate results to rewrite the subquery with bound join variables. Then,  the rewritten queries are evaluated at the corresponding sites. Therefore, when the SPARQL queries do not contain any selective triple pattern, the size of intermediate results is so large that evaluation of bound joins takes significant time. In this case, our system outperforms them by orders of magnitude. For example, in $ECD_2$, $ECD_{3}$, $ECD_{5}$, $ELS_2$, $ELS_{3}$ and $ELS_{5}$, our approach is faster than FedX and SPLENDID by an order of magnitude.

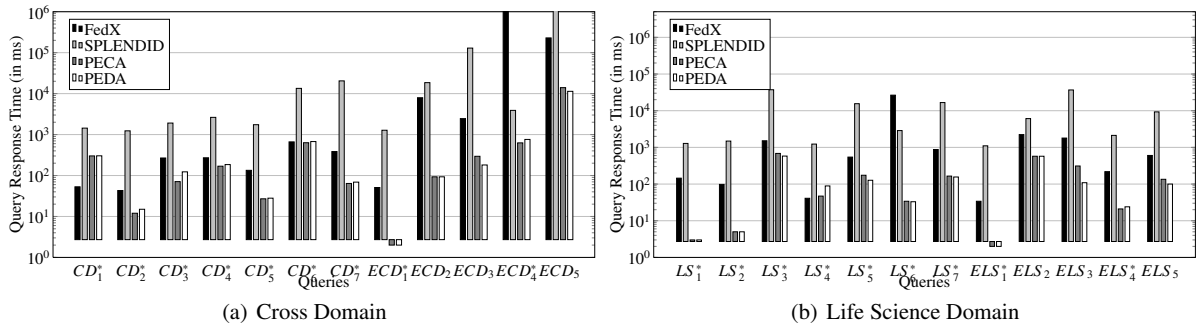
\begin{figure*}%
   \subfigure[Cross Domain]{%
		\resizebox{0.45\columnwidth}{!}{
			\begin{tikzpicture}[font=\LARGE]
    \begin{semilogyaxis}[
                anchor={(0,100)},
               width = 20cm,
               height = 10cm,
    			major x tick style = transparent,
    			ybar,
    			ymax = 1000000,
    ymin=1,
    enlarge x limits=0.075,
   			ymajorgrids = true,
   			ylabel = {Query Response Time (in ms)},
    			xlabel = {Queries},
    			symbolic x coords = {$CD_1^*$,$CD_2^*$,$CD_3^*$,$CD_4^*$,$CD_5^*$,$CD_6^*$,$CD_7^*$,$ECD_1^*$,$ECD_2$,$ECD_{3}$,$ECD_{4}^*$,$ECD_{5}$},
    			scaled y ticks = true,
			bar width=5pt,
			legend pos= north west,
 legend cell align=left
   		]

    \addplot [fill=black!100] coordinates { ($CD_1^*$, 53)       ($CD_2^*$, 43)       ($CD_3^*$, 268)  ($CD_4^*$,271)     ($CD_5^*$,134)    ($CD_6^*$,663 )         ($CD_7^*$,385)    ($ECD_1^*$, 51)       ($ECD_2$, 7935)  ($ECD_{3}$,2452)     ($ECD_{4}^*$,3842619)   ($ECD_{5}$,229369)};

    \addplot [fill=black!25]  coordinates { ($CD_1^*$, 1437)      ($CD_2^*$, 1234)    ($CD_3^*$, 1911)  ($CD_4^*$,2633)   ($CD_5^*$,1743)   ($CD_6^*$,13439)    ($CD_7^*$,20475)  ($ECD_1^*$,1276)    ($ECD_2$,18542)    ($ECD_{3}$,129405)    ($ECD_{4}^*$,3899)          ($ECD_{5}$,6113373)};

    \addplot [fill=black!50] coordinates {  ($CD_1^*$, 301)         ($CD_2^*$, 12)      ($CD_3^*$,  71)  ($CD_4^*$,169)    ($CD_5^*$, 27)   ($CD_6^*$,632)
    ($CD_7^*$,64)   ($ECD_1^*$,2)    ($ECD_2$,93)    ($ECD_{3}$,295 )    ($ECD_{4}^*$,627 )              ($ECD_{5}$,13996)};

    \addplot [fill=white] coordinates {     ($CD_1^*$,303)    ($CD_2^*$, 15)            ($CD_3^*$,  123) ($CD_4^*$, 185)     ($CD_5^*$, 28)    ($CD_6^*$,  675)
    ($CD_7^*$,69)   ($ECD_1^*$,2)    ($ECD_2$,93)    ($ECD_{3}$,181)    ($ECD_{4}^*$,761)              ($ECD_{5}$,11377)};

        \legend{FedX,SPLENDID, PECA,PEDA}

    \end{semilogyaxis}
\end{tikzpicture}
		}
       \label{fig:cdcomparison}%
       }%
\subfigure[Life Science Domain]{%
		\resizebox{0.45\columnwidth}{!}{
			\begin{tikzpicture}[font=\LARGE]
    \begin{semilogyaxis}[
                anchor={(0,100)},
               width = 20cm,
               height = 10cm,
    			major x tick style = transparent,
    			ybar,
    			ymax = 5000000,
    ymin=1,
    enlarge x limits=0.075,
   			ymajorgrids = true,
   			ylabel = {Query Response Time (in ms)},
    			xlabel = {Queries},
    			symbolic x coords = {$LS_1^*$,$LS_2^*$,$LS_3^*$,$LS_4^*$,$LS_5^*$,$LS_6^*$,$LS_7^*$,$ELS_1^*$,$ELS_2$,$ELS_{3}$,$ELS_4^*$,$ELS_{5}$},
    			scaled y ticks = true,
			bar width=4.5pt,
			legend pos= north west,
 legend cell align=left
    ]

    \addplot [fill=black!100] coordinates { ($LS_1^*$, 145)       ($LS_2^*$, 98)      ($LS_3^*$, 1515)       ($LS_4^*$,41)      ($LS_5^*$,544)        ($LS_6^*$,26560 )  ($LS_7^*$,875)   ($ELS_1^*$, 34)       ($ELS_2$, 2230)      ($ELS_{3}$, 1796)       ($ELS_4^*$,218)      ($ELS_{5}$,603)};

    \addplot [fill=black!25]  coordinates { ($LS_1^*$, 1283)       ($LS_2^*$, 1487)      ($LS_3^*$, 37311)    ($LS_4^*$,1231)   ($LS_5^*$,15546)        ($LS_6^*$,2875)  ($LS_7^*$,16686)   ($ELS_1^*$, 1104)       ($ELS_2$, 6098)      ($ELS_{3}$, 36812)       ($ELS_4^*$,2126)      ($ELS_{5}$,9295)};

    \addplot [fill=black!50] coordinates {  ($LS_1^*$, 3)         ($LS_2^*$, 5)         ($LS_3^*$, 685)      ($LS_4^*$,47)       ($LS_5^*$, 174)      ($LS_6^*$,34)   ($LS_7^*$,165 )   ($ELS_1^*$, 2)       ($ELS_2$, 574)      ($ELS_{3}$, 311 )       ($ELS_4^*$,21)      ($ELS_{5}$,135)};

    \addplot [fill=white] coordinates {  ($LS_1^*$,3 )           ($LS_2^*$, 5 )         ($LS_3^*$, 580 )      ($LS_4^*$, 89)      ($LS_5^*$, 127 )       ($LS_6^*$,  33 )   ($LS_7^*$,156 )   ($ELS_1^*$, 2)       ($ELS_2$, 574)      ($ELS_{3}$, 109)       ($ELS_4^*$,24)      ($ELS_{5}$,100)};

        \legend{FedX,SPLENDID, PECA,PEDA}
    \end{semilogyaxis}

\end{tikzpicture}
		}
       \label{fig:lfdcomparison}%
       }%
 \caption{\small Online Performance Comparison with Federated RDF systems over FedBench($*$ means that the query involves some selective triple patterns.)}%
 \label{fig:fedbenchcomparison}
\end{figure*}

\end{document}